\documentclass[11pt]{article}
\usepackage{amsfonts}
\usepackage{amssymb}
\usepackage{amstext}
\usepackage{amsmath}
\usepackage{amssymb}
\usepackage{xspace}
\usepackage{theorem}
\usepackage{thmtools}
\usepackage{graphicx}
\usepackage{graphics}
\usepackage{subcaption}
\usepackage{float}
\usepackage{wrapfig}
\usepackage{colordvi}
\usepackage{color}
\usepackage{paralist}
\usepackage{bm}
\usepackage{mathpazo}

\usepackage[normalem]{ulem}

        {\hspace*{\fill}$\Box$\par\vspace{4mm}}

\newcommand{\ceil}[1]{\ensuremath{\left\lceil#1\right\rceil}}
\newcommand{\floor}[1]{\ensuremath{\left\lfloor#1\right\rfloor}}

\newcommand{\be}{\begin{enumerate}}
\newcommand{\ee}{\end{enumerate}}
\newcommand{\bd}{\begin{description}}
\newcommand{\ed}{\end{description}}
\newcommand{\bi}{\begin{itemize}}
\newcommand{\ei}{\end{itemize}}

\usepackage[ocgcolorlinks]{hyperref} 
\usepackage{cleveref}

\declaretheorem[numberwithin=section]{theorem}
\declaretheorem[numberlike=theorem]{lemma}
\declaretheorem[numberlike=theorem]{proposition}
\declaretheorem[numberlike=theorem]{corollary}

\declaretheorem[numberlike=theorem]{claim}
\declaretheorem[numberlike=theorem]{Observation}

\newenvironment{proof}{\par \smallskip{\bf Proof:}}{\hfill\stopproof}
\def\stopproof{\square}
\def\square{\vbox{\hrule height.2pt\hbox{\vrule width.2pt height5pt \kern5pt
\vrule width.2pt} \hrule height.2pt}}

\newenvironment{prog}[1]{
\begin{minipage}{5.8 in}
{\sc\bf #1}
\begin{enumerate}}
{
\end{enumerate}
\end{minipage}
}

\renewcommand{\phi}{\varphi}

\newcommand{\N}{\ensuremath{\mathbb N}}

\newcommand{\strike}[1]{\textcolor{red}{{\ifmmode\text{\sout{\ensuremath{#1}}}\else\sout{#1}\fi}}}

\newcommand{\cset}{{\mathcal C}} 
 
\newcommand{\fset}{{\mathcal F}} 
\newcommand{\tset}{{\mathcal T}} 
\newcommand{\bdot}{\bullet}

\newcommand{\ie}{\textit{i.e. }}
\newcommand{\one}{\mathbf{1}}

\title{A Tight Extremal Bound on the Lov\'{a}sz Cactus Number in Planar Graphs}
\author{Parinya Chalermsook\thanks{Aalto University, Espoo, Finland. {\tt E-mail: parinya.chalermsook@aalto.fi}} \and Andreas Schmid\thanks{Max Planck Institute for Informatics, Saarbr\"ucken, Germany. {\tt E-mail: aschmid@mpi-inf.mpg.de}}  \and Sumedha Uniyal\thanks{Aalto University, Espoo, Finland. {\tt E-mail: sumedha.uniyal@aalto.fi}}}

\begin{document}

\begin{titlepage}
\clearpage\thispagestyle{empty}

\maketitle

\begin{abstract}
A {\em cactus} graph is a graph in which any two cycles are edge-disjoint.
We present a constructive proof of the fact that any plane graph $G$ contains a cactus subgraph $C$ where $C$ contains at least a $\frac{1}{6}$ fraction of the triangular faces of $G$. We also show that this ratio cannot be improved by showing a tight lower bound.
Together with an algorithm for linear matroid parity, our bound implies two approximation algorithms for computing ``dense planar structures'' inside any graph: (i) A $\frac{1}{6}$ approximation algorithm for, given any graph $G$, finding a planar subgraph with a maximum number of triangular faces; this improves upon the previous $\frac{1}{11}$-approximation; (ii) An alternate (and arguably more illustrative) proof of the $\frac{4}{9}$ approximation algorithm for finding a planar subgraph with a maximum number of edges.

Our bound is obtained by analyzing a natural local search strategy and heavily exploiting the exchange arguments. 
Therefore, this suggests the power of local search in handling problems of this kind\footnote{This work appeared in STACS19~\cite{CSU19-stacs}.}. 
\end{abstract}

\end{titlepage}

\clearpage

\section{Introduction} 
{\em Linear matroid parity} (introduced in various equivalent forms~\cite{lovasz1980matroid,lawler1976combinatorial,jenkyns1974matchoids}) is a key concept in combinatorial optimization that includes many important optimization problems as special cases; probably the most well-known example is the {\em maximum matching problem}. 
The polynomial-time computability of linear matroid parity made it a popular choice as an algorithmic tool for handling both theoretical and practical optimization problems. 
An important special case of linear matroid parity, the graphic matroid parity problem, is often explained in the language of {\em cacti}
(see e.g.~\cite{cheung2014algebraic}), a graph in which any two cycles must be edge-disjoint.
In 1980, Lov\'{a}sz~\cite{lovasz1980matroid} initiated the study of $\beta(G)$ (sometimes referred to as the {\em cactus number} of $G$), the maximum value of the number of triangles in a cactus subgraph of $G$, and showed that it generalizes maximum matching and can be reduced to linear matroid parity, therefore implying that $\beta(G)$ is polynomial-time computable\footnote{There are many efficient algorithms for matroid parity (both randomized and deterministic), e.g.~\cite{cheung2014algebraic,lovasz2009matching,orlin2008fast,gabow1986augmenting}}\footnote{When we study $\beta(G)$, notice that a cactus subgraph that achieves the maximum value of $\beta(G)$ would only need to have cycles of length three (triangles). Such cacti are called {\em triangular cacti}.}. 

Cactus graphs arise naturally in many applications\footnote{See for instance the wikipedia page~\url{https://en.wikipedia.org/wiki/Cactus_graph}}; perhaps the most relevant example in the context of approximation algorithms is the Maximum Planar Subgraph (MPS) problem: Given an input graph, find a planar subgraph with a maximum number of edges. 
Notice that, since any planar graph with $n$ vertices has at most $3n - 6$ edges, outputting a spanning tree with $n-1$ edges immediately gives a $\frac{1}{3}$-approximation algorithm.
Generalizing the idea of finding spanning trees, one would like to look for a planar graph $H$, denser than a spanning tree, and at the same time efficiently computable. 
Calinescu et al.~\cite{cualinescu1998better} showed that a cactus subgraph with a maximum number of triangles (which is efficiently computable via matroid parity algorithms) could be used to construct a $\frac{4}{9}$-approximation for MPS.


The $\frac{4}{9}$-approximation for MPS was achieved through an extremal bound of $\beta(G)$ when $G$ is a plane graph. 
In particular, it was proven that $\beta(G) \geq \frac{1}{3}(n - 2 - t(G))$, where $n = |V(G)|$ and $t(G) = (3n - 6) - |E(G)|$ (i.e. the number of edges missing for $G$ to be a triangulated plane graph).

\begin{figure}[h]
  \centering
    \includegraphics[width=0.4\textwidth]{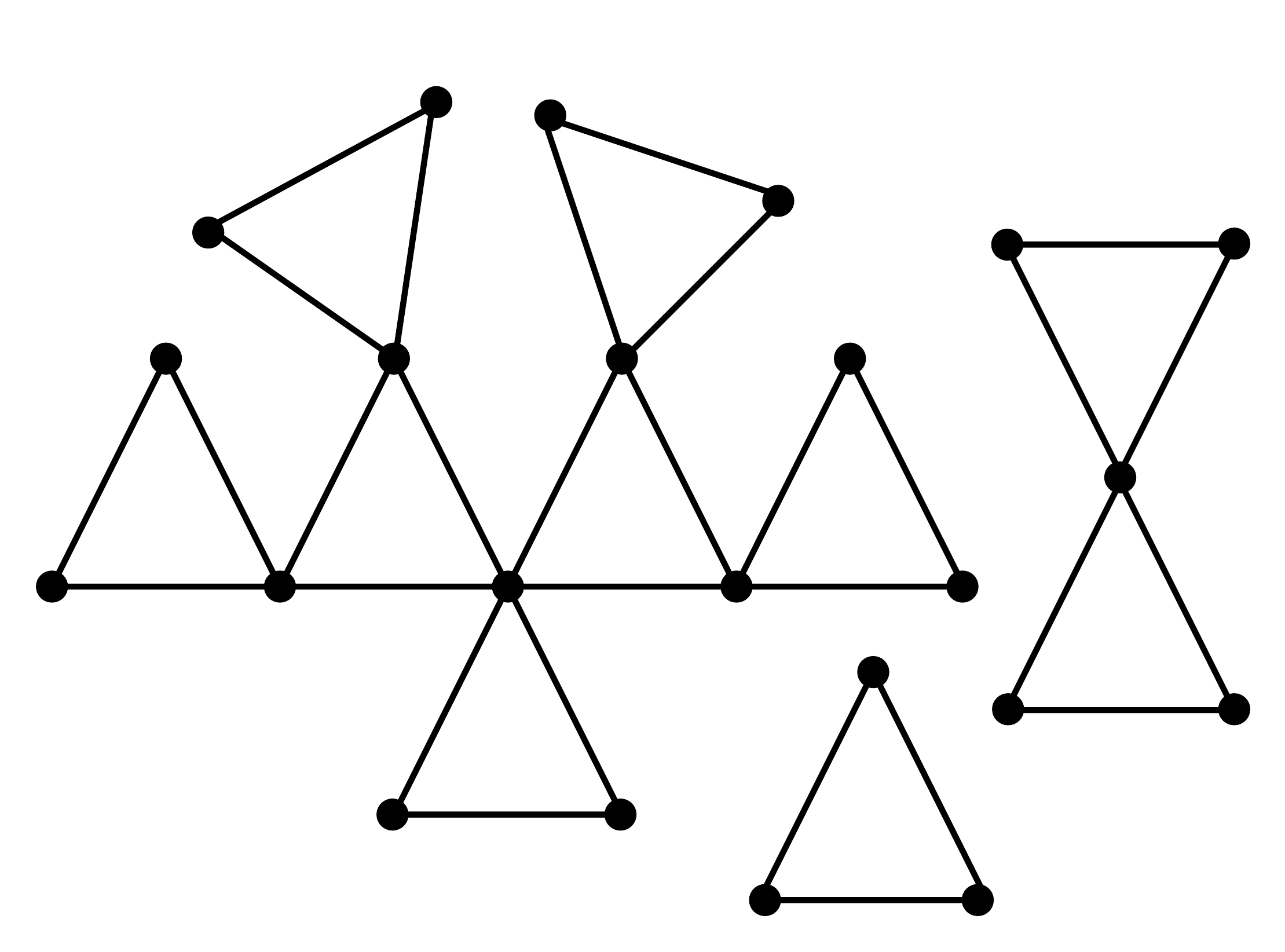}
    \caption{A triangular cactus graph.}
    \label{fig:cactus}
\end{figure}



\subsection{Our Results}

In this work, we are interested in further studying the extremal properties of $\beta(G)$ and exhibit stronger algorithmic implications.
Our main result is summarized in the following theorem. 

\begin{theorem} \label{thm:main}
Let $G$ be a plane graph. Then $\beta(G) \geq \frac{1}{6}f_3(G)$ where $f_3(G)$ denotes the number of triangular faces in $G$. 
Moreover, a natural local search $2$-swap algorithm
achieves this bound. 
\end{theorem}

It is not hard to see that $f_3(G) \geq 2n -4 -2t(G)$ where $t(G)$ denotes the number of edges {missing for $G$ to be a triangulated plane graph}. 
Therefore, we obtain the main result of~\cite{cualinescu1998better} immediately.  

\begin{corollary}
$\beta(G) \geq \frac{1}{3}(n-2-t(G))$. Hence, the matroid parity algorithm gives a $\frac{4}{9}$-approximation for MPS. 
\end{corollary}

Besides implying the MPS result, we exhibit further implications of our bound.  
Recently in~\cite{CS17}, 
the authors introduced {\em Maximum Planar Triangles (MPT)}, where the goal is to find a plane subgraph with a maximum number of triangular faces.
It was shown that an approximation algorithm for MPT naturally translates into one for MPS, where a $\frac{1}{6}$ approximate MPT solution could be turned into a $\frac{4}{9}$ approximate MPS solution. However, the authors only managed to show a $\frac{1}{11}$ approximation for MPT.

Although the only change from MPS to MPT lies in the objective of maximizing the number of triangular faces instead of edges, the MPT objective seems much harder to handle, for instance, the extremal bound provided in~\cite{cualinescu1998better} is not sufficient to derive any approximation algorithm for MPT.

Theorem \ref{thm:main} therefore implies the following result for MPT. 

\begin{corollary}
A matroid parity algorithm gives a $\frac{1}{6}$ approximation algorithm for MPT. 
\end{corollary}

Our conceptual contributions are the following: 
\begin{enumerate}
    \item Our result further highlights the extremal role of the cactus number in finding a dense planar structure, as illustrated by the fact that our bound on $\beta(G)$ is more ``robust'' to the change of objectives from MPS to MPT. 
    It allows us to reach the limit of approximation algorithms that matroid parity provides 
    for both MPS and MPT. 
    
    \item Our work implies that local search arguments alone are sufficient to ``almost'' reach the best known approximation results for both MPS and MPT in the following sense: Matroid parity admits a PTAS via local search~\cite{lee2013matroid,BlaserJP17}. 
    Therefore, combining this with our bound implies that local search arguments are sufficient to get us to a $\frac{4}{9} + \epsilon$ approximation for MPS and $\frac{1}{6} + \epsilon$ approximation for MPT. 
    Therefore, this suggests that local search might be a promising candidate for such problems.

    \item Finally, in some ways, our work can be seen as an effort to open up all the black boxes used in MPS algorithms with the hope of learning algorithmic insights that are crucial for making progress on this kind of problems. 
    In more detail, there are two main ``black boxes'' hidden in the MPS result: (i) The use of Lov\'{a}sz min-max cactus formula in deriving the bound $\beta(G) \geq \frac{1}{3}(n - 2 - t(G))$, and (ii) the use of a matroid parity algorithm as a blackbox in computing $\beta(G)$.  
    Our bound for $\beta(G)$ is now purely combinatorial (and even constructive) and manages to by-pass (i). 

\end{enumerate}

\paragraph{Open problems and future directions:} From approximation algorithms' perspectives, there is still a large gap of understanding on the approximability of MPS and MPT.
In particular, can we improve over a $\frac{4}{9}$ approximation for MPS? Can we improve the $\frac{1}{6}$-approximation for MPT? From~\cite{CS17}, improving $\frac 1 6$ for MPT would lead to improved MPS result as well. 
As discussed above, it would be interesting to further explore the power of local search in the context of MPT and MPS. 

In particular, we propose the following local search  and conjecture that it breaks $1/6$ approximation for MPT (therefore breaking $4/9$ for MPS): 

\begin{quote}
While it is possible to remove $t$ triangles and add $t+1$ disjoint triangles or diamonds\footnote{A diamond graph is $K_4$ with one edge removed}, do it.    
\end{quote}

We remark that our result gives us the first step towards the analysis: Combined with~\cite{lee2013matroid}, our result implies that the above algorithm (without diamonds) converges to a factor $1/6$ for MPT. 
Therefore, one may say that the only missing component now is to incorporate the analysis of diamonds.

\subparagraph*{Related work:} 
On the hardness of approximation side, MPS is known to be APX-hard~\cite{cualinescu1998better}, while MPT is only known to be NP-hard~\cite{CS17}. 
In combinatorial optimization, there are a number of problems closely related to MPS and MPT. 
For instance, finding a maximum series-parallel subgraph~\cite{cualinescu2012maximum} or a maximum outer-planar graph~\cite{cualinescu1998better}, as well as the weighted variant of these problems~\cite{calinescu2003new}; these are the problems whose objectives are to maximize the number of edges. 

Perhaps the most famous extremal bound in the context of cactus is the min-max formula of Lov\'{a}sz~\cite{lovasz1980matroid} and a follow-up formula that is more illustrative in the context of cactus~\cite{szigeti1998min}.  
All these formulas generalize the Tutte-Berge formula~\cite{berge1958theorie,tutte1947factorization} that has been used extensively both in research and curriculum.  

Another related set of problems has the objectives of maximizing the number of vertices, instead of edges. In particular, in the maximum induced planar subgraph (i.e. given graph $G$, one aims at finding a set of nodes $S \subseteq V(G)$ such that $G[S]$ is planar, while maximizing $|S|$.)
This variant has been studied under a more generic name, called {\em maximum subgraph with hereditary property}~\cite{lund1993approximation,lewis1980node,halldorsson2004approximations}.  
This variant is unfortunately much harder to approximate: $\tilde \Omega(|V(G)|)$\footnote{The term $\tilde \Omega$ hides asymptotically smaller factors.} hard to approximate~\cite{haastad1999clique,khot2006better}; in fact, the problems in this family do not even admit any FPT approximation algorithm~\cite{chalermsook2017gap}, assuming the {\em gap exponential time hypothesis} (Gap-ETH).

\subsection{Overview of Techniques} 
We give a high-level overview of our techniques. The description in this section assumes certain familiarity with how standard local search analysis is often done. 

Our algorithm works as follows. 
Let $G$ be an input plane graph, and let $\cset$ be a cactus subgraph of $G$ whose triangles correspond to triangular faces of $G$. The local search operation, $t$-swap, is done as follows: As long as there is a collection $X \subseteq \cset$ of $\ell: \ell \leq t$ edge-disjoint triangles and $Y$ such that $(\cset\setminus X) \cup Y$ contains more triangular faces of $G$ than $\cset$ and it remains a cactus, we perform such an improvement step. A cactus subgraph is called locally $t$-swap optimal, if it can not be improved by a $t$-swap operation.
Remark that the triangles chosen by our local search are only those which are triangular faces in the input graph $G$ (we assume that the drawing of $G$ is fixed.)

Our analysis is highly technical, although the basic idea is very simple and intuitive. 
We give a high-level overview of the analysis. 
We remark that this description is overly simplified, but it sufficiently captures the crux of our arguments.  
Let $\cset$ be the solution obtained by the local search $2$-swap algorithm. 
We argue that the number of triangles in $\cset$ is at least $f_3(G)/6$. We remark that the $2$-swap is required, as we are aware of a bad example $H$ for which the $1$-swap local search only achieves a bound of $(\frac{1}{7}+o(1))f_3(H)$. 
For simplicity, let us assume that $\cset$ has only one non-singleton component. 
Let $S \subseteq V(G)$ be the vertices in such a connected component. 

Let $t$ be a triangle in $\cset$. Notice that removing the three edges of $t$ from $\cset$ breaks the cactus into at most three components, say $\cset_1 \cup \cset_2 \cup \cset_3$ that are pairwise vertex-disjoint, i.e. sets $S_j = V(\cset_j)$ are pairwise vertex-disjoint. Recall at this point that we would like to upper bound the number of triangles in $G$ by six times $\Delta$, where $\Delta$ is the number of triangles in the cactus $\cset$.
Notice that $f_3(G)$ is comprised of $f_3(G[S_1]) + f_3(G[S_2]) + f_3(G[S_3]) + q'$, where $q'$ is the number of triangles in $G$ ``across'' the components $S_j$ (i.e. those triangles whose vertices intersect with at least two sets $S_i, S_j$, where $i \neq j$. 
Therefore, if we could somehow give a nice upper bound on $q'$, e.g. if $q'\leq 6$, then we could inductively use $f_3(G[S_j]) \leq 6 \Delta_j$ where $\Delta_j$ is the number of triangles in $\cset_j$, and that therefore 
\[f_3(G) \leq 6(\Delta_1+ \Delta_2 + \Delta_3) +6 \leq 6(\Delta-1) + 6 = 6 \Delta\] 
and we would be done. 
However, it is not possible to give a nice upper bound on $q'$ that holds in general for all situations. 
We observe that such a bound can be proven for some suitable choice of $t$: Roughly speaking, removing such a triangle $t$ from $\cset$ would create a small ``interaction'' between components $\cset_j$ (i.e. small $q'$).  
We say that such a triangle $t$ is a {\em light} triangle; otherwise, we say that it is {\em heavy}. 
Let $\cset'$ be the current cactus we are considering. As long as there is a light triangle left in $\cset'$, we would remove it (thus breaking $\cset'$ into $\cset'_1, \cset'_2, \cset'_3$) and inductively use the bound for each $\cset'_j$. 
Therefore, we have reduced the problem to that of analyzing the base case of a cactus in which all triangles are heavy. 
Handling the base case of the inductive proof is the main challenge of our result. 

We sketch here the two key ideas. Let $S = V(\cset)$. 
The first key idea is the way we exploit the locally optimal solution in certain parts of the graph $G[S]$. 
We want to point out; the fact that all triangles in $\cset$ are heavy is exploited crucially in this step. 
Recall that, each heavy triangle is such that its removal creates three components $\cset_1, \cset_2, \cset_3$ with many ``interactions'' (i.e. many triangles across components) between them. 
This large amount of interaction is the main reason why we could not use induction before. 
However, intuitively, these triangles across components could serve as candidates for making local improvements. 
So the fact that there are many interactions would become our advantage in the local search analysis. 

We briefly illustrate how we take advantage of heavy triangles. 
Let $\tset$ be the set of triangular faces in $G$ that are not contained in $\bigcup_i G[S_i]$, so each triangle in $\tset$ has vertices in at least two subsets $S_j, S_i$ where $j \neq i$. 
The local search argument would allow us to say that all triangles in $\tset$ have one vertex in $S_i$, one in $S_j$ and one outside of $S_1 \cup S_2 \cup S_3$. 
This idea is illustrated in Figure~\ref{fig:1-swap}. 

\begin{figure}[H]
  \centering
  \begin{subfigure}[b]{0.45\textwidth}
  \centering
    \includegraphics[width=\textwidth]{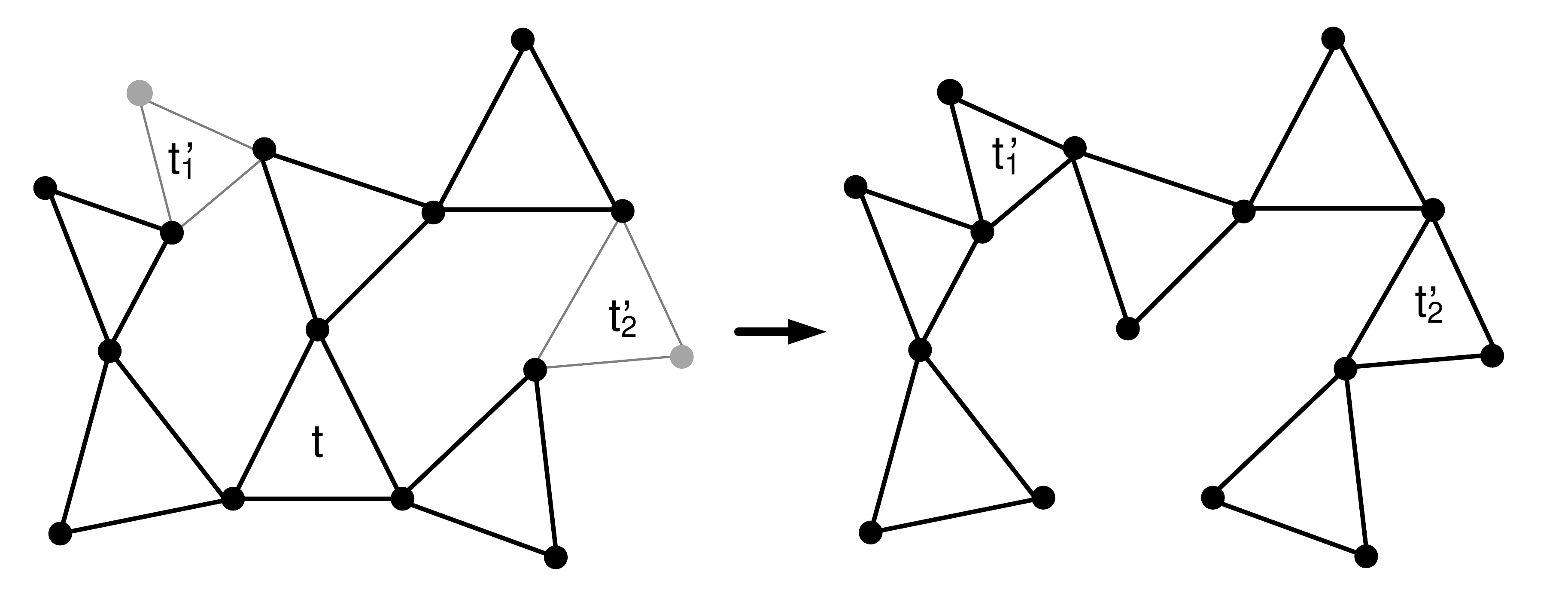}
    \caption{A $1$-swap operation. If there were two triangles $t_1'$, $t_2'$ in $\tset$ between two different pairs of components $S_j, S_i$ (where $j \neq i$), we could remove $t$ from $\cset$ and add $t_1'$, $t_2'$ to get a better cactus.}
    \label{fig:1-swap}
  \end{subfigure}\hspace{0.05\textwidth}%
      \begin{subfigure}[b]{0.5\textwidth} 
      \centering
    \includegraphics[width=\textwidth]{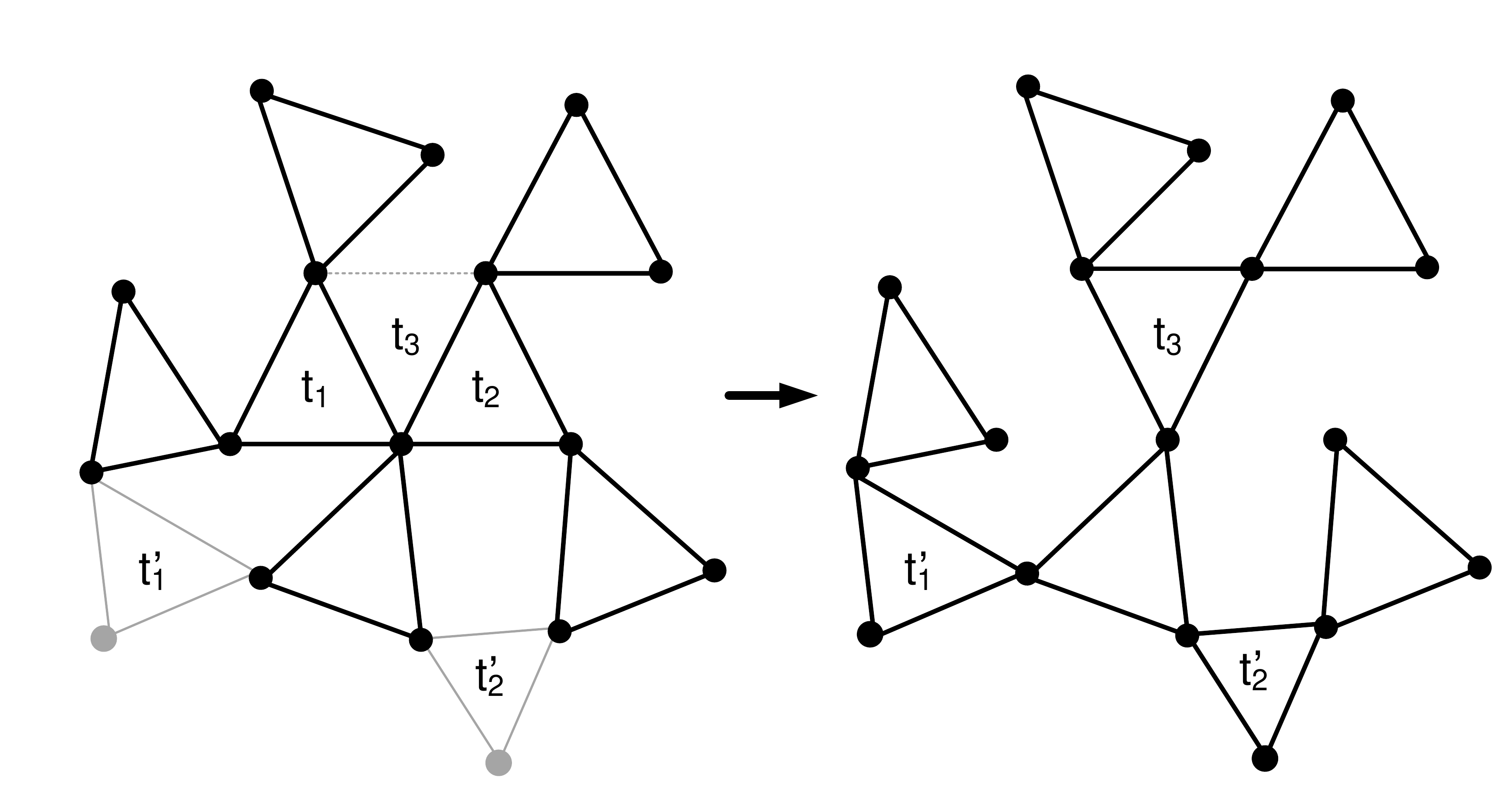}
    \caption{A $2$-swap operation. Let $t_1$ and $t_2$ be two adjacent triangles in our cactus. If there was an edge between $t_1$ and $t_2$, then there would exists a local improvement by removing $t_1$ and $t_2$ from $\cset$ and adding $t_1'$, $t_2'$ and $t_3$.}
    \label{fig:2-swap}
    \end{subfigure}
    \caption{Two examples for the swap operations.}
    \label{fig:swaps}
\end{figure}

Moreover, we will argue that there are not too many triangular faces in $G[S]$, and we give a rough idea of how the exchange argument can be used in Figure~\ref{fig:2-swap}. 

Finally, the ideas illustrated in both figures are only applied locally in a certain ``region'' inside the input planar graph $G$, so globally it is still unclear what would happen. 
Our final ingredient is a way to decompose the regions inside a plane graph into various ``atomic'' types. 
For each such atomic type, the local exchange argument is sufficient to argue optimally about the number of triangles in $G$ in that region compared to that in the cactus. 
Combining the bounds on these atomic types gives us the desired result. 
This is the most technically involved part of the paper, and we present it gradually by first showing the analysis that gives $\beta(G) \geq \frac{1}{7} f_3(G)$.  
For this, we need to classify the regions into five atomic types. 
To prove the main theorem, that $\beta(G) \geq \frac{1}{6} f_3(G)$, we need a more complicated classification into thirteen atomic types. 

\subparagraph*{Organization of the paper:} 
In Section~\ref{sec:main-overview}, we give a  detailed overview of the proof. In Section~\ref{thm:reduction-to-heavy}, we show the inductive argument, reducing the general case to proving the base case. 
In Section~\ref{sec:classification-7}, we show a slightly weaker version of the base case that implies $\beta(G) \geq \frac{1}{7} f_3(G)$, and in Section~\ref{sec:classification-6}, we prove the base case to get our main result. 

In Section~\ref{sec:extras}, we present how to construct a planar graph for which the bound proven in Theorem~\ref{thm:main} is tight. In addition we show how it implies the extremal bound provided in~\cite{cualinescu1998better}.
In Section~\ref{sec:conclusion}, we point out possible directions for future research and extensions of our work.

\section{Overview of the Proof} 
\label{sec:main-overview}

In this section, we give a formal overview of the structure of the proof of Theorem \ref{thm:main}. 
Let our input $G$ be a plane graph (a planar graph with a fixed drawing). 
Let $\cset$ be a locally optimal triangular cactus solution for the natural local search algorithm that uses $2$-swap operations, as described in the previous section. 
Let $\Delta(\cset)$ denote the number of triangular faces of $\cset$ which correspond to the triangular faces of $G$.
We will show $\Delta(\cset) \geq f_3(G)/6$. In general, we will use the function $\Delta: G \rightarrow \N$ to denote the number of triangular faces in any plane graph $G$.

We partition the vertices in $G$ into subsets based on the connected components of $\cset$, \ie $V(G) = \bigcup_i S_i$ where $\cset[S_i]$ is a connected cactus subgraph of $\cset$. For each $i$, where $|S_i| \geq 1$, let $q(S_i)$ denote the number of triangular faces in $G$ with at least two nodes in $S_i$. The following proposition holds by the $2$-swap optimality of $\cset$ which implies $f_3(G) = \sum_i q(S_i)$.

\begin{proposition}
\label{prop:reduction to 1 component} 
If $\Delta(\cset_i) \geq \frac{1}{6} q(S_i)$ for all $i$, then $\Delta(\cset) \geq \frac{1}{6} f_3(G)$. 
\end{proposition}

Therefore, it is sufficient to analyze any arbitrary component $S_i$ where $\cset[S_i]$ contains at least one triangle of $\cset$ (if the component does not contain such a triangle it is just a singleton vertex) and show that $\Delta(\cset_i) \geq \frac{1}{6} q(S_i)$.
Thus, from now on, we fix such an arbitrary component $S_i$ and denote $S_i$ simply by $S$, $q(S_i)$ by $q(S)$, and $\Delta(\cset[S_i])$ by $p$.
We will show that $q \leq 6p$ through several steps.

\subsection*{Step 1: Reduction to Heavy Cactus} 
In the first step, we will show that the general case can be reduced to the case where all triangles in $\cset$ are {\em heavy} (to be defined below).  
We refer to different types of vertices, edges and triangles in the graph $G$ as follows:
\begin{itemize}
    \item {\bf Cactus:} All edges/vertices/triangles in the cactus $\cset[S]$ are called {\em cactus edges/vertices/triangles} respectively.
\item {\bf Cross:} Edges with exactly one end-point in $S$ are called {\em cross edges}. Triangles that use one vertex outside of $S$ are {\em cross triangles}.
Notice that each cross triangle has exactly one edge in $G[S]$, that edge is called a {\em supporting edge} of the cross triangle. Similarly, we say that an edge $e \in E(G[S])$ supports a cross triangle; such a cross triangle $t$ contains exactly one vertex $v$ in some component $S_i \neq S$. The component $S_i$ is called the {\em landing component} of $t$. Similarly the vertex $v$ is called the {\em landing vertex} of $t$.

\item {\bf type-$[i]$ edges:} An edge in $G[S]$ that is not a cactus edge and does not support a cross triangle is called a {\em type-$[0]$ edge}. An edge in $G[S]$ that is not a cactus edge and supports $i$ cross triangle(s) is called a {\em type-$[i]$ edge}.
\end{itemize}
Therefore, each edge in $G[S]$ is a cactus, type-$0$, type-$1$ or type-$2$ edge. The introduced naming convention makes it easier to make important observations like the following (see Figure~\ref{fig:classification-edges-vertices-triangles} for an illustration of our naming convention).

\begin{figure}[H]
    \centering
    \includegraphics[width=0.7\textwidth]{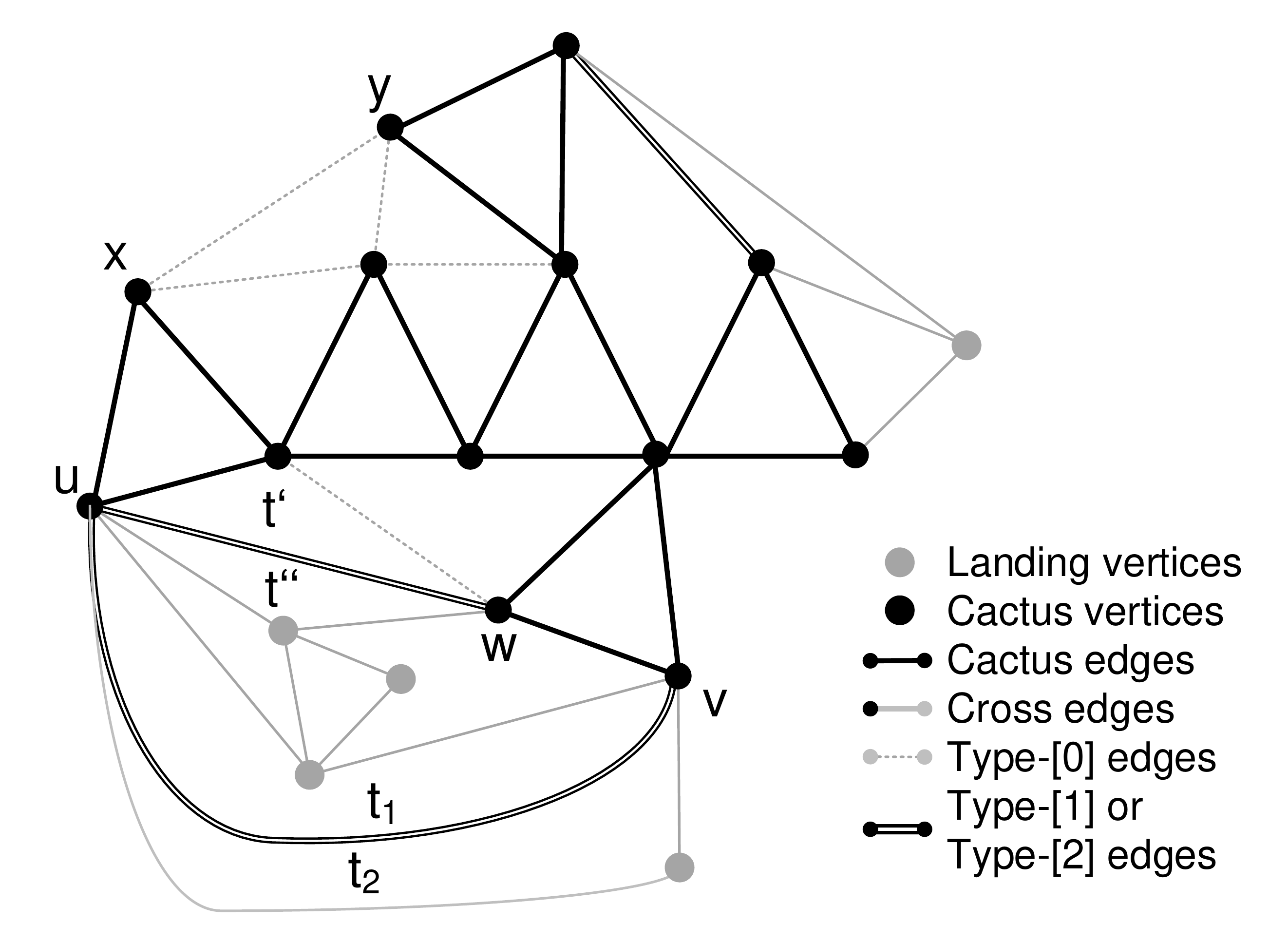}
    \caption{Various types of edges, vertices and triangles. Here the cross triangles $t''$ and $t_1$ have the same landing component.}
    \label{fig:classification-edges-vertices-triangles}
\end{figure}

\begin{Observation}
Triangles that contribute to the value of $q$ are of the following types: (i) the cactus triangles; (ii) the cross triangles; and (iii) the ``remaining'' triangles that connect three cactus vertices using at least one type-$0$, type-$1$ or type-$2$ edge, and do not have a cross triangle drawn inside.
\end{Observation}

\subparagraph*{Types of cactus triangles and Split cacti:}  
Consider a (cactus) triangle $t$ in $\cset$. For $i \in \{0,1,2,3\}$, we say that $t$
is of type-$i$ if exactly $i$ of its edges support a cross triangle. 
Let $p_i$ denote the number of type-$i$ cactus triangles, so we have that $p_0 + p_1 + p_2 + p_3 = p$. 

We denote the operation of deleting the edges of $t$ from a connected cactus $\cset[S]$ by {\em splitting} $\cset[S]$ at $t$. The resulting three smaller triangular cacti (denoted by $\{\cset^t_v\}_{v \in V(t)}$) are referred to as the {\em split cacti} of $t$.
For each $v \in V(t)$, let $S_v^t := V(\cset^t_v)$ be the {\em split component} containing $v$.
Let $u, v \in V(t): u \neq v$. Denote by $B_{uv}^t$ the set of type-$1$ or type-$2$ edges having one endpoint in $S_u^t$ and the other in $S^t_{v}$. Now we are ready to define the concept of heavy and light cactus triangles, which will be crucially used in our analysis. 

\subparagraph*{Heavy and light cactus triangles:}
We say that a cactus triangle $t$ is {\em heavy} if either there are at least four cross triangles supported by $E(t) \cup \bigcup_{uv \in E(t)} B_{uv}^t$ or there are at least three cross triangles supported by the edges in one set $B_{uv}^t \cup uv$ for some $uv \in E(t)$ and no cross triangle supported by the rest of the sets $B_{ww'}^t \cup ww'$for each $ww' \in E(t)$. Otherwise, the triangle is {\em light}.
Intuitively, the notion of a light cactus triangle $t$ captures the fact that, after removing $t$, there is only a small amount of ``interaction'' between the split components. 

We will abuse the notations a bit by using $S$ instead of $V[S]$.
Recall, that we denote by $q(S)$ the total number of triangular faces in $G$ with exactly two vertices in $S$. We denote by $p(S)$ the total number of triangles in the cactus $\cset[S]$.

\subparagraph*{Function $\phi$:} Consider a set $S \subseteq V(G)$ and a drawing of $G[S]$ (since we are talking about a fixed drawing of the plane graph $G$, this is well-defined). Denote by $\ell(S)$ the length of the outer-face $f_S$ of the graph $G[S]$. 
We define $\phi(S)$ as the number of edges on the outer-face that do not support any cross triangle drawn on the outer-face, so we have $0 \leq \phi(S) \leq \ell(S)$.

The main ingredients of Step 1 are encapsulated in the following theorem. 

\begin{theorem}[Reduction to heavy triangles] 
\label{thm:reduction-to-heavy} 
Let $\gamma \geq 6$ be a real number, and $\phi$ be as described above. 
If $q(S) \leq \gamma p(S) - \phi(S)$ for all $S$ for which $\cset[S]$ is a connected cactus that contains no light triangle, then $q(S) \leq \gamma p(S) - \phi(S)$ for all $S$. 
\end{theorem}

Therefore, if we manage to show the bound $q(S) \leq \gamma p(S) - \phi(S)$ for the heavy cactus, it will follow that $q \leq \gamma p$ {in general} (due to non-negativity of function $\phi$).
In other words, this gives a reduction from the general case to the case when all cactus triangles are heavy. 
We end the description of Step 1 by presenting the description of $\phi$. 

\subsection*{Step 2: Skeleton and Surviving Triangles}

Now, we focus on the case when there are only heavy triangles {in the given cactus}, and we will give a formal overview of the key idea we use to derive the bound $q(S) \leq 6p(S) - \phi(S)$, which in combination with Theorem~\ref{thm:reduction-to-heavy}, gives our main Theorem~\ref{thm:main}.
For convenience, we refer to the terms $p(S)$ and $q(S)$ as simply $p$ and $q$ respectively. 

\subparagraph*{Structures of heavy triangles:} Using local search's swap operations, the light and heavy triangles behave in a very well structured manner. 
The following proposition summarizes these structures for heavy triangles (proof of this proposition in Appendix~\ref{subsec:proof-prop-heavy}). 
\begin{proposition} 
\label{prop:structure-heavy} 
Let $t$ be a cactus triangle in cactus $\cset[S]$.
\begin{itemize}
    \item If $t$ is heavy, then $t$ is either type-$0$ or type-$1$. 
    
    \item If $t$ is a heavy type-$1$ triangle and the edge $uv \in E(t)$ supports the cross triangle supported by $t$, then $B^t_{ww'} = \emptyset$ for all $ww' \in E(t) \setminus \{uv\}$ and the total number of cross triangles supported by edges in $B^t_{uv}$ is greater than or equal to two.
    
    \item If $t$ is a heavy type-$0$ triangle, then there is an edge $uv \in E(t)$ such that $B^t_{ww'} = \emptyset$ for all $ww' \in E(t) \setminus \{uv\}$ and the total number of cross triangles supported by edges in $B^t_{uv}$ is greater than or equal to three.
\end{itemize}
\end{proposition}

By Proposition \ref{prop:structure-heavy} we can only have type-$0$ and type-$1$ cactus triangles in $\cset$.
Moreover, for each such heavy triangle $t$, the type-$1$ or type-$2$ edges in $G[S]$ only connect vertices of two split components of $t$.

Let $a_i$ be the number of edges of type-$i$. Notice that the number of non-cactus edges in $G[S]$ is $\sum_{i} a_i = |E(G[S])| - 3p$. 

\subparagraph*{Skeleton graph $H$:} Let $A$ be the set of all type-$0$ edges in $G[S]$ and $H := H[S] := G[S] \setminus A$.
Thus $H[S]$ contains only cactus or type-$1$ or type-$2$ edges. 

Each face $f$ of $H$ possibly contains several faces of $G$, so we will refer to such a face as a {\em super-face}.
At high-level, our plan is to analyze each super-face $f$, providing an upper bound on the number of triangular faces of $G$ drawn inside $f$, and then sum over all such $f$ to retrieve the final result. 
We call $H$ a {\em skeleton graph} of $G$, whose goal is to provide a decomposition of the faces of $G$ into structured super-faces.  
Denote by $\fset$ the set of all super-faces (except for the $p$ faces corresponding to cactus triangles). 

Let $f$ be a super-face. Denote by $survive(f)$ the number of triangular faces of $G$ drawn inside $f$ that do not contain any cross triangles.
Now we do a simple counting argument for $q$ using the skeleton $H$ as follows:
(i) There are $p$ cactus triangles in $H$, (ii) There are $p_1 + a_1 + 2a_2$ cross triangles supported by edges in $G[S]$, and (iii) There are $\sum_{f \in \fset} survive(f)$ triangular faces in $G$ that were not counted in (i) or (ii).
Combining this, we obtain: 
\begin{equation}
\label{eq:main-upper-bound-on-q} 
q \leq p + (p_1 + a_1 + 2a_2) + \sum_{f \in \fset} survive(f)
\end{equation} 
The first and second terms are expressed nicely as functions of $p$'s and $a$'s, so the key is to achieve the best upper bound on the third term in terms of the same parameters.
Roughly speaking, the intuition is the following: 
When $a_2$ or $a_1$ is high (there are many edges in $G[S]$ supporting cross triangles), the second term becomes higher. However, each cross triangle would need to be drawn inside some face in $G[S]$, therefore decreasing the value of the term $\sum_{f \in \fset} survive(f)$. Similar arguments can be made for $p_1$. 
Therefore, the key to a tight analysis is to understand this trade-off. 

\subparagraph*{The structure of super-faces:} Let $f \in \fset$ be a super-face. 
Recall that an edge in the boundary of $f$ is either a type-$1$ or type-$2$ edge, or a cactus edge. 
We aim for a better understanding of the value of $survive(f)$. 
In general, this value can be as high as $|E(f)| - 2$, e.g. if $G[V(f)]$ is a triangulation of the region bounded by the super-face $f$ using type-$0$ edges.
However, if some edge in the boundary of $f$ supports a cross triangle whose landing component is drawn inside of $f$ in $G$, this would decrease the value of $survive(f)$, by {\em killing} the triangular face adjacent to it, hence the term $survive$. 

The following observation is crucial in our analysis: 

\begin{Observation}
Consider each edge $e \in E(f)$. There are two possible cases: 
\begin{itemize}
    \item Edge $e$ is a type-$1$ or type-$2$ or cactus edge and supports a cross triangle drawn in $f$.
    \item Edge $e$ is a type-$1$ or type-$2$ or cactus edge and does not support any cross triangle drawn in $f$.
\end{itemize}
Edges lying in the first case are called {\em occupied} edges (the set of such edges in $E(f)$ is denoted by $Occ(f)$), while the others are called {\em free} edges in $f$ (the set of free edges in $E(f)$ is denoted by $Free(f)$).
The length of $f$ can be written as $|E(f)| = |Occ(f)| + |Free(f)|$.
\end{Observation}
A very important quantity for our analysis is $\mu(f) = \frac{1}{2} \cdot |Occ(f)| + |Free(f)|$, roughly bounding the value of $survive(f)$ (within some small constant additives terms.) 

We will assume without loss of generality that $survive(f)$ is the maximum possible value of surviving triangles that can be obtained by drawing type-$0$ edges in $f$, so $\mu(f)$ is a function that depends only on the bounding edges in $f$.  
We define $gain(f) = \mu(f) - survive(f)$, which is again a function that only depends on bounding edges of $f$. 
Intuitively, the higher the term $gain(f)$, the better for us (since this would lower the value of $survive(f)$), and in fact, it will later become clear that $gain(f)$ roughly captures the ``effectiveness'' of a local exchange argument on the super-face $f$. 
Hence, it suffices to show that $\sum_{f \in \fset} gain(f)$ is sufficiently large. The following proposition makes this precise: 

\begin{proposition}
$\sum_{f \in \fset} survive(f) = (3p - 0.5 p_1 +1.5 a_1 + a_2) - \sum_{f \in \fset} gain(f) $
\end{proposition}

\begin{proof}
Notice that $\sum_{f \in \fset} \mu(f)$ can be analyzed as follows: 
\begin{itemize}
    \item Each cactus triangle is counted three times (once for each of its edges), and for a type-$1$ triangle, one of the three edges contribute only one half. Therefore, this accounts for the term $3p - 0.5 p_1$. 
    
    \item Each type-$1$ or type-$2$ edge is counted two times (once per super-face containing it in its boundary). 
    For a type-$2$ edge, the contribution is always half (since it always is accounted in $Occ(f)$). For a type-$1$ edge, the contribution is half on the occupied case, and full on the free case. 
    Therefore, this accounts for the term $1.5 a_1 + a_2$. 
\end{itemize}
Overall we get, $\sum_{f \in \fset} \mu(f) = 3p - 0.5 p_1+ 1.5a_1 + a_2$, which finishes the proof.
\end{proof}
 
Combining this proposition with Equation~\ref{eq:main-upper-bound-on-q}, we get: 
\begin{equation}
\label{eq:gain-highlighted} 
q \leq 4p + 0.5p_1 + 2.5a_1 + 3 a_2 - \sum_{f \in \fset} gain(f)     
\end{equation}

\subparagraph*{A warm-up: Using the gains to prove a weaker bound:} To recap, after Step 1 and Step 2, we have reduced the analysis to the question of lower bounding $\sum_{f \in \fset} gain(f)$. 
We first illustrate that we could get a weaker (but non-trivial) result compared to our main result by using a generic upper bound on the gains. 
In Step 3, we will show how to substantially improve this bound, achieving the ratio of our main Theorem~\ref{thm:main} which is tight.

\begin{lemma}
For any super-face (except for the outer-face) in $\fset$, we have $gain(f) \geq 1.5$. 
\end{lemma}

As the outer (super-)face $f_0$ of $H[S]$ is special, we can achieve a lower bound on the quantity $gain(f_0)$ that depends on $\phi(S)$. This is captured by the following lemma.

\begin{lemma}
\label{lem:outer-face} 
For the outer-face $f_0$, we have that $gain(f) \geq \phi(S) - 1$. 
\end{lemma}

\begin{equation}
\label{eq:trivial-gains} 
\sum_{f \in \fset} gain(f) \geq \phi(S) -1 +1.5 (|\fset| - 1) = \phi(S) + 1.5|\fset| -0.5    
\end{equation}
 
The following lemma upper bounds the number of skeleton faces (i.e. super-faces of the skeleton.)  

\begin{lemma}
\label{lem:skeleton-faces} 
$|\fset| = a_1+ a_2 +1 \leq 2p - 2$.
\end{lemma}
\begin{proof} 
Proposition~\ref{prop:structure-heavy} allows us to modify the graph $H$ into another simple planar graph $\widetilde{H}$ such that the claimed upper bound on $|\fset|$ will follow simply from Euler's formula.

\begin{figure}[H]
    \centering
    \includegraphics[width=0.7\textwidth]{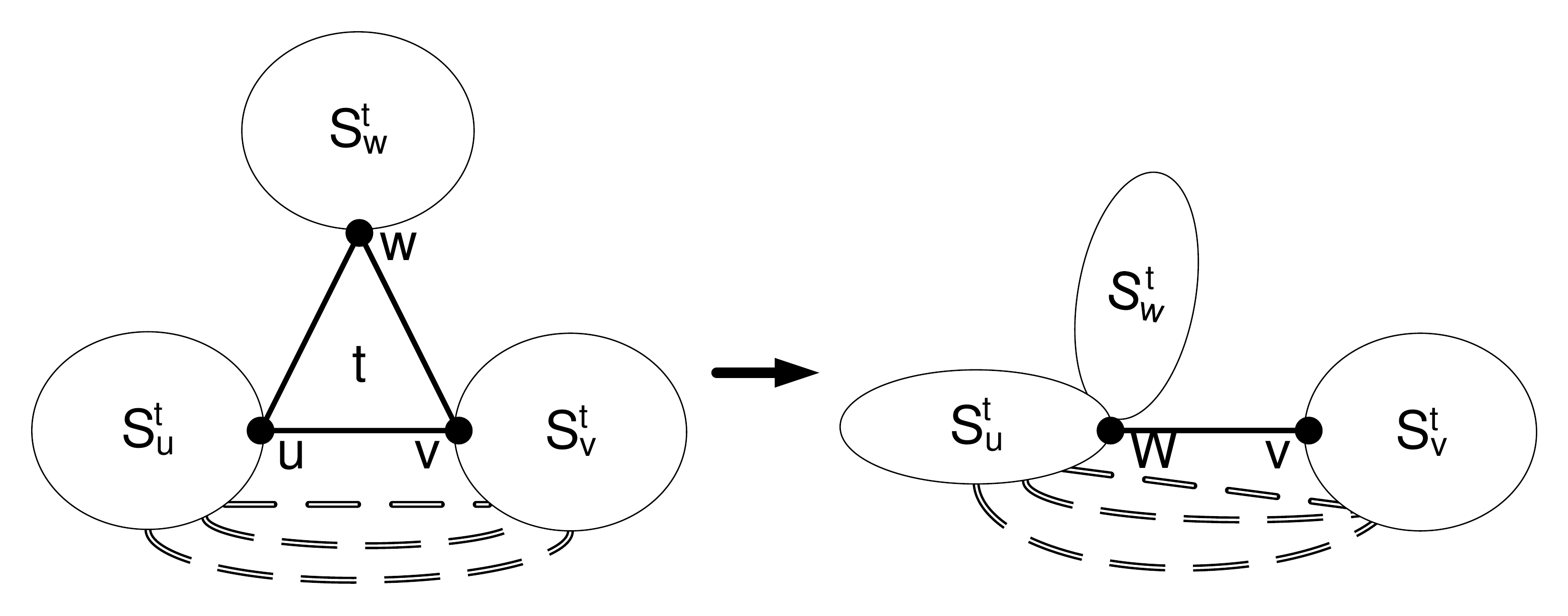}
    \caption{An example of the contraction transformation.}
    \label{fig:eta_2p}
\end{figure}

Let $t$ be a cactus triangle where $V(t) = \{u, v, w\}$ and $uw \in E(t)$ be such that the edge set $B^t_{uw}$ is empty, as guaranteed in Proposition~\ref{prop:structure-heavy}. 
For every cactus triangle $t$ we contract the edge $uw$ into one new vertex $W$. Note that this operation creates two parallel edges with endpoints $W$ and $v$ in the resulting graph. To avoid multi-edges in the resulting graph $\widetilde{H}$ we remove one of them (see Figure~\ref{fig:eta_2p} for an illustration of this operation).
Since $B_{uw}^t$ is empty this operation cannot create any other multi-edges in $\widetilde{H}$. In addition the contraction of an edge maintains planarity, hence after each such transformation the graph remains simple and planar.
As a result of applying the above operation to all cactus triangles, the graph $\widetilde{H}$ has $p+1$ vertices and $p$ edges corresponding to the contracted triangles. By Euler's formula the number of edges in $\widetilde{H}$ is at most $3(p+1) - 6 = 3p - 3$, which implies that $a_1 + a_2 \leq 2p - 3$, and as $|\fset|=a_1+a_2+1$ we get that $|\fset| \leq 2p - 2$.
\end{proof}

Combining the trivial gains (i.e. Inequality~\ref{eq:trivial-gains}) with Inequality~\ref{eq:gain-highlighted}, we get 
\[q \leq (4p+0.5p_1 + 2.5a_1+3a_2) - (\phi(S) + 1.5 (a_1+a_2+1) - 2.5) = 4p +0.5p_1 +a_1 +1.5 a_2 - \phi(S) +1\] 
Now, using Lemma~\ref{lem:skeleton-faces} and the trivial bound that $p_1 \leq p$, we get $q(S) \leq 4.5p + 1.5 (a_1+ a_2) - \phi(S) +1 \leq 7.5 p(S) - \phi(S)$, therefore implying a factor $7.5$ upper bound.

\subsection*{Step 3: Upper Bounding Gains via Super-Face Classification}

In this final step, we show another crucial idea that allows us to reach a factor $6$. 
Intuitively, the most difficult part of lower bounding the total gain is the fact that the value of $gain(f)$ is different for each type of super-face, and one cannot expect a strong ``universal'' upper bound that holds for all of them. For instance, Figure~\ref{fig:tight-1.5-face} shows a super-face with $gain(f) = 1.5$, so strictly speaking, we cannot improve the generic bound of $1.5$.  

\begin{figure}[H]
    \centering
    \includegraphics[width=0.2\textwidth]{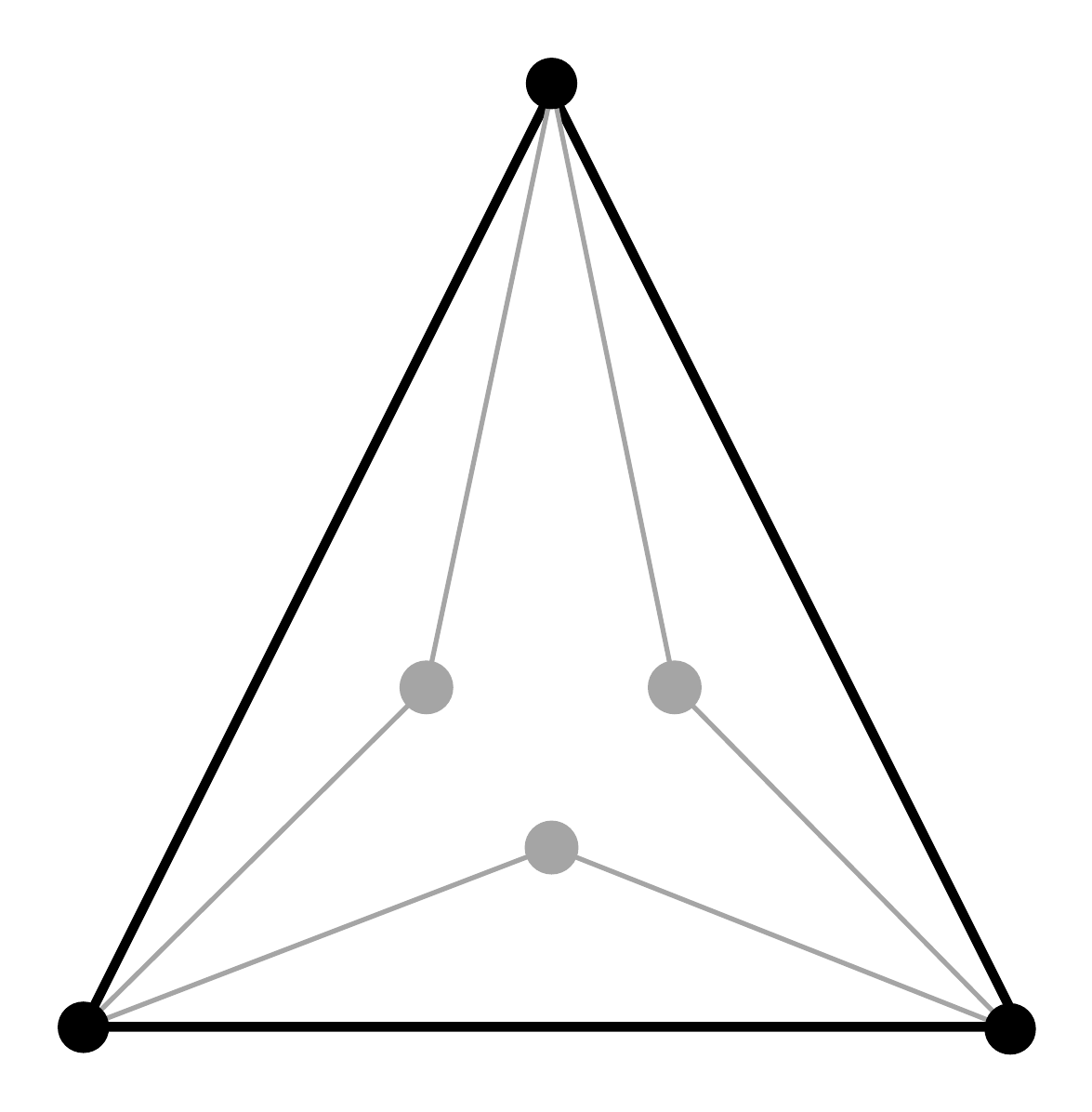}
    \caption{A super-face $f\in \fset$ having $gain(f) = 1.5$; $\mu(f)=1.5$ and $survive(f)=0$.}
    \label{fig:tight-1.5-face}
\end{figure}

This is where we introduce our final ingredient, that we call {\em classification scheme}. 
Roughly, we would like to ``classify'' the super-faces in $\fset$ into several types, each of which has the same gain. 
Analyzing super-faces with similar gains together allows us to achieve a better result.

\subparagraph*{Super-face classification scheme:} We are interested in coming up with a set of rules $\Phi$ that classifies $\fset$ into several types. 
We say that the rule $\Phi$ is a $d$-type classification if the rules classifies $\fset$ into $d$ sets $\fset = \bigcup_{j=1}^d \fset[j]$. Let $\vec{\chi}$ be a vector such that $\vec{\chi}[i] = |\fset[i]|$.
We would like to prove a good lower bound on the gain for each such set.  
We define the gain vector by $\overrightarrow{gain}$ where $\overrightarrow{gain}[i] = \min_{f \in \fset[i]} gain(f)$. 
The total gain can be rewritten as: 
\[\sum_{f \in \fset} gain(f) = \overrightarrow{gain} \cdot \vec{\chi} \] 
Notice that, the total gain value $\overrightarrow{gain} \cdot \vec{\chi}$ would be written in terms of the $\vec{\chi}[j]$ variables, so we would need another ingredient to lower bound this in terms of variables $p$'s and $a$'s. 
Therefore, another component of the classification scheme is a set of {\em valid linear inequalities} $\Psi$ of the form $\sum_{j=1}^d C_j \vec{\chi}[j] \leq \sum_{j \in \{0,1\}} d_j p_j + \sum_{j \in \{1,2\}} d'_j a_j$. 
This set of inequalities will allow us to map the formula in terms of $\vec{\chi}[j]$ into one in terms of only $p$'s and $a$'s. 

A classification scheme is defined as a pair $(\Phi, \Psi)$. We say that such a scheme certifies the proof of factor $\gamma$ if it can be used to derive $q(S) \leq \gamma p(S) - \phi(S)$. 
Given a fixed classification scheme and a gain vector, we can check whether it certifies a factor $\gamma$ by using an LP solver (although in our proof, we would show this derivation.)

Our main result is a scheme that certifies a factor $6$. Since the proof is complicated, we also provide a simpler, more intuitive proof that certifies a factor $7$ first. 

\begin{theorem}
There is a $5$-type classification scheme that gives a factor $7$. 
\end{theorem}

We remark that the analysis of factor $7$ only requires a cactus that is locally optimal for $1$-swap. 

\begin{theorem}
There is a $13$-type classification scheme that gives a factor $6$. 
\end{theorem}

\subparagraph*{Intuition:} The classification scheme would intuitively set the rules to separate the super-faces that would benefit from local search's exchange argument from those that would not. 
Therefore, for the good cases, we would obtain a much better gain, e.g., in one of our classification type, $gain(f)$ is as high as $4.5$. 
In the bad cases that there is no such benefit, we would still use the lower bound of $1.5$ that holds in general for any super-face.

\section{Reduction to Heavy Cacti (Proof of Theorem~\ref{thm:reduction-to-heavy})} 
\label{sec:reduction-to-heavy}

Let $t$ be a light triangle.  Assume that the bound $q(S) \leq \gamma p(S) - \phi(S)$ holds for all $S$ where {$G[S]$} contains only heavy triangles. 
Our goal is to prove that it holds for all $S$.
We will prove this by induction on the number of light triangles {$G[S]$} contains. 
The base case (when all triangles are heavy) follows from the precondition {and the trivial base case when $|S|=1$ is clearly true}.
Now assume that there is a light triangle $t$ in the {in graph $G[S]$}. 
Our plan is to apply the induction hypothesis on the {subgraphs $\{G[S^t_v]\}_{v \in V(t)}$ since each $G[S^t_v]$ contains less light triangles than $G[S]$}. 

Since we will be dealing with light triangle $t$, the following proposition (proof in Appendix~\ref{subsec:proof-prop-light}) gives some important structural properties of such a triangle: 

\begin{proposition}[Structure of light triangles]
\label{prop:structure-light}
Let $t$ be a light triangle in $\cset[S]$. The following statements hold: 

\begin{itemize}
    \item If $t$ is a light type-$0$ triangle and $uv \in E(t)$, such that $B^t_{ww'} = \emptyset$ for all $ww' \in E(t) \setminus \{uv\}$, then the total number of cross triangles supported by edges in $B^t_{uv}$ is at most two. 
    
    \item If $t$ is a light type-$1$ triangle and the edge $uv \in E(t)$ supports the cross triangle supported by $t$ and $B^t_{ww'} = \emptyset$ for all $ww' \in E(t) \setminus \{uv\}$, then the total number of cross triangles supported by edges in $B^t_{uv}$ is at most one. 
    
    \item If $t$ is a light triangle where edges in $\bigcup_{uv \in E(t)} B_{uv}^t${$\cup E(t)$} support {either two or} three cross triangles {such that at least two different set of edges $\{uv\} \cup B_{uv}^t$ for $uv \in E[t]$ supports a cross triangle each, then each set of edges $\{uv\} \cup B_{uv}^t$ supports at most one cross triangle and all the supported cross triangles have the same landing component.}
\end{itemize}
\end{proposition}

We will also need the following observation. 

\begin{Observation}
\label{Observation:different-components-green-blue-cycle}
Any circuit $C$ in $G$, which comprises of only cactus, type-$0$, type-$1$ and type-$2$ edges and cactus vertices, divides the plane into several regions (two if $C$ is a cycle) such that any cross triangle which is drawn in one of the regions cannot share its landing component with any other cross triangle drawn in some different region.
\end{Observation}

\paragraph{Free and occupied edges:} We call the edges in the outer-face $f_S$ of $G[S]$ that contribute to $\phi(S)$ {\em free} and every other edge in $f_S$ that is not free is called {\em occupied}. Let ${o(S)}$ be the total number of occupied edges. It follows that ${\phi(S)} = \ell(S) - {o(S)}$.

\subsection{Inductive proof} 
Now we proceed with the proof. 
Consider a cactus triangle $t \in G[S]$ with $V(t)=\{u,v,w\}$ which is light.
To upper bound $q(S)$, we break it further into two distinct terms $q'+q''$:

The term $q'$ counts all triangles with all the three vertices in the same split component and the cross triangles supported by edges or triangles in $G[S_x^t]$ for some $x \in \{u,v,w\}$.
As each split component of $t$ is also a cactus subgraph, by induction we have for $G[S_x^t]$ for all $x \in \{u,v,w\}$: $q(S_x^t) \leq \gamma p(S_x^t) - {\phi(S_x^t)}$.
and as $q'$ is equal to the sum over $q(S_x^t)$ for all $x \in \{u,v,w\}$ we get 
\[q' \leq \gamma (p-1) - ({\phi(S_u^t) + \phi(S_v^t) + \phi(S_w^t)})  =  \gamma p - ({\phi(S_u^t) + \phi(S_v^t) + \phi(S_w^t)}) - \gamma\]  

The term $q'' $ counts all remaining triangles in $q(S)$, \ie the triangles whose vertices belong to at least two different split components of $t$.
We will proceed to show that
\[q'' \leq 6 + {\phi(S_u^t) + \phi(S_v^t) + \phi(S_w^t) - \phi(S)}\] 
hence, upper bounding $q'+q''$ by the desired quantity for any $\gamma \geq 6$.     

To this end, we upper bound the contributions to $q''$ from two separate terms: The first term, $q''_1$, is the number of cross triangles supported by the edges in $B_{uv}^t \cup B_{uw}^t \cup B_{vw}^t$ plus the cross triangles supported by $t$ {plus one for $t$ itself}, and (ii) The second term, $q''_2$, is the number of ``surviving'' triangular faces in $G[S] \setminus (\bigcup_{x \in V(t)} G[S^t_x])$ without any cross triangle drawn inside it.

Note that by definition of light triangles, there are at most three cross triangles supported by the edges in $B_{uv}^t \cup B_{uv}^t \cup B_{vw}^t$ and $t$ itself. 
Now we consider two cases, based on the value of $q''_1$. 

\begin{itemize}
    \item(At most two supported cross triangles): In this case $q''_1 \leq 3$, \ie $t$ itself and the supported cross triangles. Hence if we can show that $q''_2 \leq 3 + {\phi(S_u^t) + \phi(S_v^t) + \phi(S_w^t) - \phi(S)}$, then we are done.
    
    \item(Exactly three supported cross triangles): Similarly in this case $q''_1 = 4$, \ie $t$ itself and the supported cross triangles. Hence showing that $q''_2 \leq 2 + {\phi(S_u^t) + \phi(S_v^t) + \phi(S_w^t) - \phi(S)}$ gives us the entire reduction. 
    
\end{itemize}

In particular, the following lemma (which we spend the rest of this section proving) will complete the proof of Theorem~\ref{thm:reduction-to-heavy}.

\begin{lemma}
\label{lem:survive_light}
 For any light triangle $t$, the number of surviving triangles $q''_2$ is at most $3 + \phi(S_u^t) + \phi(S_v^t) + \phi(S_w^t) - \phi(S)$. Moreover, if there are three cross triangles supported by the edges in $B_{uv}^t \cup B_{uw}^t \cup B_{vw}^t$ and $t$ itself, then $q''_2$ is at most $2 + \phi(S_u^t) + \phi(S_v^t) + \phi(S_w^t) - \phi(S)$.
\end{lemma}

\subsection{Proof of Lemma~\ref{lem:survive_light}} 

To facilitate the counting arguments that we will use, we will be working with an auxiliary graph $\widetilde{G}$ instead of $G[S]$.
Let $\Gamma_x$ be the cycle (in particular, the set of edges on the cycle) bounding the outer-face of $G[S_x^t]$ for $x\in \{u,v,w\}$ and let $\Gamma$ be the cycle bounding the outer-face of $G[S]$ (so $\Gamma$ contains exactly all the outer-edges). Because $\cset[S]$ is a connected triangular cactus, there cannot be any repeated edge in these faces, hence $\Gamma$, $\Gamma_i$'s are circuits; the vertices can occur multiple times in $\Gamma_x$. Now we {\em cut open} each of the circuits $\Gamma$,$\Gamma_x$, for each $x \in \{u, v, w\}$ to convert them to simple cycles. 
The idea is to make copies (equal to the number of times it appears in the corresponding circuit) of each vertex contained in the circuit and joining the edges incident to the original vertex to one of the copies, such that the structure of the drawing is preserved. We also make sure that there exists a triangular face corresponding to $t$ containing some copy of each of the vertex in $\{u, v, w\}$. 
After cut opened, $\Gamma_x$, for each $x \in \{u, v, w\}$ will be empty cycle in $\widetilde{G}$. 
Notice that the values of $\phi$ as well as the types of edges on these cut-opened cycles are preserved. 



Note that the surviving triangles that contribute to $q''_2$ correspond exactly to the triangles drawn in the regions of $G$ exterior of $\Gamma_x$ for all $x \in \{u,v,w\}$ but in the interior of $\Gamma$. Also, $t$ is drawn inside of $\Gamma$.
In order to bound $q''_2$ we construct an auxiliary graph $\widetilde{G}$ as follows. 
For each $x \in \{u,v,w\}$, we remove all edges and vertices drawn in the interior of cycle $\Gamma_i$ from $G[S]$. The resulting graph after such a removal is our $\widetilde{G}$, such that $V(\widetilde{G}) = V(\Gamma) \cup V(\Gamma_u \cup \Gamma_v \cup \Gamma_w) = V(\Gamma_u \cup \Gamma_v \cup \Gamma_w)$.
Any triangle that contribute to the term $q''_2$ also exist as triangular faces in $\widetilde{G}$, so we only need to upper bound $f_3(\widetilde{G})$. 

\begin{claim}
If $E(\Gamma) \setminus (E(t) \cup E(\Gamma_u \cup \Gamma_v \cup \Gamma_w)) = \emptyset$, then the bound for $q''$ holds. 
\end{claim}
\begin{proof}
If the set is empty, then $q''_2=0$ and ${\phi(S) \leq \phi(S_u^t) + \phi(S_v^t) + \phi(S_w^t)} + 3$ in general. In the three cross triangles case, having no such edge implies that $t$ is a type-$3$ triangle, {because all the three cross triangles has to be supported by $E(t)$} and hence ${\phi(S) = \phi(S_u^t) + \phi(S_v^t) + \phi(S_w^t)}$.
\end{proof}

Now we continue with the case where there exists at least on edge {in $E(\Gamma) \setminus (E(t) \cup E(\Gamma_u \cup \Gamma_v \cup \Gamma_w))$}.
Clearly, $\widetilde{G}$ is a subgraph of $G[S]$ and any surviving triangle in $G$ must be drawn in a region of $\widetilde{G}$.
In order to bound the number of surviving triangles corresponding to $q''_2$, we will first identify these regions and then make a region-wise analysis to get the full bound. 
For this purpose, we remove any non-cactus edge from $\widetilde{G}$ {that is drawn in the interior of $\Gamma$ and does not belong to one of $\Gamma_u,\Gamma_v$ or $\Gamma_w$} to form another auxiliary graph $\widetilde{G}'$. The faces in the graph $\widetilde{G}'$ which are drawn inside the cycle $\Gamma$ and outside every cycle $\Gamma_x$ (except the triangular face $t$), will correspond to the regions in $\widetilde{G}$ which we would analyze later. First we prove the following claim which quantifies the structure of these regions (see Fig.~\ref{fig:region-R} which illustrates all possible structures for these regions).

\begin{figure}
    \centering
    \includegraphics[width=\textwidth]{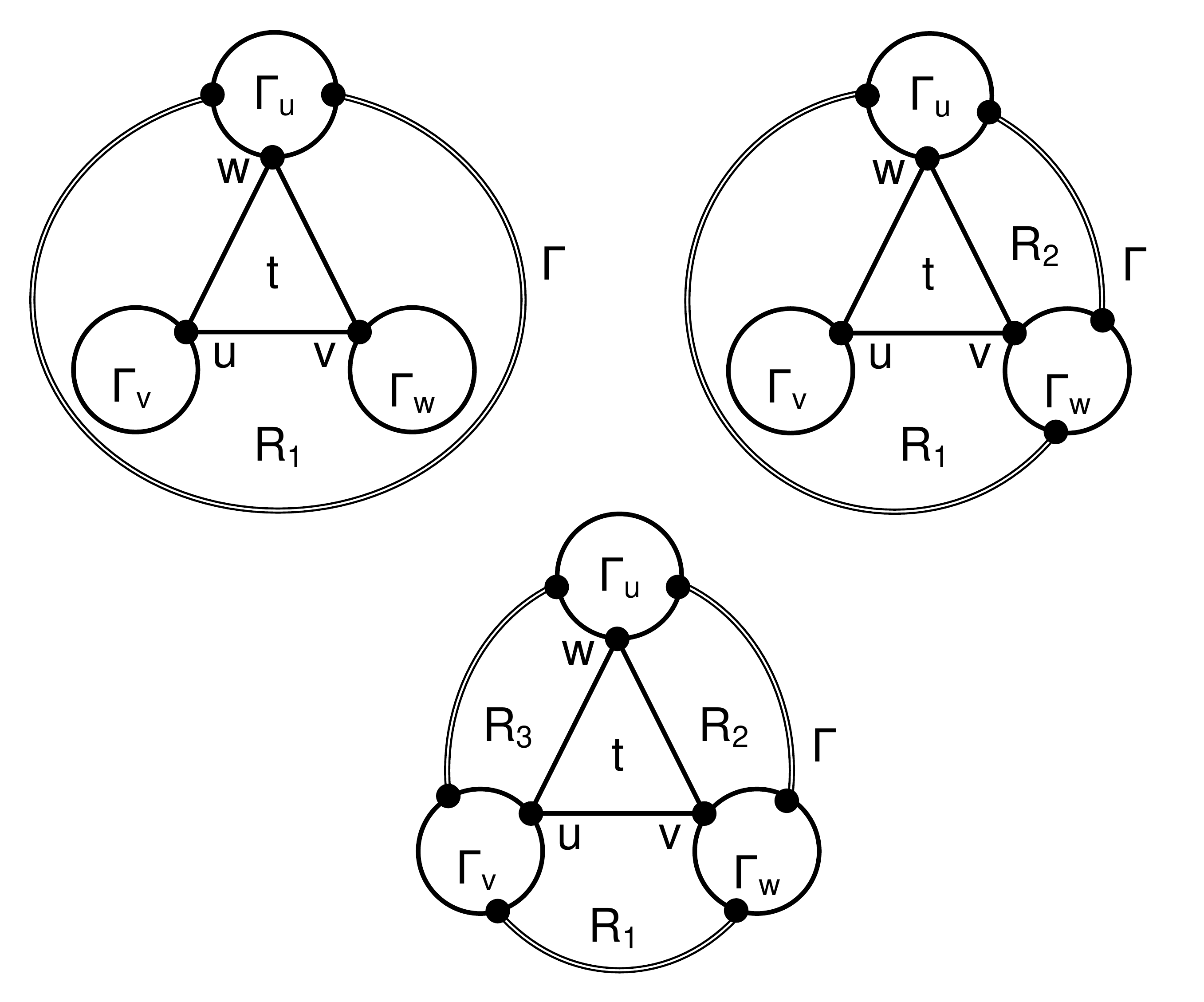}
    \caption{The three structures corresponding to each $k \in \{1, 2, 3\}$ for the faces $R_1,\ldots R_k$ of $\widetilde{G}'$. These corresponds to the regions of $\widetilde{G}$ which we analyze separately to get our bound on $q''_2$.}
    \label{fig:region-R}
\end{figure}

\begin{claim} \label{claim:regions-to-triangulate}
If $R_1,\ldots,R_k$ (except the triangular face $t$) are the faces in $\widetilde{G}'$ which are drawn inside $\Gamma$ and outside every cycle $\Gamma_x$ for each $x \in \{u, v, w\}$, then $1 \leq k \leq 3$. Moreover, every such face contains exactly one edge of $\Gamma$.
\end{claim}

The proof of this claim appears later in this section.

Let $R_1,\ldots,R_k$ (for $1 \leq k \leq 3$) be the regions in $\widetilde{G}$ which are the faces of $\widetilde{G}'$ given by the above claim (see Figure \ref{fig:region-R} for an illustration). We denote by $\ell(R_i)$\footnote{Notice that we slightly abuse the notation $\ell(\cdot)$ here. Before, we use $\ell(S)$ where $S$ is a subset of cactus-vertices, and now we are using $\ell(R)$ where $R$ is a cycle bounding a region.} the overall number of edges and by $o(R_i)$ the number of {occupied} edges in the boundary of $R_i$ (these are the edges belonging to some cycle $\Gamma_x$ for $x \in \{u, v, w\}$.) In the next step, we will upper bound the number of surviving triangles that exist in $G$ in each such region $R_i$.

\begin{Observation}
\label{obs:faces-in-each-region}
Any face in the graph $\widetilde{G}$ which is drawn inside one of the regions $R_i$ contains vertices from at least two cycles $\Gamma_x, \Gamma_y$ for $x, y \in \{u, v, w\}$ and $x \neq y$.
\end{Observation}


How many surviving triangles can there be in region $R_i$?
Intuitively, if we triangulate $R_i$ by adding edges in its interior, we would have $\ell(R_i)- 2$ triangular faces. Among these faces, $o(R_i)$ of them would not be surviving since the edge bounding the face is occupied. 
In certain cases, we would get an advantage and the term would become $-3$ instead of $-2$.

\begin{claim} 
\label{claim:st-each-region}
The number of surviving triangles drawn inside $R_i$ in $\widetilde{G}$ are at most $\ell(R_i) - o(R_i) - 2$. 
{Moreover, if the common landing component $L$ for the three cross triangles supported by $B_{uv}^t \cup B_{vw}^t \cup B_{uw}^t \cup E(t)$ is drawn inside $R_i$, then we get the stronger bound of $\ell(R_i) - o(R_i) - 3$.}
\end{claim}

The proof of this claim relies on a standard triangulation trick used in the context of planar graphs.
We defer the proof to later in Section~\ref{sec:st-each-region}. 

Now we are ready to complete the proof for Lemma~\ref{lem:survive_light}.

{Let $\one_S^t \in \{0, 1\}$ be the indicator variable such that $\one_S^t = 1$ if we are in the case when there exists exactly three cross triangles supported by $B_{uv}^t \cup B_{vw}^t \cup B_{uw}^t \cup E(t)$ such that the common landing component $L$ for these triangles is drawn inside some region $R_i$, otherwise $\one_S^t = 0$.} Using the bounds for each region from Claim~\ref{claim:st-each-region} we can upper bound $q''_2$ by summing over the number of surviving triangles in each region.

\begin{align}\label{eq:sum-over-regions}
q''_2 & \leq \sum_{i=1}^k (\ell(R_i) - o(R_i)-2) {- \one_S^t}\nonumber \\
    & \leq \sum_{i=1}^k \ell(R_i) - \sum_{i=1}^k {o(R_i)} - 2k {- \one_S^t}
\end{align}

Next we take a closer look at the $\ell({R_i})$ term in the sum. 
By Claim~\ref{claim:regions-to-triangulate}, each region $R_i$ contains exactly one edge of $\Gamma$, and $R_i \subseteq \Gamma \cup E(t) \cup \left(\bigcup_{x \in V(t)} \Gamma_x \right)$. 
Therefore, we can decompose the length of face $R_i$ into three parts: 

\begin{align*}
\ell(R_i) & = 1+ \sum_{x \in V(t)} |E(R_i) \cap \Gamma_x| + |E(R_i) \cap E(t)| 
\end{align*}

Plugging this into Eq.~(\ref{eq:sum-over-regions}) we get,

\begin{align}\label{eq:sum-over-regions-2}
    q''_2 & \leq \sum_{i=1}^k (1 + \sum_{x \in V(t)} |E(R_i) \cap \Gamma_x| + |E(R_i) \cap E(t)|]) - \sum_{i=1}^k {o({R_i})}-2k {- \one_S^t} \\
    & \leq \sum_{i=1}^k (\sum_{x \in V(t)} |E(R_i) \cap \Gamma_x| + |E(R_i) \cap E(t)|) - \sum_{i=1}^k {o({R_i})} - k {- \one_S^t}
\end{align}

\begin{figure}
    \centering
    \includegraphics[width=0.5\textwidth]{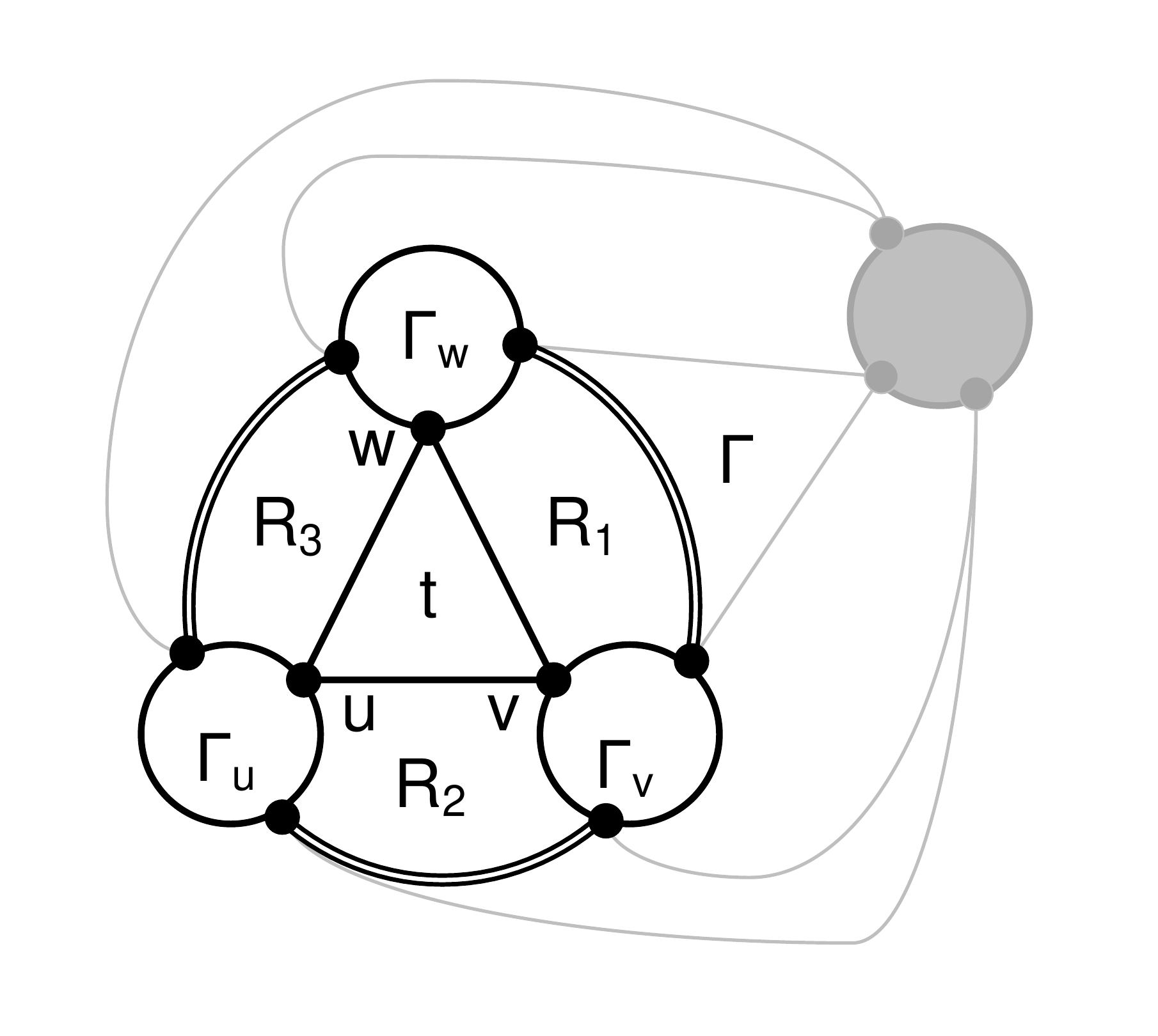}
    \caption{The case when the three supported cross triangles are drawn in the exterior of $\Gamma$. This can only happen when we are in the three regions $R_1, R_2, R_3$ case.}
    \label{fig:three-gray-ext-gamma}
\end{figure}

Note that $t$ can not contribute more than its three edges to the boundaries of all $k$ regions, thus $\sum_{i=1}^k |E(R_i) \cap E(t)| \leq 3$. Using this in Eq~(\ref{eq:sum-over-regions-2}), we get

\begin{align}\label{eq:sum-over-regions-3}
q''_2 & \leq 3 + \sum_{i=1}^k \sum_{x \in V(t)} |E(R_i) \cap \Gamma_x| - \sum_{i=1}^k {o({R_i})} - k {- \one_S^t}
\end{align}

\begin{claim} 
$\sum_{i=1}^k \sum_{x \in V(t)} |E(R_i) \cap \Gamma_x| = \ell(S_{u}^t)+\ell(S_{v}^t)+\ell(S_{w}^t) - \ell(S) + k$ 
\end{claim}
\begin{proof}
Notice that the sum on the left-hand-side counts all edges in $(\bigcup_{x \in V(t)} \Gamma_x) \setminus \Gamma$ where each edge is counted exactly once, and this contribution is $\sum_{x \in V(t)} \ell(S_x^t)  - \ell(S)$.   
Additionally, by Claim~\ref{claim:regions-to-triangulate}, each edge in $\Gamma \setminus (\bigcup_{x \in V(t)} \Gamma_x)$ is also counted exactly once as well, and this contribution is $+k$. 
\end{proof}

Combining all of this with Inequality~(\ref{eq:sum-over-regions-3}) we get,

\begin{align}\label{eq:sum-over-regions-4}
    q''_2 & \leq 3 + \ell(S_u^t) + \ell(S_v^t) + \ell(S_w^t) -\ell(S) - \sum_{i=1}^k {o({R_i})}  {- \one_S^t}
\end{align}

Let $o^t_{across}(S)$ be the number of {occupied} edges among the $o(S)$ {occupied} edges belonging to $\Gamma$ such that they do not belong to any of the $\Gamma_x$ for $x \in \{u, v, w\}$. These edges are the ones which are drawn across two different cycles $\Gamma_x, \Gamma_y$ for $x,y \in \{u, v, w\}$ and $x \neq y$ (potentially some of the edges drawn in double-line style in Fig.~\ref{fig:region-R}). Hence $o^t_{across}(S)$ captures precisely the number of occupied edges in $\Gamma \setminus (E(t) \cup \bigcup_{x \in V(t)} \Gamma_x)$ for which the supported cross triangles are drawn in the exterior of $\Gamma$. By the way we define $o({R_i})$, the following equality holds.
{ 
\begin{align}\label{eq:bad-edges-relation}
\sum_{i=1}^k o(R_i) = o(S_u^t) + o(S_v^t) + o(S_w^t) - ( o(S) - o^t_{across}(S))
\end{align}}
Using this in Inequality~(\ref{eq:sum-over-regions-4}) we get,

\begin{align*}
q''_2 & \leq  3 + \ell(S_u^t) + \ell(S_v^t) + \ell(S_w^t) - \ell(S) {- (o(S_u^t) + o(S_v^t) + o(S_w^t) - (o(S) - o^t_{across}(S))) - \one_S^t}\\
 & \leq  3 + {(\ell(S_u^t)-o(S_u^t)) +  (\ell(S_v^t)-o(S_v^t)) +  (\ell(S_w^t)-o(S_w^t)) -  (\ell(S)-o(S))  - o^t_{across}(S) - \one_S^t}
 \end{align*}
 
Since $\ell(S_x^t) = {\phi(S_x^t) + o(S_x^t)}$ for every $x \in \{u, v, w\}$ we get,
\begin{align}\label{eq:sum-over-regions-5}
q''_2 & \leq   3 + {\phi(S_u^t) + \phi(S_v^t) + \phi(S_w^t) - \phi(S)}  {- o^t_{across}(S)} {- \one_S^t} 
\end{align}

The general inequality $q''_2 \leq  3 + {\phi(S_u^t) + \phi(S_v^t) + \phi(S_w^t) - \phi(S)}$ for the Lemma~\ref{lem:survive_light} trivially follows from the above inequality.
The following claim will complete the proof.

\begin{claim}
If  there are three cross triangles supported by edges in $\bigcup_{uv \in E(t)} B_{uv}^t \cup E(t)$ with the common landing component $L$, then ${o^t_{across}(S) + \one_S^t \geq 1}$. 
\end{claim}
\begin{proof}
There could be two sub-cases: (i) The landing component $L$ is in the exterior of $\Gamma$. In this case, by the definition of {$o^t_{across}(S) \geq 1$}, all the three edges which support one of the three cross triangles will contribute to $o^t_{across}(S)$ (see Fig~\ref{fig:three-gray-ext-gamma} for illustration); and (ii) The cross triangles are drawn inside $\Gamma$. In this case, we have that {$\one_S^t=1$.}
In any case, we have ${o^t_{across}(S) + \one_S^t \geq 1}$, thus proving the lemma.
\end{proof}
  

\subsection{Proof of Claim~\ref{claim:regions-to-triangulate}}

By the assumption that there exists at least one edge in $E(\Gamma) \setminus (E(t) \cup E(\Gamma_u \cup \Gamma_v \cup \Gamma_w))$. Let $ab: = e \in E(\Gamma) \setminus (E(t) \cup E(\Gamma_u \cup \Gamma_v \cup \Gamma_w))$ be one such edge.

To prove the claim, we will show that for any such edge, there exists a unique face $R$ satisfying the conditions of the claim and it contains at least one edge from $E(t)$.
As each edge of $t$ is also incident to the face bounded by $t$, this would imply that there can not be more than three such faces in $\widetilde{G}'$ and since there exists the edge $e$, hence we will be done. 

Let $a \in \Gamma_x$ for some $x \in \{u, v, w\}$. We will always use the fact that, since $e \in E(\Gamma)$, there are two directions starting from $a$ to traverse the boundary of $\Gamma_x$, such that in one direction edges of $\Gamma_x$ belongs to $\Gamma$ and in the other they are drawn in the interior of $\Gamma$. No we split into two possible cases.
\begin{itemize}
    \item ($b \in \Gamma_y$ for some $y \in \{u, v, w\}$ such that $y \neq x$):  Since $a, x \in \Gamma_x$, there exists a path $P_x$ from $a$ to $x$ containing edges of $\Gamma_x$ such that all these edges are drawn in the interior of $\Gamma$ (possibly $x=a$ and $P_x$ is a zero length path). Similarly there exist a path $P_y$ going from $b$ to $y$ containing edges of $\Gamma_y$ such that all these edges are drawn in the interior of $\Gamma$. Hence the circuit $C$ which includes the edge $e$, the edge $xy \in E(t)$ and two paths $P_x$ and $P_y$, is drawn inside of $\Gamma$ (except the edge $e$ which is on the boundary of $\Gamma$). Clearly, there cannot be any other edge from $\Gamma$ which is drawn inside $C$, hence any face drawn inside $C$ can contain at most the edge $e$ from $E(\Gamma) \setminus (E(t) \cup E(\Gamma_u \cup \Gamma_v \cup \Gamma_w))$. Also, by the way we define $\widetilde{G}'$, there cannot be any other edge inside $C$ drawn across different $\Gamma_i$ cycles. Now if $t$ is drawn outside of $C$, then $C$ itself is the face $R$ of $\widetilde{G}'$ satisfying our requirements. Otherwise, the whole of $\Gamma_z$ for $z \neq x$ and $z\neq y$, is drawn inside of $C$. This means that region inside the circuit $C$ can be decomposed into the triangular face $t$, the cycle $\Gamma_z$ and another face $R$ whose boundary comprises of edges $xz, zy \in E(t)$, the edge of $\Gamma_z$, the edge $e$ and two paths $P_x$ and $P_y$. Hence $R$ is the face corresponding to $e$ which we require.
    \item ($b \in \Gamma_x$): Notice that in this case, the circuit comprising of edge $e$ along with a path $P_x$ from $a$ to $b$ containing edges of $\Gamma_x$ such that all these edges are drawn in the interior of $\Gamma$, will enclose the triangle $t$ and the other two cycles $\Gamma_y, \Gamma_z$ such that $y, z \in \{u, v, w\}$ and $x \neq y \neq z \neq x$. Similar to the previous case, there cannot be any other edge from $\Gamma$ which is drawn inside $C$, $C$ is drawn in the interior of $\Gamma$ (except the edge $e$ which is on the boundary of $\Gamma$) and also no other edge is drawn across different $\Gamma_i$ cycles inside of $C$. Hence, any face drawn inside $C$ can contain at most the edge $e$ from $E(\Gamma) \setminus (E(t) \cup E(\Gamma_u \cup \Gamma_v \cup \Gamma_w))$. Also, $C$ can be decomposed into the triangular face $t$, the cycles $\Gamma_y, \Gamma_z$ and another face $R$ whose boundary comprises of edge $e$, all three edges of $t$, all the edge of $\Gamma_y, \Gamma_z$ and two paths $P$ and $P'$ from $a$ to $x$ and $x$ to $b$ containing edges of $\Gamma_x$ drawn inside $\Gamma$. Hence $R$ is the face corresponding to $e$ which we require. 
\end{itemize}

\subsection{Proof of Claim~\ref{claim:st-each-region}} 
\label{sec:st-each-region}

To prove this we will perform a series of monotone operations within the region $R_i$ in graph $\widetilde{G}$, such that in each operation the number of surviving triangles drawn within $R_i$ cannot reduce. 
In the end we will reach a structure for which the bound holds trivially. Since the operations here are monotone, the bound which we get also holds for the original number of surviving triangles drawn within $R_i$.
Notice that we make these modifications in the auxiliary graph $\widetilde{G}$ only for counting purposes and never change the structure of our graph $G$. 

In the first step, except for the three cross triangles supported by the edges in $B_{uv}^t \cup B_{vw}^t \cup B_{uw}^t \cup E(t)$, we {\em decouple} all the other supported cross triangles drawn inside $R_i$ which share their landing components by adding a dummy landing vertex for each such cross triangles and making the new dummy vertex its landing component. 
Note that the decoupling step allows us to get a full triangulation for $R_i$ in its interior (except the face containing the common landing component $L$) and at the same time does not affect the number of surviving triangles drawn inside $R_i$ in $\widetilde{G}$. 

After this we triangulate the interior of $R_i$ by adding extra type-$0$ edges, such that the end point for each additional edge lies in two different $\Gamma_x$ and $\Gamma_y$ for $x \neq y$. This is possible to achieve due to Obs.~\ref{obs:faces-in-each-region} and also this operation is monotone and cannot reduce the number of surviving triangles drawn inside $R_i$ in $\widetilde{G}$. 
Also, all the faces inside $R_i$ are triangular faces except the one containing $L$ in graph $\widetilde{G}$.
The way we triangulate the regions of $R_i$ ensures that the Obs.~\ref{obs:faces-in-each-region} continues to hold which implies that any face in $R_i$ can contain at most one edge from the boundary of $\Gamma_x$ for any $x \in \{u, v, w\}$. Also, $\widetilde{G}$ will remain a simple planar graph since the added type-$0$ edge connect vertices from the boundary of two different cycles $\Gamma_x$ and $\Gamma_y$ for $x \neq y$. In the end, we have at most $\ell(R_i) - 2$ triangular faces and any {occupied} edge counted in $o({R_i})$ (\ie {occupied} edges in $R_i$ which belongs to some cycle $\Gamma_x$ for $x \in \{u, v, w\}$) can kill at most one triangle, hence the claimed upper-bound follows in the general case.

\begin{figure}
    \centering
    \includegraphics[width=0.5\textwidth]{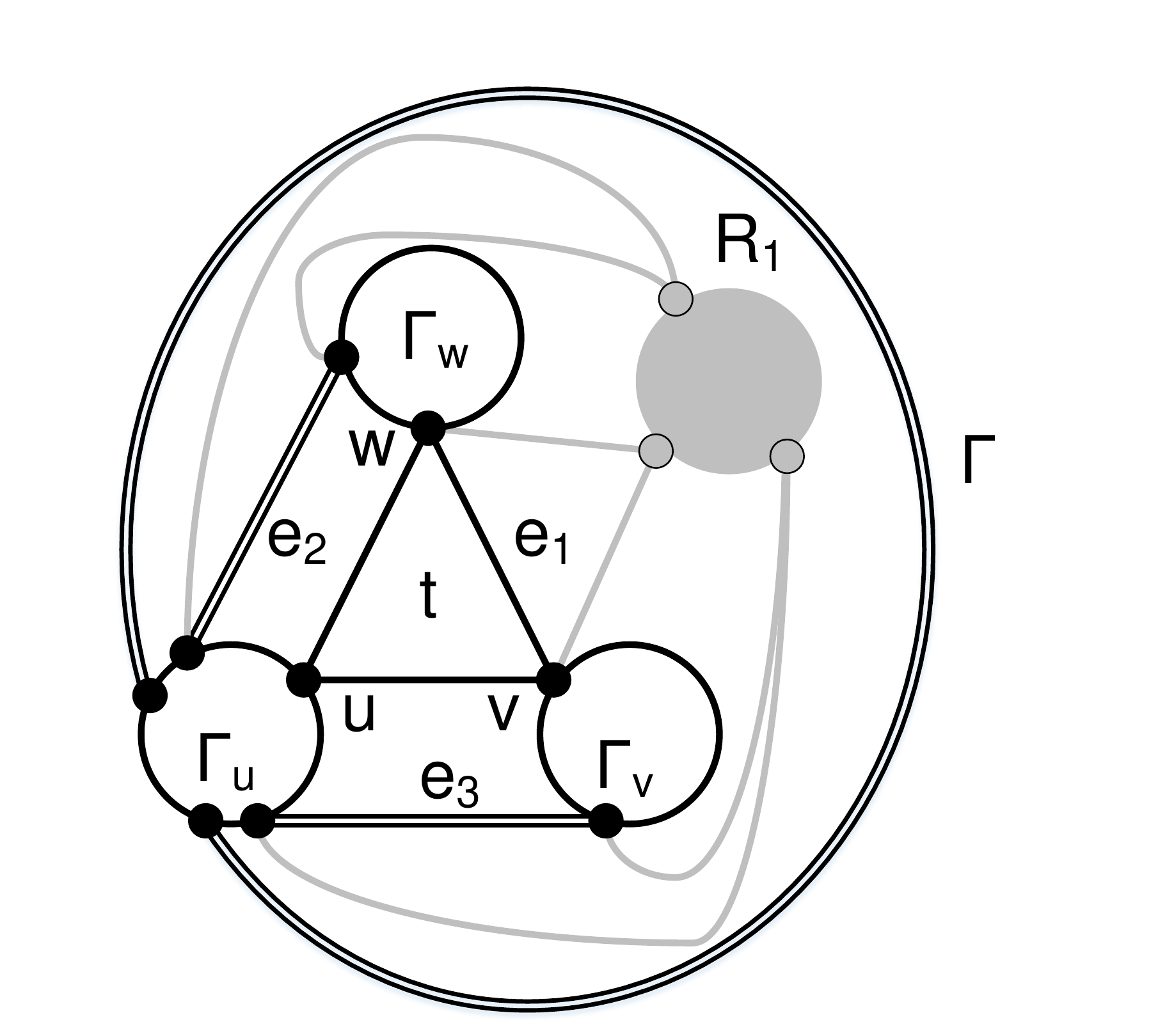}
    \caption{The case when the three supported cross triangles are drawn in the interior of $\Gamma$. This can only happen when we are in one region $R_1$ case.}
    \label{fig:three-gray-int-gamma}
\end{figure}

Now in the case where we have the three cross triangles supported by the edges in $B_{uv}^t \cup B_{vw}^t \cup B_{uw}^t \cup E(t)$, we will prove that the face (say $f$) of $R_i$ inside which the common landing component $L$ is drawn, contains at least one more edge in addition to the three edge from $B_{uv}^t \cup B_{vw}^t \cup B_{uw}^t \cup E(t)$ which supports the three cross triangles. This implies that this face has length at least $4$ and the triangulation of $R_i$ misses at least $2$ triangular faces. Also, in the worst case the fourth edge which we consider here could contribute to the term $o(R_i)$. Hence overall, we get at least $1$ less surviving triangular face than the previous bound and the claim follows.

To prove the claim for face $f$, first recall that (by Prop.~\ref{prop:structure-light}) the three edges in $B_{uv}^t \cup B_{vw}^t \cup B_{uw}^t \cup E(t)$ which supports the cross triangles are drawn across different pair of cycles $\Gamma_u, \Gamma_v, \Gamma_w$. Let $e_1 \in B_{vw}^t \cup vw$, $e_2 \in B_{uw}^t \cup uw$ and $e_3 \in B_{uv}^t \cup uv$ be the three edges supporting the three cross triangles. There is a cycle $C$ comprising of edges $e_1, e_2, e_3$ and paths $P_u, P_v, P_w$ joining the two ends of these edge in $\Gamma_u, \Gamma_v, \Gamma_w$ respectively, such that the triangle $t$ is drawn inside $C$ and the exterior of the $\Gamma$ is outside of $C$. Now since $R_i$ is a bounded region in graph $\widetilde{G}$ hence the face $f$ is a bounded face. Now we show that for $f$ to be bounded face, its length has to be at least four. In the corner case when $e_1 = vw, e_2 = uw, e_3 = uv$, $C$ is precisely the triangular face $t$ and the edges are $e_1, e_2, e_3$ are touching $f$ from the outside of $t$. Hence, for $f$ to be bounded, there should exist at least one more edge to complete the loop going from $\Gamma_u$ to $\Gamma_v$ to $\Gamma_w$ and back to to $\Gamma_u$. Otherwise, assume $e_1 \neq vw$ (other cases are symmetric). Since the cross triangles supported by $e_1, e_2, e_3$ share their landing component, and there exists a cycle $C'$ containing only cactus/type-$0$/type-$1$/type-$2$ edges including edges $e_1$, $vw$ and paths in $\Gamma_v$ and $\Gamma_w$ connecting the end points of $e_1$ and $vw$, such that the face $f$ should be drawn outside of $C$. Now again for $f$ to be bounded, it should contain one more edge and we are done.

\section{Classification Scheme for Factor $7$}
\label{sec:classification-7} 

In this section we will show a classification scheme that allows us to prove our factor seven result. For simplicity, from now on we will use $p, q$ instead of $p(S), q(S)$. More precisely, the aim is to prove the following lemma.


\begin{lemma} 
\label{lem:gain-factor-7} 
There is a $5$-type classification scheme for which  
\[- (\sum_{f \in \fset} gain(f))  \leq  - \phi(S) + (2p + 0.5p_1 - 2.5 a_1 - 3 a_2 - 1.5)\] 
\end{lemma}
First, we show that Lemma \ref{lem:gain-factor-7} is sufficient for proving the final result. 
For this we substitute the bound from Lemma~\ref{lem:gain-factor-7} into Inequality~(\ref{eq:gain-highlighted}) to get: 
\[q \leq (4p+0.5p_1 + 2.5a_1 + 3a_2) - g(S) + (2p + 0.5p_1 - 2.5 a_1 - 3 a_2 - 1.5) = 6p + p_1 - g(S) - 1.5  \] 

This implies $q \leq 7p - \phi(S)$ as desired. 

In order to define the classification schemes, we further classify the edges, vertices and split components for any heavy triangle $t$ in $G[S]$ into several types.

\paragraph{Further classification of cactus vertices, edges and split components:}
The cactus edges on each heavy triangle are further classified into {\em free} and {\em base} edges as follows: For any heavy triangle $t$, where $V(t) = \{u,v,w\}$. Let $uv \in E(t)$ be an edge for which $B^t_{uv} \neq \emptyset$.
By Proposition~\ref{prop:structure-heavy} there is exactly one such edge in $E(t)$. We say that the edge $uv$ is the {\em base} edge and both $u$ and $v$ are called {\em base} vertices. 
We say that the other two edges in $E(t) \setminus uv$ are {\em free}, and the vertex $w$ is called a {\em free} vertex. Both $S^t_u$ and $S^t_v$ are called {\em occupied} components and $S^t_w$ is a {\em free} component. See Figure~\ref{fig:classification-free-base-occupied} for an illustration.

    \begin{figure}[H]
        \centering
        
        
        
        \includegraphics[width=0.9\textwidth]{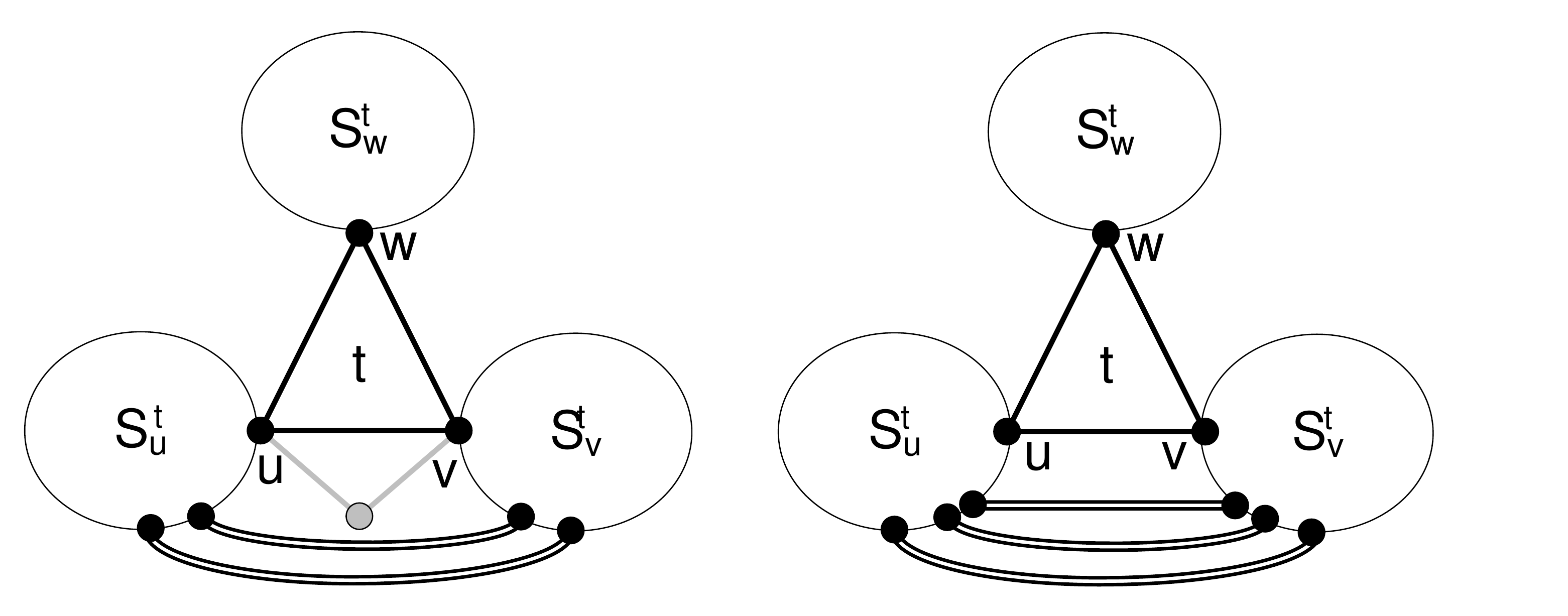}
        \caption{Classification of of cactus edges and split components, based on the type-$1$ and type-$2$ edges going across the split components of any heavy triangle $t$. The split components $S_u^t, S_v^t$ are occupied components and $S_w^t$ is the free component of $t$. For a type-$1$ triangle (left figure), we know that $uv$ also supports a cross triangle.}
        \label{fig:classification-free-base-occupied}
    \end{figure}

The following claim follows from the properties of heavy type-$0$ and type-$1$ triangles shown in Proposition \ref{prop:structure-heavy}.
\begin{claim}
\label{claim:free-edges-same-face}
The two free cactus edges of any cactus triangle are part of the same super-face in $\fset$.
\end{claim}
\begin{proof}
Let $vw$ and $uw$ be the free edges in $E(t)$.
Assume for contradiction that there is a super-face $f \in \fset$ that only contains $uw$ but not $vw$. Any super-face boundary needs to contain at least one type-$1$ or type-$2$ edge in order to form a cycle. Therefore a path along the super-face $f$, not including the edge $uw$, from $u$ to $w$ must leave $S_w^t$ using a type-$1$ or type-$2$ edge, a contradiction to the fact that for a heavy triangle, $B_{vw}^t$ and $B_{uw}^t$ are empty in graph $H[S]$.
\end{proof}

We will upper bound the number of surviving triangles inside each super-face $f \in \fset$ based on the characteristics of the edges bounding $f$ (see Figure~\ref{fig:classification-face-edges}). 

\paragraph{Classification of Edges in the Face Boundaries of $G[S]$:} 
Edges that bound $f$ are further partitioned into the following types:
\begin{itemize}
    \item The two free edges of the cactus triangles. Let $p_0^{free}(f)$ and $p_1^{free}(f)$ denote the total number of type-$0$ and type-$1$ triangles respectively whose free edges participate in $f$. 
    \item The base edges of the cactus triangles. Let $p_0^{base}(f)$ and $p_1^{base}(f)$ denote the total number of such triangles whose base edges participate in $f$. 
    
    \item The type-$2$ edges. Let $a_2(f)$ denote the total number of such edges on $f$. 
    
    \item The type-$1$ edges whose supported cross triangle are drawn inside $f$. This side of any type-$1$ edge is referred to as the {\em occupied} side. Let $a_1^{occ}(f)$ denote the total number of such edges in the boundary of $f$. 
    \item The type-$1$ edges whose supported cross triangle is drawn in $G$ in some region bounded by a super-face other than $f$. This side of any type-$1$ edge which does not support a cross triangle is referred to as the {\em free} side. We denote the number of such edges by $a_1^{free}(f)$. 
\end{itemize}

    \begin{figure}[H]
        \centering
      \includegraphics[width=0.6\textwidth]{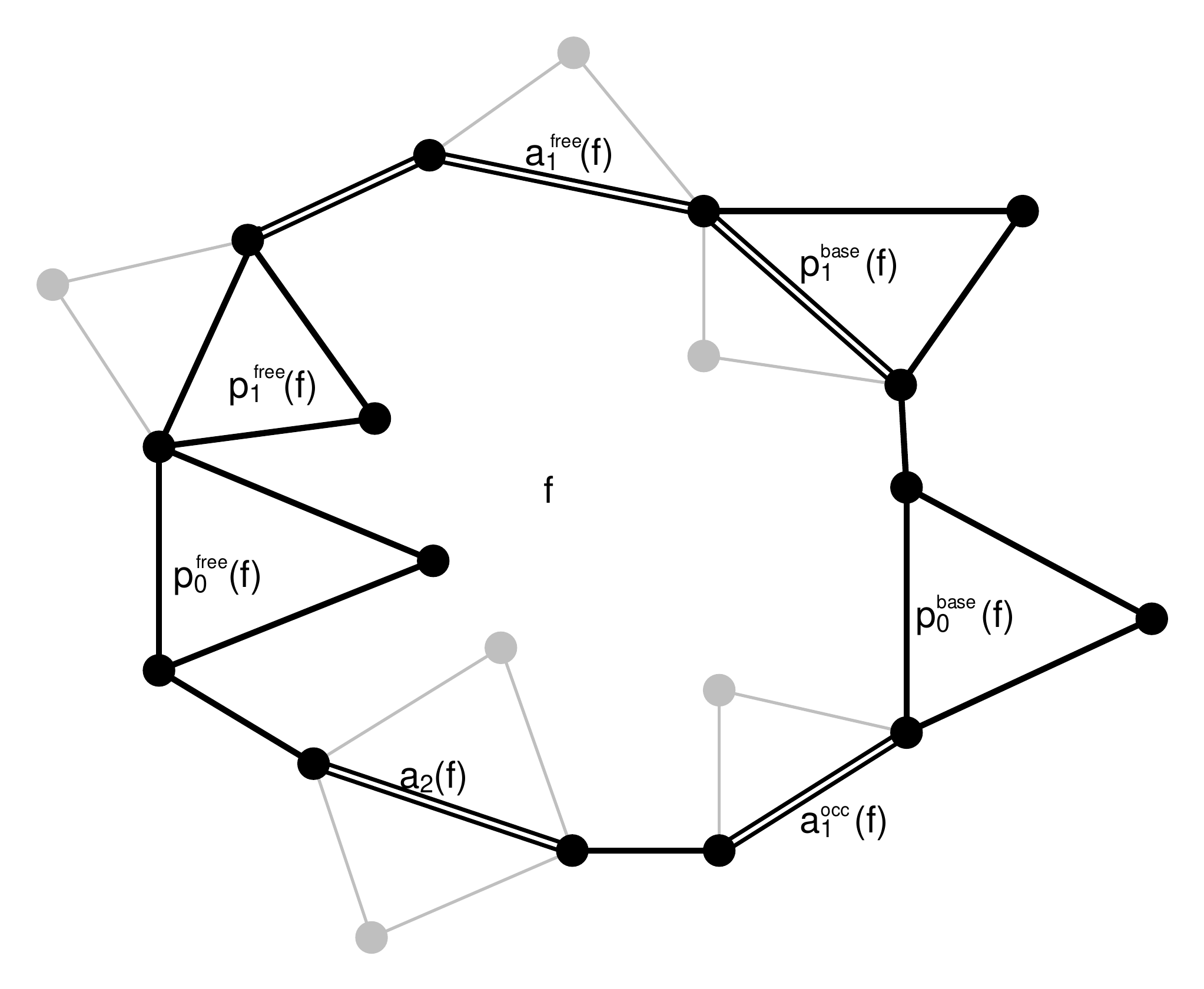}
        \caption{Different type of edges in a super-face $f \in \fset$, which corresponds to the region bounded by single (cactus edges) and double-lined (type-$1$ or type-$2$ edges) black edges (ignoring the gray edges which are the cross edges). For each type we indicate to which number they contribute. }
        \label{fig:classification-face-edges}
    \end{figure}

Notice that $|\fset| > 1$, since all cactus triangles in $G[S]$ are heavy, hence $a_1 + a_2 \geq 1$.
Since $\cset$ is a triangular cactus and $|\fset|>1$, the following can be observed.

\begin{Observation}
\label{obs:f_a1_a2}
For any super-face $f \in \fset$, $a_2(f) + a_1(f) \geq 1$.
\end{Observation}

Let $p^{free}(f) := p_0^{free}(f) + p_1^{free}(f)$, $p^{base}(f) := p_0^{base}(f) + p_1^{base}(f)$ and $a_1(f) := a_1^{occ}(f) + a_1^{free}(f)$. 

\begin{Observation}
\label{rem:surviving}
Any surviving triangular face cannot be incident to any type-$2$ edge, the occupied side of a type-$1$ edge or the base side of a type-$1$ triangle.
\end{Observation}

By Observation \ref{rem:surviving} and \ref{obs:f_a1_a2}, $|E(f)| = 2p^{free}(f) + p^{base}(f) + a_2(f) + a_1(f)$. Also, $|Free(f)| =  2 p^{free}(f) + a_1^{free}(f) + p_0^{base}(f)$ and $|Occ(f)|= a_1^{occ}(f) + a_2(f) + p_1^{base}(f)$.

\subsection{Classification Rules}
\label{sec:factor-7-rules} 
Now we are ready to define the classification rules for our analysis. 
Since the bound on the number of surviving triangles (hence the $gain(f)$ quantity) that can be drawn inside each super-face heavily depends on the type of edges contained in its face boundary, we classify each super-face $f \in \fset$ (except the outer-face $f_0$) into three broad categories, based on the total number of $p_1^{base}(f) + a_2(f) + a_1(f)$ edges. We also sub-categorize each super-face $f \in \fset$ for which $p_1^{base}(f) + a_2(f) + a_1(f) = 1$ into further classes, based on whether it contains an $a^{free}_1(f)$ edge or not.

\paragraph{Classifications of super-faces:} 
 A super-face $f$ will be of type-$[i, j]$ if $p_1^{base}(f) + a_2(f) + a_1(f) = i$ and $a^{free}_1(f) = j$. If there is no restriction on some dimension, then we put a dot ($[\bdot]$) there. Following is the precise categorization for the super-faces in $\fset \setminus \{f_0\}$.

\begin{itemize}
\item A super-face $f$ is of type-$[1, \bdot]$), if $p_1^{base}(f) + a_2(f) + a_1(f) = 1$ in addition
\begin{itemize}
    \item $f$ is of type-$[1, 0]$ , if $a^{free}_1(f) = 0$ or
    \item of type-$[1, 1]$), if $a^{free}_1(f) = 1$.
\end{itemize}
\item A super-face $f$ is of type-$[2, \bdot]$, if $p_1^{base}(f) + a_2(f) + a_1(f) = 2$
\item A super-face $f$ is of type-$[\geq3, \bdot]$, if $p_1^{free}(f) + a_2(f) + a_1(f) \geq 3$
\end{itemize}
Let the set $\fset[i, j]$ $\subseteq \fset$ be the subset of type-$[i, j]$ super-faces in $H[S]$ and analogously let $\eta[i, j] = |\fset[i, j]|$ for each type-$[i, j]$ super-face.
Notice that $\fset[1, \bdot] \cup \fset[2, \bdot] \cup \fset[\geq3, \bdot] \cup \{f_0\}= \fset$ and $\fset[i, \bdot] \cap \fset[j, \bdot]$ for any $i \neq j$, which implies, $|\fset| = 1+ \eta[1, \bdot] + \eta[2, \bdot] + \eta[\geq3, \bdot]$. Also, $\fset[1, j] \subseteq \fset[1, \bdot] $ for any $j \in \{0, 1\}$, hence, $\eta[1, \bdot] = \eta[1, 0] + \eta[1, 1]$.

The following lemma (whose proof will appear in Section~\ref{sec:generic-face}) gives lower bounds on the quantity $gain(f)$ for each type of super-faces in $\fset \setminus f_0$. For $f_0$ we will use Lemma~\ref{lem:outer-face} (whose proof will appear in Section~\ref{sec:outer-face}).

\begin{lemma}
\label{lem:all-bounds-factor-7} 
For any super-face $f \in \fset$, the following holds: 
\begin{enumerate}
\item If $f$ is of type-$[1, 0]$), then $gain(f) \geq 2.5$. 
\item If $f$ is of type-$[1, 1]$), then $gain(f) \geq 2$. 
\item If $f$ is of type-$[2, \bdot]$), then $gain(f) \geq 2$. 
\item If $f$ is of type-$[\geq3, \bdot]$), then $gain(f) \geq 1.5$. 
\end{enumerate}
\end{lemma}

\subsection{Proof for Lemma~\ref{lem:gain-factor-7}}
\label{sec:proof-factor-7-gain-lemma} 

Notice that the bounds for type-$[1,\bdot]$ and type-$[2,\bdot]$ are better than the trivial bound of $1.5$, which leads to the improvement from $7.5$ to $7$.

We apply Lemma~\ref{lem:outer-face} and Lemma~\ref{lem:all-bounds-factor-7} to $\sum_{f \in \fset} gain(f)$, depending on the type of each super-face: In particular, this includes the lower bounds for each super-face of type-$[1, 0]$, type-$[1, 1]$, type-$[2, \bdot]$, type-$[\geq 3]$ and the outer-face $f_0$. 

\begin{align*}
-(\sum_{f} gain(f)) &\leq  (1 - \phi(S))- \sum_{f \in \fset[1, 0]} 2.5 - \sum_{f \in \fset[1, 1]} 2 - \sum_{f \in \fset[2, \bdot]} 2 - \sum_{f \in \fset[\geq3, \bdot]} 1.5  \\ 
& =  1 - \phi(S) - 2.5 \eta[1, 0] - 2 \eta[1, 1] - 2 \eta[2, \bdot] - 1.5 \eta[\geq3, \bdot] \\ 
\end{align*}
Here we use the fact that $|\fset| = \eta[1,\bdot] + \eta[2,\bdot] + \eta[3,\bdot]+ 1$. 
\begin{align}\label{form:sum over super-faces}
-(\sum_{f} gain(f)) &\leq 1 - \phi(S) - 2.5 (|\fset|-1) + \fbox{$ 0.5 \eta[1, 1]  + 0.5 \eta[2, \bdot] +\eta[\geq 3, \bdot]$} \nonumber \\ 
& = 3.5 - \phi(S) -2.5 |\fset| + \fbox{$ 0.5 \eta[1, 1]  + 0.5 \eta[2, \bdot] +\eta[\geq 3, \bdot]$} 
\end{align}

Next, we deal with the ``residual terms'' highlighted in the formula above by the box. For this purpose, we present various upper bounds on the number of super-faces of certain type: 

\begin{lemma}[Two upper bounds on the number of super-faces]
\label{lem:two-bounds-factor-7}  
The following upper bounds hold: 

\begin{enumerate}
    \item $\eta[1,1] \leq a_1$. 
    \item $\eta[2, \bdot] + 2 \eta[\geq 3, \bdot] \leq p_1 + |\fset| -2$.  
\end{enumerate}
\end{lemma}
\begin{proof}
We start by proving the first upper bound. Since $a_1^{free}(f)=1$ for a type-$[1,1]$ super-face $f$ and each type-$1$ edge can contribute to $a_1^{free}(f)$ to exactly one super-face in $\fset$, we have that $\eta[1,1] \leq a_1$.

The second upper bound can be proved by a simple charging argument. 
On each super-face $f \in \fset$, we give $1$ unit of money to a certain set of edges on the face. 
In particular, each of the following types of edges gets a unit: (i) base of the type-$1$ cactus triangle, (ii) type-$1$ edge, and (iii) type-$2$ edge. 
Therefore, the total amount of money put into the system is exactly: 
$$\sum_{f \in \fset} (p_1^{base}(f) + a_1(f) + a_2(f)) = p_1 + 2 a_1 + 2a_2 = p_1 + 2 |\fset| -2$$
Counting from a different viewpoint, each super-face of type-$[j,\bdot]$ receives at least $j$ units of money, so the total amount is at least $1+\eta[1,\bdot] + 2\eta[2,\bdot] + 3\eta[\geq 3,\bdot] = |\fset| + \eta[2, \bdot] + 2 \eta[\geq 3, \bdot]$. 
This immediately implies the inequality: 
\[|\fset| + \eta[2,\bdot] + 2 \eta[\geq 3, \bdot] \leq p_1 + 2 |\fset| -2\] 

\end{proof}

Applying Lemma~\ref{lem:two-bounds-factor-7} to Inequality~(\ref{form:sum over super-faces}), we get that 
\begin{align} \label{form:sum over super-faces 2}
-(\sum_{f} gain(f)) &\leq 3.5 - \phi(S)  -2.5 |\fset| + 0.5 (a_1 + p_1 + |\fset| - 2) \nonumber \\ 
 & = 2.5 - \phi(S) - \fbox{$2 |\fset|$} + 0.5a_1 + 0.5 p_1
\end{align}

Using equality $|\fset| = a_1 + a_2 + 1$ in Inequality~(\ref{form:sum over super-faces 2}), we get:
\begin{align} \label{form:sum over super-faces 3}
-(\sum_{f} gain(f)) &\leq 2.5 - \phi(S)  -2 a_2 - 1.5a_1 - 2 + 0.5 p_1 = 0.5 - \phi(S) + \fbox{$a_1 + a_2 $} -2.5 a_1 - 3a_2 + 0.5 p_1
\end{align}

Using Lemma~\ref{lem:skeleton-faces} in Inequality~(\ref{form:sum over super-faces 3}), we get:
\[-(\sum_{f} gain(f)) \leq  0.5 - \phi(S) + 2p - 2 - 2.5 a_1 - 3a_2 + 0.5 p_1 = - \phi(S) + (2p + 0.5p_1 - 2.4 a_1 - 3 a_2 - 1.5)\] 

\subsection{Handling the Non-Outer-Faces (Proof of Lemma~\ref{lem:all-bounds-factor-7})} \label{sec:generic-face}

We split the proof of Lemma~\ref{lem:all-bounds-factor-7} into three parts. First we show an upper bound for the number of surviving triangles if a super-face $f$ has $|E(f)| >3$ or $a_1^{free}(f) + p_0^{base}(f) > 0$. 
Then we show that $survive(f) \leq \mu(f) - 1.5$, if $|E(f)| = 3$ and $a_1^{free}(f) + p_0^{base}(f) = 0$.
Finally we combine both results to give the upper bound for the number of surviving triangles in each type of super-face in $\fset$.

\begin{lemma}
\label{lem:sf-trivial-longface}
Let $f \in \fset$, if $|E(f)| >3$ or $a_1^{free}(f) + p_0^{base}(f) > 0$ we have
\[survive(f) \leq |Free(f)| + \floor{\frac{|Occ(f)|}{2}} - 2 \]
\end{lemma}
\begin{proof}
If $|E(f)| = 3$ and $a_1^{free}(f) + p_0^{base}(f) \geq 1$, it is easy to enumerate all possible compositions of the face boundary of $f$ and check for each case that the claimed bound holds.
\begin{itemize}
    \item ($a_1^{free}(f) + p_0^{base}(f) = 1$:) In this case, $survive(f) = 0$, $|Free(f)| =1$ and $|Occ(f)| =2$. 
    
    \item ($a_1^{free}(f) + p_0^{base}(f) = 2$:) In this case, $survive(f) = 0$, and $|Free(f)| =2$.  
    
    \item ($a_1^{free}(f) + p_0^{base}(f) = 3$:) In this case, $survive(f)= 1$, $|Free(f)| = 3$, and $|Occ(f)| = 0$. 
\end{itemize}

Now consider the case where $|E(f)| > 3$. In order to bound $survive(f)$ in this case, we {\em locally} modify the internal structure for a fixed $f$ in a special way. Notice that we make these modifications only for counting purposes and they do not change the structure of our graph $G$ in any way. First we {\em decouple} the supported cross triangles drawn inside $f$ which share their landing components by adding a dummy landing vertex for each such cross triangle and making the new dummy vertex its landing component. Then using additional type-$0$ edges we triangulate the super-face $f$ in an arbitrary way.
Note that the decoupling step allows us to get a full triangulation for $f$ and at the same time this operation does not reduce the value of $survive(f)$ for $f$ (see Figure~\ref{fig:decoupling-triangulation-f} for illustration). Hence, any bound which we get after performing this operation also holds for the original quantity $survive(f)$. 
This triangulation of the super-face $f$ has exactly $|E(f)| - 2$ triangular faces. Starting with this bound, we use the particular structure of $f$ to achieve the desired bound for $survive(f)$.

    \begin{figure}[H]
        \centering
        \includegraphics[width=0.9\textwidth]{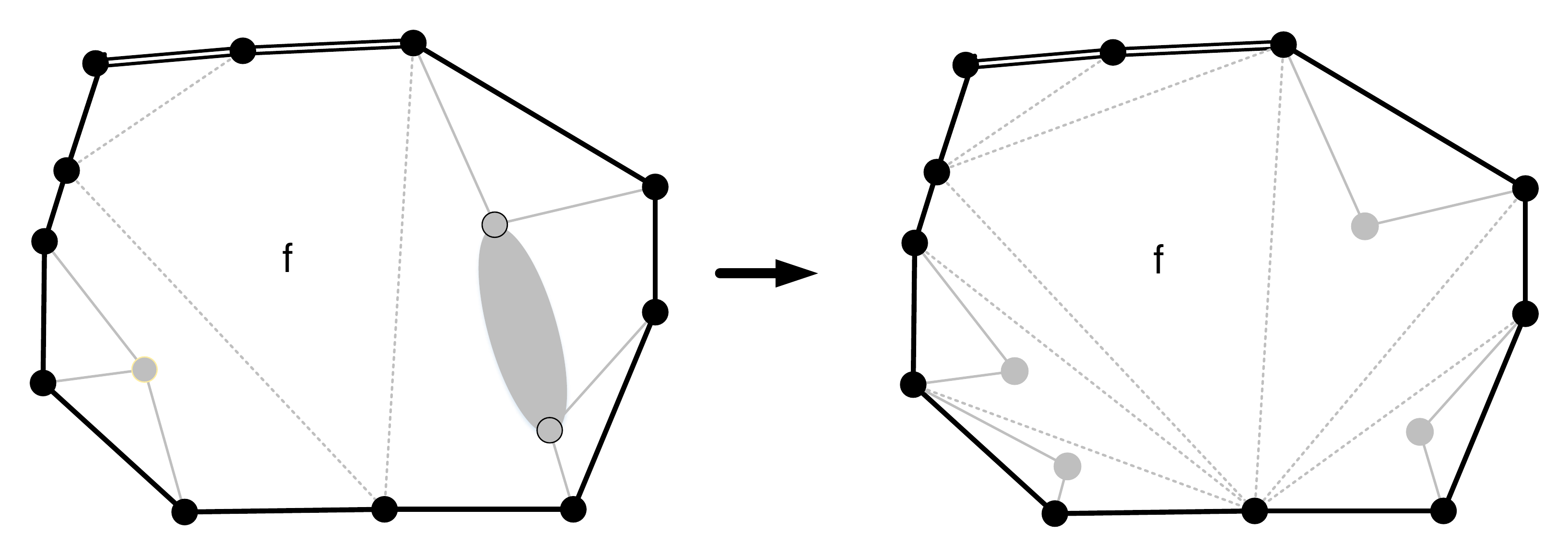}
        \caption{The decoupling and triangulation operation for some super-face $f \in \fset$. Notice that we make these modifications only for counting purposes and they do not change the structure of our graph $G$ in any way.}
        \label{fig:decoupling-triangulation-f}
    \end{figure}

By Observation~\ref{rem:surviving} no edge of type-$2$, occupied side of a type-$1$ edge or base side of a type-$1$ triangle can be adjacent to any triangular face in $survive(f)$. Also, at most two of these edges could belong to any triangular face in $f$. Hence, out of all the potential $|E(f)| - 2$ faces in the triangulate super-face $f$, at least $\ceil{\frac{|Occ(f)|}{2}}$ faces will be killed and hence we get the claimed bound on $survive(f)$.
\end{proof}

For other cases, we can still get a slightly weaker bound. 

\begin{lemma}
\label{lem:sf-trivial-shortface}
Otherwise, if $|E(f)| = 3$ and $a_1^{free}(f) + p_0^{base}(f) = 0$, then we have 
\[survive(f) \leq \mu(f) - 1.5\]
\end{lemma}
\begin{proof}
Notice that $|E(f)| = 3$ implies $p^{free}(f) = 0$. 
Hence the first inequality is trivially true by substituting the value $a_1^{occ}(f) + a_2(f) + p_1^{base}(f) = 3$ and $2 p^{free}(f) + a_1^{free}(f) + p_0^{base}(f) = 0$. 
\end{proof}

Now, we are ready to complete the proof of Lemma~\ref{lem:all-bounds-factor-7}. 
For any type-$[1, 0]$ super-face $f$, $|Occ(f)| = 1$ and $|E(f)| > 3$, hence using Lemma~\ref{lem:sf-trivial-longface}, we get
\[survive(f) \leq |Free(f)| + \floor{\frac{|Occ(f)|}{2}} - 2 = |Free(f)| + \frac{|Occ(f)|}{2} - 2.5 = \mu(f) - 2.5\]
For any type-$[1, 1]$ or type-$[2, \bdot]$ super-face $f$ we have that $|E(f)| > 3$, hence by Lemma~\ref{lem:sf-trivial-longface}, we get
\[survive(f) \leq |Free(f)| + \floor{\frac{|Occ(f)|}{2}} - 2 \leq |Free(f)| + \frac{|Occ(f)|}{2} - 2 = \mu(f) - 2\]
For any type-$[\geq3, \bdot]$ super-face $f$, if $|Occ(f)| = |E(f)| = 3$, then Lemma~\ref{lem:sf-trivial-shortface} implies
\[survive(f) \leq |Free(f)| + \floor{\frac{|Occ(f)|}{2}} - 1.5 \leq |Free(f)| + \frac{|Occ(f)|}{2} - 1.5 = \mu(f) - 1.5\]
Otherwise using Lemma~\ref{lem:sf-trivial-longface} we get
\[survive(f) \leq |Free(f)| + \floor{\frac{|Occ(f)|}{2}} - 2 \leq |Free(f)| + \frac{|Occ(f)|}{2} - 2 = \mu(f) - 2\]

\subsection{Handling the Outer-Face $f_0$ (Proof of Lemma~\ref{lem:outer-face})} 
\label{sec:outer-face}

In this section we will prove that $survive(f_0) \leq \mu(f_0) - g(S) + 1$. 
If $\phi(S) \leq 3$, this bound can be easily achieved by enumerating all possible compositions of the face boundary of $f_0$. If $\phi(S) > 3$, the $\phi(S)$ term in the bound we want to prove becomes more significant and hence this case needs special treatment.

In contrast to the other super-faces in $\fset$, the number of surviving triangles in $f_0$ also depends on $\phi(S)$. We first give an intuition on how this term influences the number of surviving triangles in $f_0$ and then use the idea behind it to prove Lemma~\ref{lem:outer-face}.
Starting from $G[S]$, we can construct an auxiliary graph $\widetilde{G}$ by modifying the outer-face $f_S$, such that this part of the graph is fully triangulated 
using type-$0$ edges, such that in total we obtain $\phi(S) - 2$ extra triangles.
Also, in this process the structure of the free and occupied edges of the outer-face (say $\widetilde{f}_0$) of the subgraph $\widetilde{H} := \widetilde{G} \setminus A$ (where $A$ is the set of type-$0$ edges) of $\widetilde{G}$ remains exactly the same as that of the original outer-face $f_0$ of  $H[S]$. 
Finally, we use the trivial upper bound given by Lemma~\ref{lem:sf-trivial-longface} on the number of triangular faces drawn inside the outer-face $\widetilde{f}_0$ in graph $\widetilde{G}$, which in turn gives us the $-\phi(S)$ term for the bound on the number of triangular faces drawn inside the outer-face $f_0$ in graph $G[S]$.
Notice that the modified graph $\widetilde{G}$ is created only for counting purposes and the modification does not change the structure of our original graph $G$ in any ways.

The following lemma formalizes this idea of triangulating the outer-face.

\begin{lemma}
\label{lem:gs_triangulation}
For the graph $G[S]$ with outer-face $f_S$ having $\phi(S) > 3$ free edges, there exists another simple planar graph $\widetilde{G}$ with outer-face $\widetilde{f_S}$, such that 
\begin{itemize}
    \item The graphs $\widetilde{G}$ and $G$ only differ inside the outer-face $f_S$ of $G[S]$.
    \item The structure of the outer-face $\widetilde{f}_0$ for the graph $\widetilde{H} := \widetilde{G} \setminus A$ (where $A$ is the set of type-$0$ edges) is the same as that of $f_0$, \ie $|Occ(f_0)| = |Occ(\widetilde{f}_0)|$ and $|Free(f_0)| = |Free(\widetilde{f}_0)|$.
    \item There are at least $\phi(S) - 2$ extra surviving triangles drawn inside the outer-face $\widetilde{f}_0$ in $\widetilde{G}$ as compared to the outer-face $f_0$ in $G[S]$.
\end{itemize}
\end{lemma}

\begin{proof}
In order to prove this lemma, we will transform $G[S]$ to $\widetilde{G}$ by creating at least $\phi(S) - 2$ new surviving triangles in $f_S$ by first pre-processing and then triangulating $f_S$ using extra type-$0$ edges in a specific way.

First we {\em decouple} the supported cross triangles drawn inside $f_S$ which share their landing components by adding a dummy landing vertex for each such cross triangle and making the new dummy vertex its landing component. 
Notice that the decoupling step makes the induced graph $G[V(f_S)]$ an outer-planar graph, where $V(f_S)$ are the vertices contained in face $f_S$. 
Also, it does not change the structure of the graph $G$ anywhere else except inside face $f_S$. 
Since $G[V(f_S)]$ is outer-planar, there exists a vertex $u_1 \in V(f_S)$, such that the degree of $u_1$ in $G[V(f_S)]$ is two. Now we number the vertices in the face $f_S$ in clockwise order as $u_1, u_2, \dots u_{\ell_S}$, where $u_1$ is the degree $2$ vertex in $G[V(f_S)]$. Next we triangulate the outer-face $f_S$ by adding a star of type-$0$ edges with vertex $u_1$ as the root for this star and vertices $u_3, u_4 \dots u_{\ell_S - 1}$ as the leaves of the star (see Figure~\ref{fig:f_S-triangulation}). 
This completes the construction of our auxiliary graph $\widetilde{G}$. Notice that this operation cannot create a parallel edge in $\widetilde{G}$, implied by the way we fixed $u_1$. Also, the decoupling and triangulation will maintain the planarity of $\widetilde{G}$. Finally, it is easy to see that the occupied and the free edges for the outer-face $\widetilde{f}_0$ of graph $\widetilde{H}$ are the same as that of the original outer-face $f_0$, hence the second property is satisfied.

    \begin{figure}[H]
        \centering
        \includegraphics[width=0.9\textwidth]{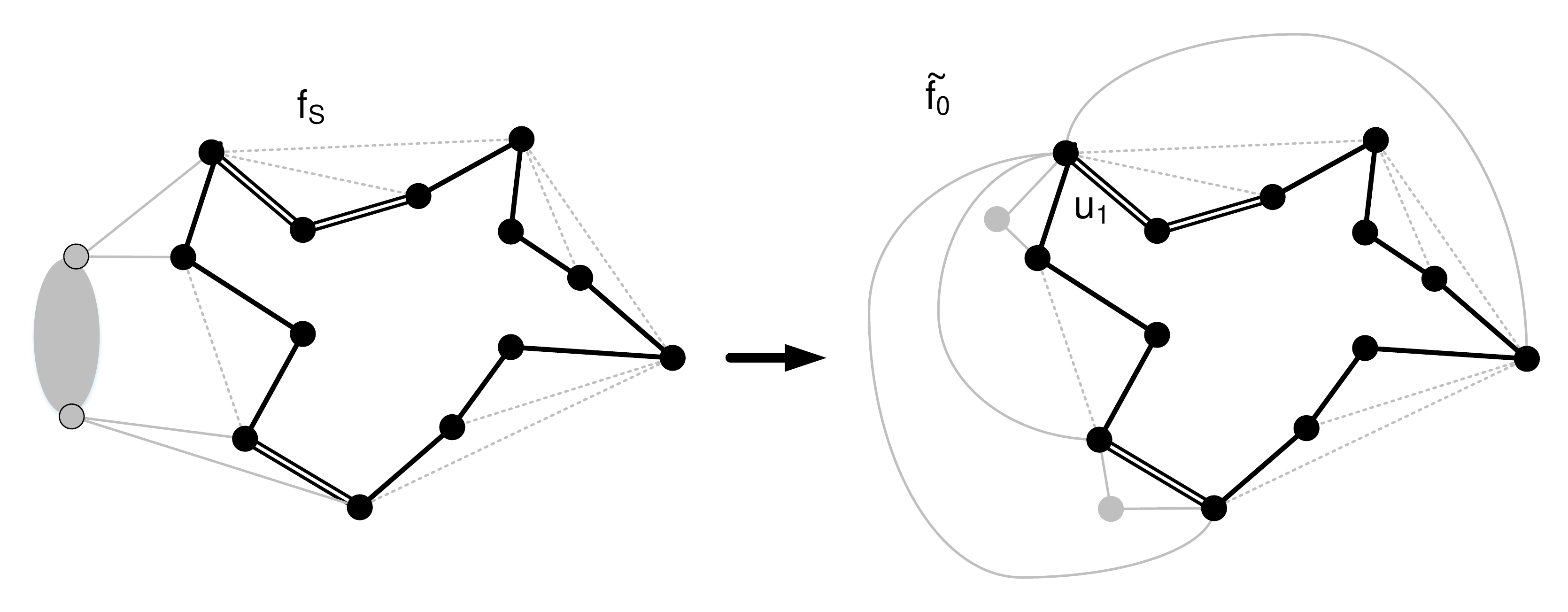}
        \caption{The decoupling and triangulation of the face $f_S$. On the left $f_S$ is identical to the outer face of the drawn graph when ignoring the gray solid (cross) edges and component. On the right $\widetilde{f}_0$ is formed by ignoring all gray (dotted (type-$0$) and solid (cross)) edges.}
        \label{fig:f_S-triangulation}
    \end{figure}

Each of the triangles $(u_1, u_2, u_3)$ and $(u_{\ell_S - 1}, u_{\ell_S}, u_1)$ could either survive if both the edges coming from $f_S$ are free or not survive if at least one of these edges is occupied. Any triangle of the form $(u_1, u_i, u_{i+1})$ for $2<i<\ell_S -1$ will survive if the $(u_i, u_{i+1})$ edge is free. Now if both the triangles $(u_1, u_2, u_3)$ and $(u_{\ell_S - 1}, u_{\ell_S}, u_1)$ do not survive, then at most two out of the $\phi(S)$ free edges can be a part of these triangles and hence there will be at least $\phi(S) - 2$ triangles of the form $(u_1, u_i, u_{i+1})$ for $2<i<\ell_S - 1$ which survive. 
If one of the triangles $(u_1, u_2, u_3)$ and $(u_{\ell_S - 1}, u_{\ell_S}, u_1)$ survives, then at most three out of the $\phi(S)$ free edges can be part of these triangles and hence there will be at least $\phi(S)-3$ triangles of the form $(u_1, u_i, u_{i+1})$ for $2 < i < \ell_S - 1$ which survive. Else both of the $(u_1, u_2, u_3)$ and $(u_{\ell_S - 1}, u_{\ell_S}, u_1)$ triangles survive, then four out of the $\phi(S)$ free edges will be part of these triangles and hence there will be at least $\phi(S)-4$ triangles of the form $(u_1, u_i, u_{i+1})$ for $2 < i < \ell_S - 1$ which survive. Hence, overall in each case, $\phi(S) - 2$ triangles survive and the lemma follows.
\end{proof}

Note that $\phi(S)$ consists of a subset of the edges counted in $p^{free}(f_0), p_0^{base}(f_0), a_1^{free}(f_0)$ and $a_0(f_0)$. Also $|E(f_0)| \geq \ell_S \geq \phi(S)$, since $f_S$ is formed after including all the $a_0(f)$ edges drawn inside $f_0$ in $G$ (see Figure~\ref{fig:outer-faces-f_0-f_S}).

    \begin{figure}[H]
        \centering
        \includegraphics[width=0.9\textwidth]{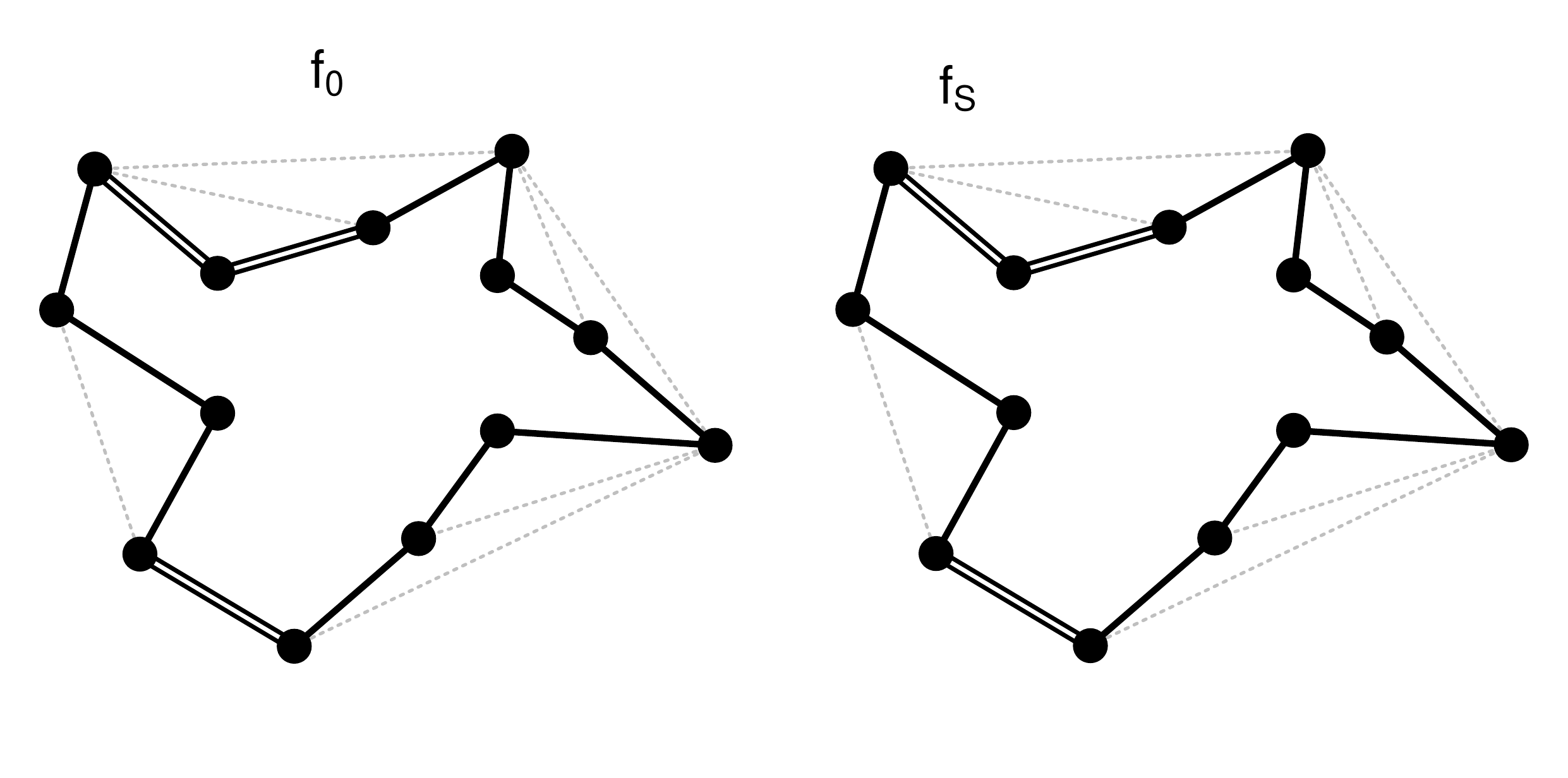}
        \caption{In the first figure the outer face boundary resulting from deleting all doted gray (type-$0$) edges corresponds to the outer-super-face $f_0$ of graph
$H[S]$. In the second figure the outer face corresponds to the outer-face $f_S$ of graph $G[S]$.}
        \label{fig:outer-faces-f_0-f_S}
    \end{figure}

Now, we are ready to present the proof of Lemma~\ref{lem:outer-face}. 
We split the analysis into two cases: 
\begin{itemize}
    \item First, consider the case when $|E(f_0)| = 3$. The worst case then is when $\phi(S)=3$, which implies $|Free(f_0)| = 3$, $|Occ(f_0)| = 0$ and $\mu(f_0) = 3$.  In this case, $survive(f_0) = 1$, which gives the inequality.

Otherwise, when $\phi(S) \leq 2$, we have $survive(f_0)= 0$ (there would be an occupied edge that supports a cross triangle in $f_0$ which kills it), $|Free(f_0)| \leq 2$ and $|Occ(f_0)| \geq 1$. 
This gives $\mu(f_0) \geq 1.5$, and $\mu(f_0) - \phi(S) +1 \geq 0.5 > survive(f_0)$.

\item If $|E(f_0)| > 3$ and $\phi(S) \leq 3$, then the trivial bounds given by Lemma~\ref{lem:sf-trivial-longface} and –\ref{lem:sf-trivial-shortface} imply the inequality. 

From now on we assume that $\phi(S) > 3$. For this case, we use Lemma~\ref{lem:gs_triangulation} on $G[S]$ to get the auxiliary graph $\widetilde{G}$ with at least $\phi(S) - 2$ extra surviving faces in its outer-face, totaling to $survive(f_0) + \phi(S) - 2$. Now using the trivial bound given by Lemma~\ref{lem:sf-trivial-longface} on the outer-face $\widetilde{f}_0$ for the corresponding graph $\widetilde{H}$, we get
\[survive(f_0) + \phi(S) - 2 \leq survive(\widetilde{f}_0) \leq \mu(\widetilde{f}_0) - 2 \leq \mu(f_0) - 2\]

which proves the lemma.

\end{itemize}

\paragraph{Intuition for next step}
The tight example (see Appendix~\ref{subsec:tight-7}) of the factor $7$ analysis for a $1$-swap optimal solution shows that looking at a locally optimal solution for $1$-swap is not enough to achieve our main Theorem~\ref{thm:main}. In this example there exists an improving $2$-swaps, which indicates that further classification of super-faces in $H[S]$ and achieving stronger bounds for some special type of super-faces (which we refer to as type-$[1,0,0]$, type-$[1,1,0]$ and type-$[2,0,0]$ super-faces) could lead to an improvement. These special super-faces are the ones where an adversary can efficiently packs a lot of surviving triangles which leads to $q > 6p$. It turns out that in every such super-face, there is an improving-$2$ swap. Hence to prove the tight $q \leq 6p$ bound, it is necessary to get stronger bounds for such faces by looking at locally optimal solution for $2$-swaps. This intuition lead us to the sub-categorization of the type-$[1, \bdot]$ and type-$[2, \bdot]$ faces in $H[S]$.

\section{Classification Scheme for Factor $6$}
\label{sec:classification-6} 

We will show a classification scheme that certifies the factor $6$. 
This scheme extends the one given in the previous section. 

\subsection{Classification Rules}
\label{sec: factor 6, rules} 
The important observation that leads to a better bound is to derive a better gain for super-faces of type-$[1,\bdot]$ and type-$[2,\bdot]$ in the previous classification.
We notice that, for a certain sub-class of these super-faces, a better bound can be obtained.

\paragraph{A New Super-face Classification:}
Now we sub-categorize type-$[1, \bdot]$ and type-$[2, \bdot]$ super-faces into further classes, based on the values of $a^{free}_1(f)$ and $p^{base}_0(f)$. 
A super-face $f$ will be of type-$[i, j, k]$ if $p_1^{base}(f) + a_2(f) + a_1(f) = i$, $a^{free}_1(f) = j$ and $p^{base}_0(f) = k$. If there is no restriction on a particular dimension, then we put a dot ($[\bdot]$) there. Following is the categorization of super-faces which we use.
\begin{itemize}
\item type-$[1, \bdot, \bdot]$: $p_1^{base}(f) + a_2(f) + a_1(f) = 1$
\begin{itemize}
\item type-$[1, 0, \bdot]$: $a^{free}_1(f) = 0$
\begin{itemize}
\item type-$[1, 0, 0]$: $p^{base}_0(f) = 0$ 
\item type-$[1, 0, \geq1]$: $p^{base}_0(f) \geq 1$ 
\end{itemize}
\item type-$[1, 1, \bdot]$: $a^{free}_1(f) = 1$
\begin{itemize}
\item type-$[1, 1, 0]$: $p^{base}_0(f) = 0$ 
\item type-$[1, 1, \geq1]$: $p^{base}_0(f) \geq 1$ 
\end{itemize}
\end{itemize}

\item type-$[2, \bdot, \bdot]$: $p_1^{base}(f) + a_2(f) + a_1(f) = 2$
\begin{itemize}
\item type-$[2, 0, \bdot]$: $a^{free}_1(f) = 0$
\begin{itemize}
\item type-$[2, 0, 0]$: $p^{base}_0(f) = 0$ 
\item type-$[2, 0, \geq1]$: $p^{base}_0(f) \geq 1$ 
\end{itemize}

\item type-$[2, 1, \bdot]$: $a^{free}_1(f) = 1$ 

\item type-$[2, 2, \bdot]$: $a^{free}_1(f) = 2$ 
\end{itemize}

\item type-$[\geq3, \bdot, \bdot]$: $p_1^{base}(f) + a_2(f) + a_1(f) \geq 3$

\end{itemize}
Let the subset $\fset[i, j, k] \subseteq \fset$ be the set of type-$[i, j, k]$ super-faces and analogously let $\eta[i, j, k] = |\fset[i, j, k]|$. It is easy to see that the categorization partitions the set $\fset \setminus \{f_0\}$,  $\fset[i, j, k] \subseteq \fset[i, j, \bdot] \subseteq \fset[i, \bdot, \bdot] $ for any $i, j, k$, which implies, $|\fset| = 1 + \eta[1, \bdot, \bdot] + \eta[2, \bdot, \bdot] + \eta[\geq3, \bdot, \bdot]$. Also, $\eta[i, \bdot, \bdot] = \sum_j \eta[i, j, \bdot]$ for each $i$, $\eta[i, j, \bdot] = \sum_k \eta[i, j, k]$ for each $i,j$.

We classify a sub-class of type-$[1,0,0]$, type-$[1,1,0]$, and type-$[2,0,0]$ super-faces that admits an improved bound via several new notions. 

\paragraph{Adjacent triangles and edges and friends:}
\label{par:adj-fri}
Let $t_1$ and $t_2$ be two cactus triangles that share a vertex. Denote their vertices by $V(t_i) = \{u_i, v_i, w_i\}$, where $v_1 = v_2$ (say $v$). 
In this case, we call them {\em adjacent triangles}. Let $w_i$ be a free vertex of $t_i$. 
If there is a way to draw an edge $w_1 w_2$ such that the region bounded by $(v, w_1, w_2)$  is empty, we say that these triangles are {\em strongly adjacent}; otherwise, they are {\em weakly adjacent}.  
Furthermore, if the $t_1$ and $t_2$ are strongly adjacent in $H$ and $w_1 w_2 \in E(G[S])$, then we say that $t_1$ and $t_2$ are {\em friends} or {friendly triangles}.

    \begin{figure}[H]
        \centering

        \includegraphics[width=0.9\textwidth]{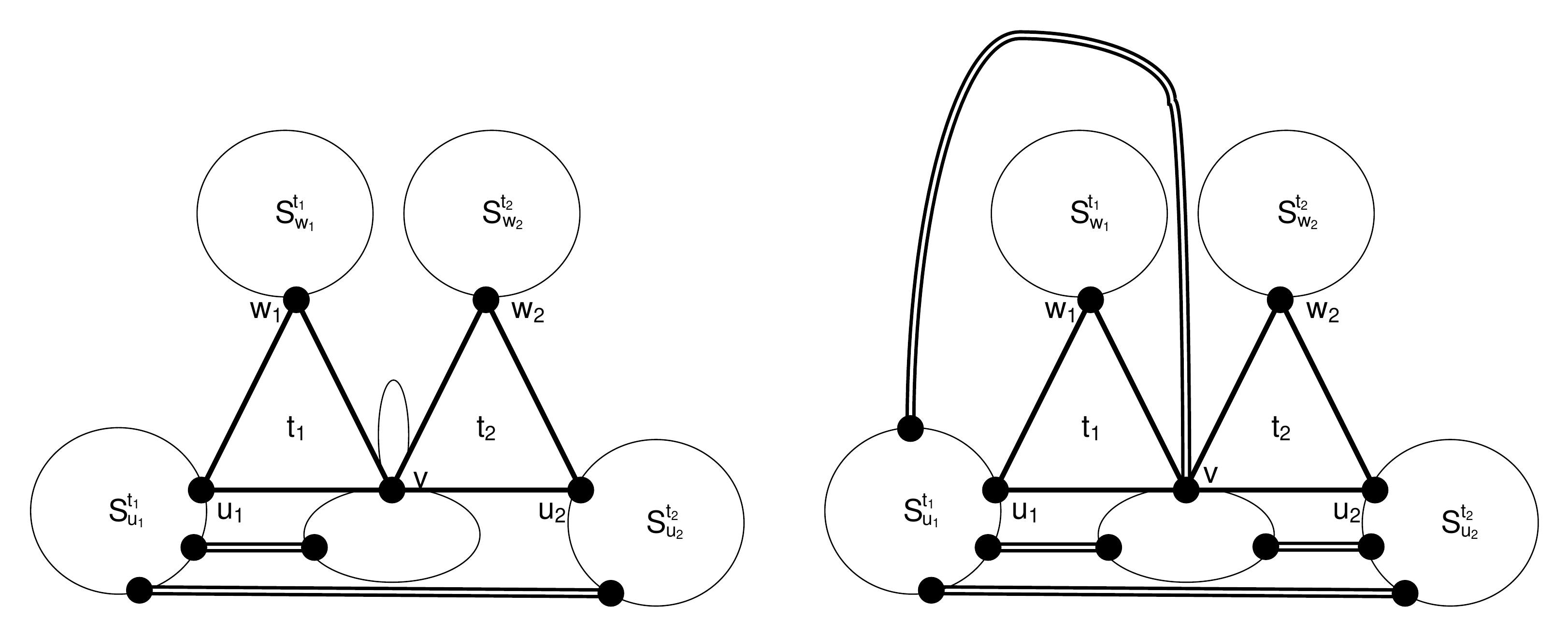}
        \caption{Cases when two adjacent triangles are {\em weakly}-adjacent in graph $H[S]$.}
        \label{fig:weakly-adjacent-triangles}
    \end{figure}

     \begin{figure}[H]
        \centering
        \includegraphics[width=0.5\textwidth]{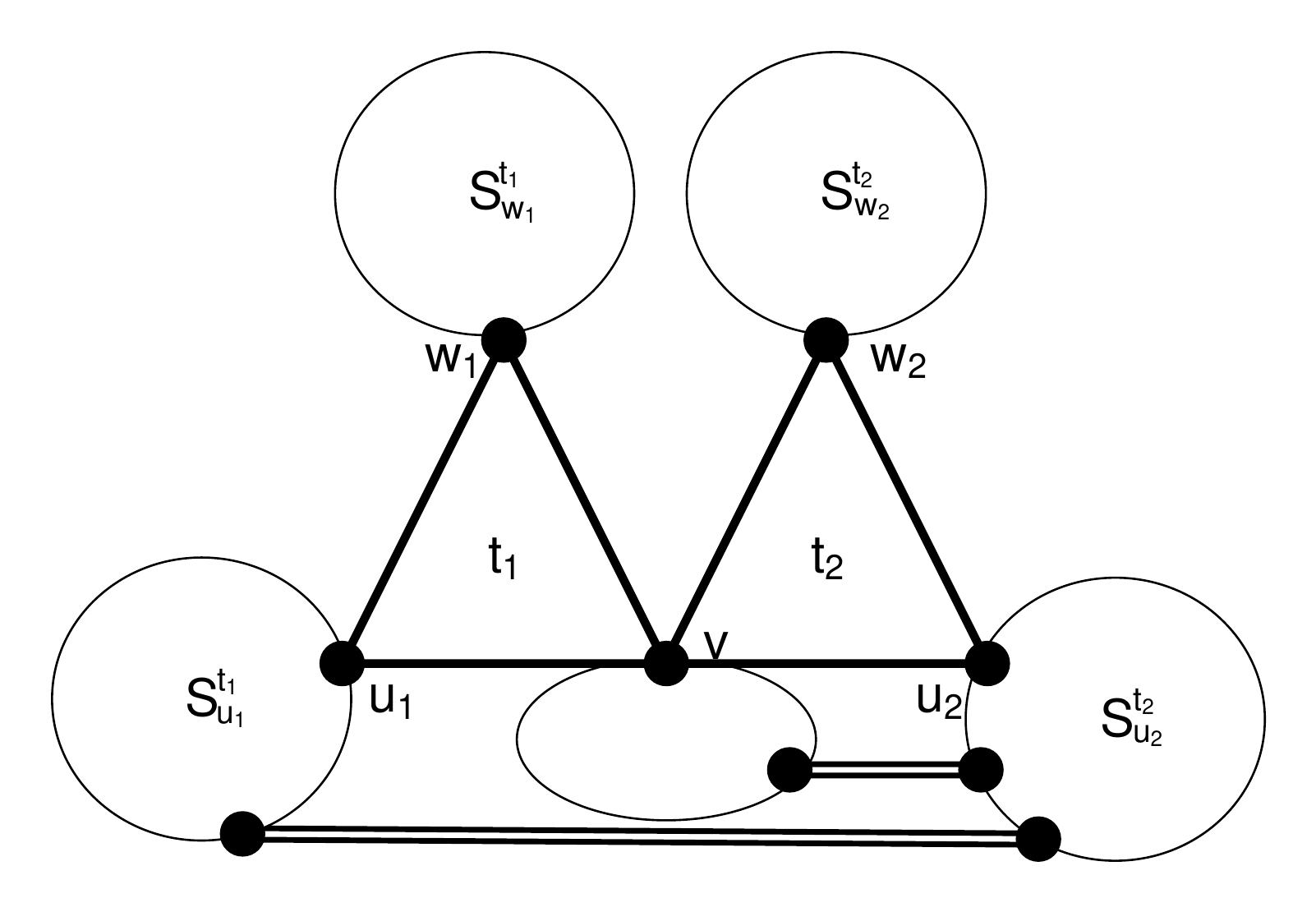}
        \caption{Two adjacent triangles which are {\em strongly}-adjacent in graph $H[S]$.}
        \label{fig:strongly-adjacent-triangles}
    \end{figure}

    \begin{figure}[H]
        \centering
        \includegraphics[width=0.5\textwidth]{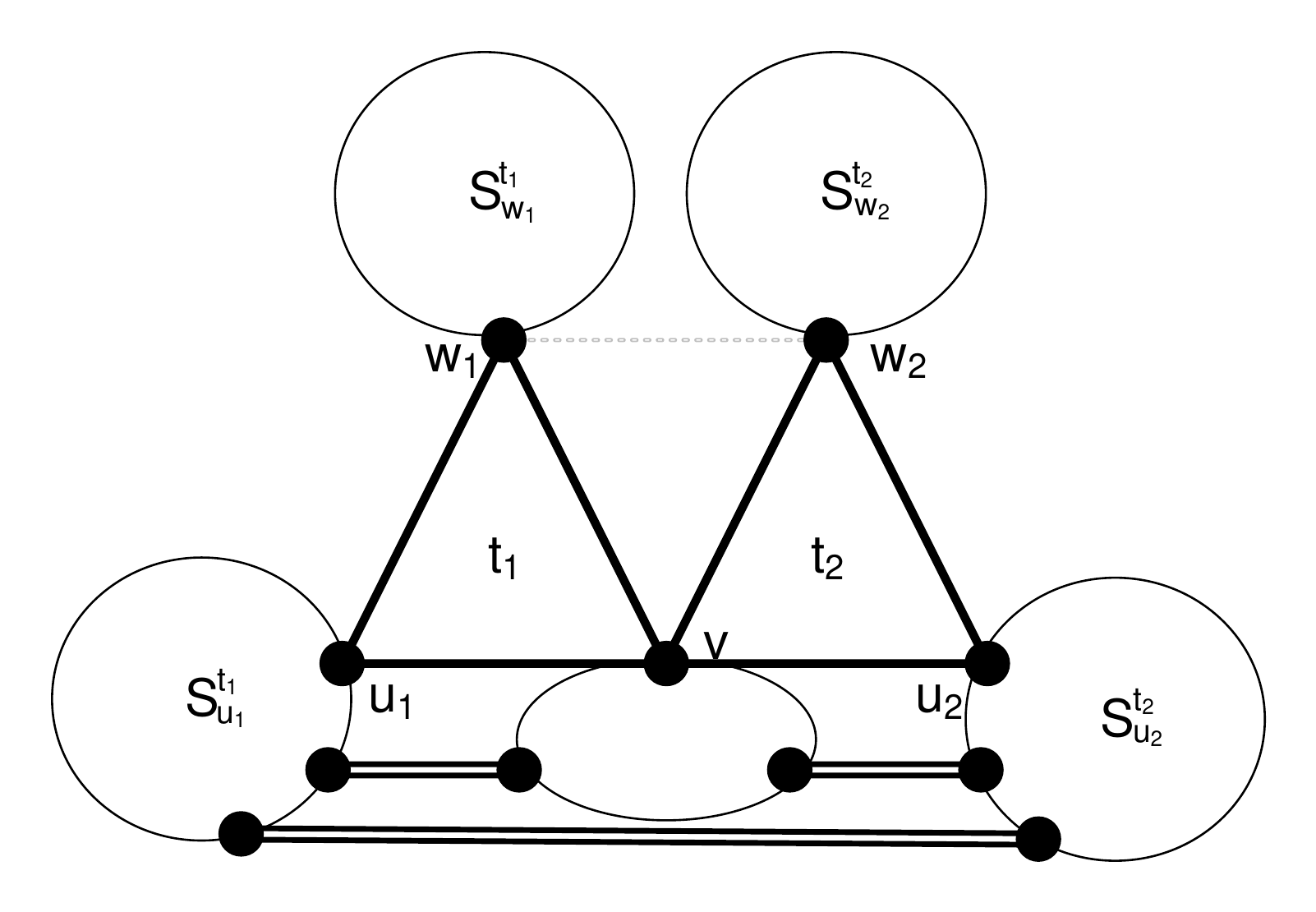}
        \caption{Two adjacent triangles which are {\em friends} in graph $G[S]$. Notice they will be {\em strongly}-adjacent in graph $H[S]$.}
        \label{fig:friendly-triangles}
    \end{figure}

\begin{Observation}
\label{Observation:adj-fri}
    The free sides for any pair of triangles which are strongly-adjacent or friends are part of the same super-face in $\fset$.
\end{Observation}


We will crucially rely on the following lemma, whose proof is provided later in Section~\ref{sec:friend-lemma} 

\begin{lemma}[Friend Lemma]
\label{lem:friend} 
The following properties hold: 
\begin{itemize}
    \item No type-$1$ heavy triangle is friends with any other heavy cactus triangle.
    \item For any pair of type-$0$ triangles which are friends, their corresponding base sides belong to a common super-face in $\fset$.
\end{itemize}
\end{lemma}

From the lemma, whenever we talk about friends, we always mean a pair of type-$0$ triangles. 

\paragraph{Friendly super-faces:} 
We call a super-face $f \in \fset$ of type-$[1, 0, 0], [1, 1, 0]$ or $[2, 0, 0]$ a {\em friendly} super-face if it contains at least one pair of cactus triangles that are friends. 
Let $\fset_{fri}[1, 0, 0] \subseteq \fset[1, 0, 0]$, $\fset_{fri}[1, 1, 0] \subseteq \fset[1, 1, 0]$ and $\fset_{fri}[2, 0, 0] \subseteq \fset[2, 0, 0]$ be the set of friendly super-faces of type-$[1, 0, 0], [1, 1, 0]$ and $[2, 0, 0]$ respectively. 
Also, let $\eta_{fri}[i, j, k] = |\fset_{fri}[i, j, k]|$. Let $\eta_{fri} = \eta_{fri}[1, 0, 0] + \eta_{fri}[1, 1, 0] + \eta_{fri}[2, 0, 0]$.



The subsequent lemmas (which we prove later) give us stronger bounds on $survive(f)$ for super-faces of type-$[1, 0, 0], [1, 1, 0]$ or $[2, 0, 0]$ which are not friendly.  

\begin{lemma}
\label{lem:f1_0}
For any type-$[1, 0, 0]$ super-face $f \in \fset[1, 0, 0] \setminus \fset_{fri}[1, 0, 0]$, the following bound holds for $gain(f)$.
\[gain(f) \geq 4.5\]
\end{lemma}

\begin{lemma}
\label{lem:f11_0}
For any type-$[1, 1, 0]$ super-face $f \in \fset[1, 1, 0] \setminus \fset_{fri}[1, 1, 0]$, the following bound holds for $survive(f)$.
\[gain(f) \geq 4\]
\end{lemma}

\begin{lemma}
\label{lem:f2_0}
For any type-$[2, 0, 0]$ super-face  $f \in \fset[2, 0, 0] \setminus \fset_{fri}[2, 0, 0]$, the following bound holds for $survive(f)$.
\[gain(f) \geq 3\]
\end{lemma}

Now, we have identified the set of super-faces for which we obtain an improved bound. 
The rest of the super-faces only relies on trivial upper bounds. 

\begin{lemma} 
\label{cor:trivial-6}
For any super-face $f \in \fset$, the respective bounds hold for $gain(f)$
\begin{itemize}
\item type-$[1, 0, \geq1]$: \[gain(f) \geq 2.5\]
\item ($f \in \fset_{fri}[1, 0, 0]$): \[gain(f) \geq  2.5\]

\item type-$[1, 1, \geq1]$: \[gain(f) \geq 2\]
\item ($f \in \fset_{fri}[1, 1, 0]$): \[gain(f) \geq 2\]

\item type-$[2, 0, \geq1]$: \[gain(f) \geq   2\]
\item ($f \in \fset_{fri}[2, 0, 0]$): \[gain(f) \geq   2\]

\item type-$[2, 1, \bdot]$:   \[gain(f) \geq   2.5\]

\item type-$[2, 2, \bdot]$:  \[gain(f) \geq   2\]

\item type-$[\geq3, \bdot, \bdot]$:  \[gain(f) \geq   1.5\]

\end{itemize}
\end{lemma} 

\begin{proof}
For any type-$[1, 0, \bdot]$ or type-$[2, 1, \bdot]$ super-face $f$, $|Occ(f)| = 1$ and $|E(f)| > 3$, hence using Lemma~\ref{lem:sf-trivial-longface}, we get
\[survive(f) \leq |Free(f)| + \floor{\frac{|Occ(f)|}{2}} - 2 = |Free(f)| + \frac{|Occ(f)|}{2} - 2.5 = \mu(f) - 2.5\]
For any type-$[1, 1, \bdot]$ or type-$[2, 0, \bdot]$ or type-$[2, 2, \bdot]$ super-face $f$, $|E(f)| > 3$, hence using Lemma~\ref{lem:sf-trivial-longface}, we get
\[survive(f) \leq |Free(f)| + \floor{\frac{|Occ(f)|}{2}} - 2 \leq |Free(f)| + \frac{|Occ(f)|}{2} - 2 = \mu(f) - 2\]
For any type-$[\geq3, \bdot, \bdot]$ super-face $f$, if $|Occ(f)| = |E(f)| = 3$, using Lemma~\ref{lem:sf-trivial-shortface}, we get
\[survive(f) \leq |Free(f)| + \floor{\frac{|Occ(f)|}{2}} - 1.5 \leq |Free(f)| + \frac{|Occ(f)|}{2} - 1.5 = \mu(f) - 1.5\]
Else, using Lemma~\ref{lem:sf-trivial-longface}, we get
\[survive(f) \leq |Free(f)| + \floor{\frac{|Occ(f)|}{2}} - 2 \leq |Free(f)| + \frac{|Occ(f)|}{2} - 2 = \mu(f) - 2\]
\end{proof}

\subsection{Valid Inequalities}
\label{sec: factor 6, eq} 
We present various upper bounds on the number of super-faces of certain type. 
We denote by $\Phi$ the following system of linear inequalities. 

\begin{lemma}[Various upper bounds on the number of super-faces]
\label{lem:various-bounds-factor-6} 
The following bounds hold: 
\begin{itemize}
    \item $\eta[2, \bdot, \bdot] + 2 \eta[\geq3, \bdot, \bdot] \leq p_1 + |\fset| -2$
    \item $\eta[1, 1, \bdot] + \eta[2, 1, \bdot] + 2 \eta[2, 2, \bdot]  \leq a_1$
    
    \item $\eta_{fri} + \eta[1, 0, \geq1] + \eta[1, 1, \geq1] + \eta[2, 0, \geq1] \leq p_0$
\end{itemize}
\end{lemma}
\begin{proof}
The first bound is derived in exactly the same manner as in Lemma~\ref{lem:two-bounds-factor-7}. 
The second bound is also similar. Consider the sum:  
\[\sum_{f \in \fset[1,1,\bdot] \cup \fset[2,1,\bdot] \cup \fset[2,2, \bdot] } a_1^{free}(f) \leq a_1\] 
Notice that each super-face of type-$[1,1,\bdot]$ or type-$[2,1,\bdot]$ gets the contribution of at least $1$, while the other type gets the contribution of $2$, so we have that the sum is at least $\eta[1,1,\bdot] + \eta[2,1,\bdot] + 2 \eta[2,2, \bdot]$. 

Finally, for the third bound, we give a combinatorial charging argument. First, we imagine giving $1$ unit of money to each type-$0$ triangle. 
Therefore, $p_0$ units of money are placed into the system. 
We will argue that we can ``transfer'' this amount such that each super-face in $\fset_{fri}[1, 0, 0] \cup \fset_{fri}[1, 1, 0] \cup \fset_{fri}[2, 0, 0] \cup \fset[1, 0, \geq1] \cup \fset[1, 1, \geq1] \cup \fset[2, 0, \geq1]$ receives at least one unit of money, hence establishing the desired bound. 

\begin{itemize}
    \item  For each face $f \in  \fset_{fri}[1, 0, 0] \cup \fset_{fri}[1, 1, 0] \cup \fset_{fri}[2, 0, 0] $, we know that there must be at least one pair of friends. By Lemma~\ref{lem:friend}, no type-$1$ triangle is friends with any other heavy cactus triangle. 
    The super-face $f$ receives $1$ unit of money from each such triangle in the pair, so we have $2$ units on each such super-face.

\item Now consider a super-face $f \in \fset[1, 0, \geq1] \cup \fset[1, 1, \geq1] \cup \fset[2, 0, \geq1]$. 
On such super-face, there is at least one type-$0$ triangle, and such cactus triangle would (i) pay super-face $f$ if it still has the money, or (ii) the ``extra'' money would be put in the system to pay $f$ if no cactus triangle in $f$ has money left with it. 
\end{itemize}

In the end, all such super-faces would have at least one or two units of money, so the total money in the system is at least $2\eta_{fri} + \eta[1, 0, \geq1] + \eta[1, 1, \geq1] + \eta[2, 0, \geq1]$. 
The total payment into the system is at most $p_0$ plus the extra money. 
There can be at most $\eta_{fri}$ units of extra money spent: Due to Lemma~\ref{lem:friend}, i.e. whenever a face contains a triangle that spent in the first step, it must also contain its pair of friends, so there can be at most $\eta_{fri}$ such faces that cause an extra spending. 
This reasoning implies that 
$$2\eta_{fri} + \eta[1, 0, \geq1] + \eta[1, 1, \geq1] + \eta[2, 0, \geq1] \leq p_0 + \eta_{fri}$$
\end{proof}

\paragraph{Deriving Factor 6:} Now that we have both the inequalities and the gain bounds, the following is an easy consequence (e.g. it can be verified by an LP solver.) 
For completeness, we produce a human-verifiable proof in Appendix~\ref{subsec:proof-factor-6}. 

\begin{lemma}
\label{lem: derive factor 6} 
$$q \leq 4p + 0.5 p_1 + 2.5 a_1 + 3 a_2 -\overrightarrow{gain} \cdot \vec \chi \leq 6p  - \phi(S)$$ 
\end{lemma}

\subsection{Gain analysis for other cases}

In this section, we analyze the gain for various types of faces where we get improved bounds. 

\subsubsection{$\fset[1,0,0] \setminus \fset_{fri}[1,0,0]$ Super-faces (Proof for  Lemma~\ref{lem:f1_0})}
A super-face in this set turns out to behave in a very structured way, i.e. the edges of the cactus triangles bounding this face look like a ``fence'', which is made precise below. 

\paragraph{Cactus fence:} A {\em cactus fence} of size $k$ is a maximal sequence of cactus triangles $(t_1,\ldots, t_k)$ such that any pair $t_i$ and $t_{i+1}$ are strongly adjacent. Moreover, for each triangle $t$, if $w \in V(t)$ is a free vertex of $t$, then $S^t_{w}$ is a singleton.

    \begin{figure}[H]
        \centering
        \includegraphics[width=0.4\textwidth]{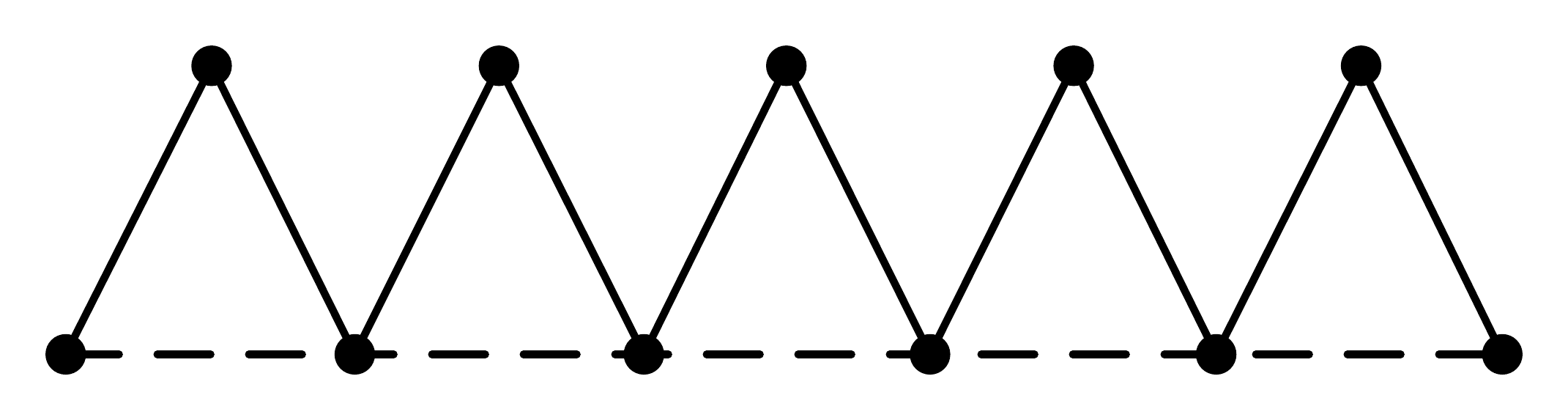}
        \caption{The cactus fence structure of size five.}
        \label{fig:fence-structure}
    \end{figure}

\begin{lemma}[Fence lemma]
\label{lem: fence 1} 
Any super-face $f \in \fset[1,0,0] \setminus \fset_{fri}[1,0,0]$ is bounded by free sides of a cactus fence together with one edge $e$ that is of type-$2$. 
\end{lemma}

The proof of this lemma is quite intricate and is deferred to the next subsection. 
Moreover, from definition of the set $\fset[1,0,0] \setminus \fset_{fri}[1,0,0]$, each pair of cactus triangles on this face is not a pair of friends. 
It suffices to show that $survive(f) \leq |E(f)| - 5$: 
Since $|Occ(f)| = 1$, this would imply  $survive(f) \leq |E(f)| - 5 = |Free(f)| + |Occ(f)| - 5 = |Free(f)| + |Occ(f)|/2 - 4.5 = \mu(f) - 4.5$ which proves the lemma.

For obtaining the bound on $survive(f)$, we obtain an auxiliary graph $H'$ on $V(f)$ by modifying the inside of the super-face $f$. First we decouple the supported cross triangles drawn inside $f$ which share their landing components by adding a dummy landing vertex for each such cross triangle and making the new dummy vertex its landing component. Then the inside of $f$ is fully triangulated using additional type-$0$ edges such that in total it contains $|E(f)| - 2$ triangular faces. Notice that, this process cannot decrease the number of $survive(f)$ triangles drawn inside of $f$ in $H'$.

\begin{lemma}
\label{lem: triangulation contains friends}
If a super-face $f: |E(f)| \geq 5$ contains a single cactus fence structure and only one additional edge, then any triangulation of $f$ using type-$0$ edges must contain the free sides for at least one pair of cactus triangles which are friends.
\end{lemma}
\begin{proof}
The lemma follows easily using the facts that in any triangulation of a polygon there are at least two triangles each containing two side of the polygon and no two base vertices can be joined by an edge inside super-face $f$ as this will create a multi-edge, hence there should be at least one triangular face containing two adjacent free edges each belonging a different cactus triangle from a pair of strongly adjacent cactus triangles.
\end{proof}

It is clear that $|E(f)| \geq 5$, hence by Lemma~\ref{lem: triangulation contains friends}, $H'$ contains an edge $e'$ joining strongly adjacent pair of cactus triangles.
Hence, $e' \in E(H')$ but not drawn inside $f$ in $G$ (since $G$ cannot contain any pair of friends), so $H' \setminus e'$ still contains all surviving faces in the original graph and has only $|E(f)| - 4$ triangular faces inside $f$. 
Since the friends edge $e'$ goes across the two free vertices of two cactus triangles but $e$ joins two base vertices of two cactus triangles, hence they cannot form a triangle together. This implies at least one more triangular face which is bounded by $e$, does not survive, which proves Lemma~\ref{lem:f1_0}.

\subsubsection{$\fset[1,1,0] \setminus \fset_{fri}[1,1,0]$ (Proof for  Lemma~\ref{lem:f11_0})}

Same reasoning as in the proof of the previous case proves this case as well, the only difference is that, since $a_2(f) + a_1(f) + p_1^{base}(f) = 1$ and $a_1^{free}(f) = 1$, it implies $|Occ(f)| = 0$. 
We can show $survive(f) \leq |E(f)| - 4 = \mu(f) - 4$ by simply using the absence of edge $e'$ (from Lemma~\ref{lem: triangulation contains friends}), therefore missing two surviving faces from the triangulation in the interior of super-face $f$.



\subsubsection{$\fset[2,0,0] \setminus \fset_{fri}[2,0,0]$ (Proof for  Lemma~\ref{lem:f2_0})}

It suffice to show that $survive(f) \leq |E(f)| - 4$: 
Since $|Occ(f)| = 2$, this would imply  $survive(f) \leq |E(f)| - 4 = |Free(f)| + |Occ(f)| - 4 = |Free(f)| + |Occ(f)|/2 - 3 = \mu(f) - 3$ which proves the lemma. 

Similarly to the previous case, let $H'$ be the maximal auxiliary graph on $V(f)$ that contains all edges drawn in the interior of $f$ in $G$. 
Then $H'$ has $|E(f)| - 2$ triangular faces inside of $f$. 
Since $p_1^{base}(f) + a_2(f) + a_1^{occ}(f) = 2$, let $e_1$ and $e_2$ be the two edges bounding $f$ that contribute to this sum.
If $e_1$ and $e_2$ bound different faces of $H'$, then we are done, since the number of surviving faces of $H'$ is at most $|E(f)| - 4$. 

Now, assume that $e_1$ and $e_2$ bound the same face of $H'$. 

\begin{lemma} [The second fence lemma]
\label{lem: fence 2}
For any super-face $f \in \fset[2, 0, 0] \setminus \fset_{fri}[2,0,0]$, if the two edges corresponding to $p^{base}_1(f) + a_2(f) + a_1(f)$ are adjacent, then the face consists of a cactus fence of size $p^{free}(f)$ together with two edges $e_1$ and $e_2$ that contribute to the sum $p^{base}_1(f) + a_2(f) + a_1(f)$. 
\end{lemma}

Since, both $e_1$ and $e_2$ bounds the same triangular face of $H'$, they must be adjacent. Let $e$ be the third edge which bounds the triangular face adjacent to both $e_1$ and $e_2$ drawn inside of $f$ in $H'$.
Now, consider the graph $\tilde H = H' \setminus \{e_1, e_2\}$, so $\tilde H$ consists of a cactus fence together with $e$. 
Using Lemma~\ref{lem: triangulation contains friends}, $\tilde H$ must contain an edge joining strongly adjacent pair of cactus triangles. 
This edge cannot exist in the original graph since $f$ contains no pair of friends, so $\tilde H \setminus e$ still contains all surviving faces of the original graph. But it contains at most $|E(f)| - 5$ surviving faces. 

\subsection{Proof of the Fence Lemmas} 

In this section, we prove the two fence lemmas used in deriving the gain bounds in the previous section. 
An important notion that we will use is that of the {\em trapped triangles.}

\paragraph{Trapped and free triangles:} We further classify heavy cactus triangles based on whether their free component is a singleton or not. Let $f$ be a face that contains free sides of heavy triangle $t$. 
If the free component of  heavy triangle $t$ is a singleton, then we call $t$ a {\em free} triangle, else it will be a {\em trapped} triangle inside $f$. 



The following lemma implies the two fence lemmas used in the previous section.

\begin{lemma}
\label{lem:fence}
For any super-face $f \in \fset$ with $p^{base}(f) = 0$, if $a_1(f) + a_2(f) = 1$ or if $a_1(f) + a_2(f) = 2$ but the two type-$1$ and type-$2$ edges are adjacent: 
\begin{itemize}
    \item Then, there can be no triangle trapped inside $f$ and 
    
    \item Every pair of adjacent triangles is strongly adjacent. 
\end{itemize}
\end{lemma}

\begin{proof}
We do this in two steps. 
\begin{itemize}
    \item In the first step, we argue that every triangle is not trapped inside $f$. 
    Assume otherwise, that some $t: V(t) = \{u,v,w\}$ is trapped, and the free component $S_w^t$ is not a singleton. 
    Since $S^t_w$ is a free component, we have that $B^t_{wu} \cup B^t_{wv}$ is empty. 
     By using Observation~\ref{obs:f_a1_a2} on $S_w^t$, there is at least one type-$1$ or type-$2$ edge, says $e$, bounding the outer-face of the graph $H[S_w^t]$ and edge $e$ also bounds the face $f$ (see Figure~\ref{fig:trapped}). 
    
    Now consider the contracted graph that contracts $S_w^t$ into a single vertex. 
    Let $f'$ be the residual super-face corresponding to $f$ and $S'$ be the residual component after the contraction of $S_w^t$. 
    Notice that, the graph $H[S']$ contains only heavy triangles: For any cactus triangle $t'$ in $H[S']$, no type-$1$ or type-$2$ edge that contributes to its ``heaviness'' was contracted. This implies that the super-face $f'$ of $H'[S']$ contains at least one type-$1$ or type-$2$ edge, says $e'$  (by Observation~\ref{obs:f_a1_a2}). 
    It is easy to verify that $e$ and $e'$ are not adjacent.
    
     \begin{figure}[H]
        \centering
        \includegraphics[width=0.9\textwidth]{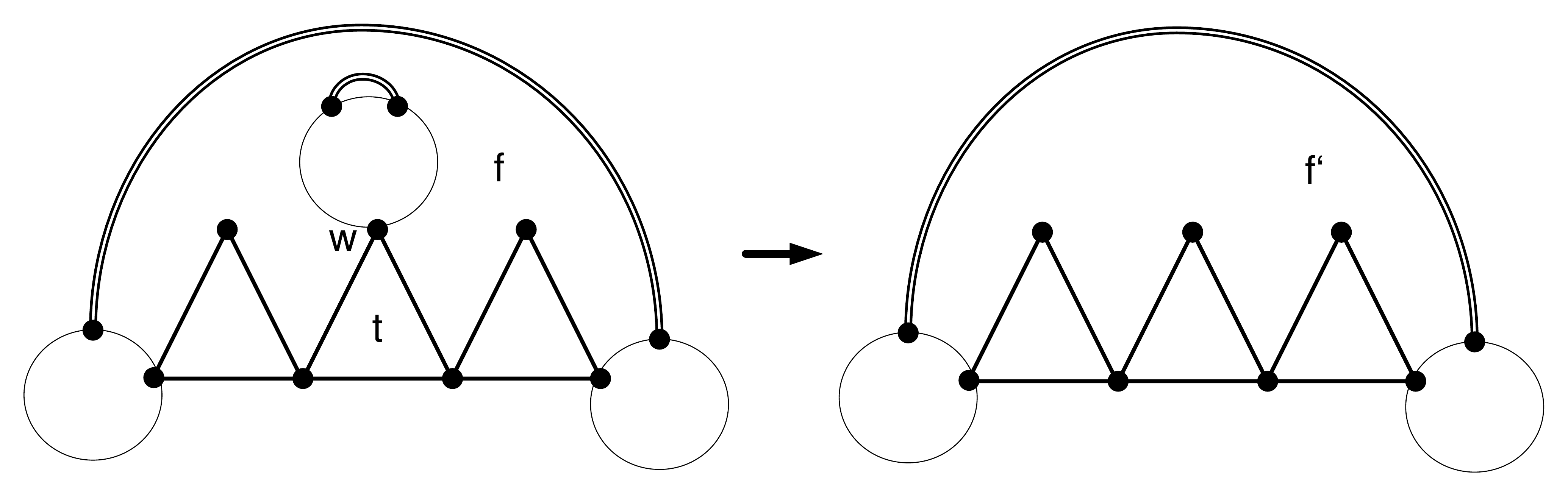}
        \caption{Contraction operation when $f \in \fset$ contains a trapped triangle's free side.}
        \label{fig:trapped}
    \end{figure}
    
    \item Now we prove the second property. Let $t_1$ and $t_2$ be an adjacent pair of triangles whose free sides bound the super-face $f$. 
    We will argue that $t_1: V(t_1) = \{u_1,v_1, w_1\} $ and $t_2: V(t_2) = \{u_2, v_2, w_2\}$ are strongly adjacent, with $w_i$ being the free vertex of $t_i$ and $v_1 = v_2$ being the common vertex. 
    Assume that they were not strongly adjacent.

    
    Notice that, since the free sides for both $t_1, t_2$ bound a common super-face $f$, this can only happen if $S^{t_1}_{v_1}$ has a connected component $S': S' \subseteq S^{t_1}_{v_1} \cap S^{t_2}_{v_2}$ drawn inside $f$ (see Figure~\ref{fig:weakly-adjacent-free-sides}.)  
    
    Observe that $C[S']$ contains only heavy cactus triangles: Any type-$1$ or type-$2$ edge with exactly one end point in $S'$ can only be incident on $v_1$ and must be drawn in the exterior of $f$. 
    Again, as in the previous case, we can do the contraction trick to argue that there exist two type-$1$ or type-$2$ edges $e$ and $e'$ bounding face $f$ such that $e \neq e'$ and they are not adjacent.
     
     \begin{figure}[H]
        \centering
        \includegraphics[width=0.9\textwidth]{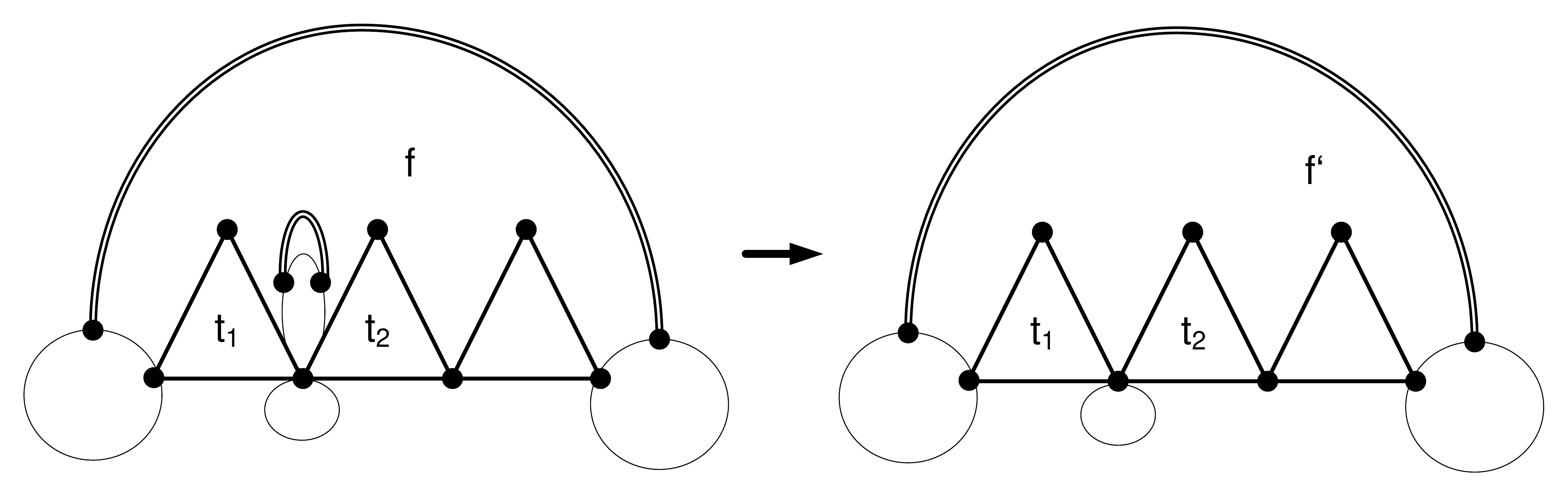}
        \caption{Contraction operation when $f \in \fset$ contains a pair of free sides which corresponds to a pair of weakly-adjacent triangles.}
        \label{fig:weakly-adjacent-free-sides}
    \end{figure}
\end{itemize}
\end{proof}

\subsection{Proof of the Friend Lemma (Proof of Lemma~\ref{lem:friend})} 
\label{sec:friend-lemma} 
In this section, we prove the friend lemma. We will rely on some structural observations:

\begin{lemma}
\label{lem:one-blue-per-occupied-t}
Let $f \in \fset$ and $t=(u,v,w)$ be any heavy cactus triangle such that $E(t) \cap E(f) \neq \emptyset$, and $uv \in E(t)$ be its (unique) cactus edge for which $B^t_{uv} \neq \emptyset$.  Then, we have $|E(f) \cap B^t_{uv}| = 1$.
\end{lemma}
\begin{proof}
Let $P_2$ be a maximal trail along the boundary of $f$ starting from $u$ and only visiting vertices in $S_u^t$ in graph $H[S]$. Notice that $P_2$ may use cactus edges or type-$1$ or type-$2$ edges. Let $u_2$ be the other endpoint of $P_2$ and $u_2u_3$ be the next edge on the boundary of $f$, such that $u_3 \in S_v^t \cup S_w^t$. 
First, notice that $u_3$ cannot be in $S^w_t$ , for otherwise, we would have the free sides of $t$ on different super-faces. Therefore, $u_3 \in  S_v^t$ . Now, let $P_3$ be a maximal trail from $u_3$ along the boundary of $f$, visiting only vertices in $S_v^t$. We claim that $P_3$ must contain $v$: Otherwise, let $v'$ be the last node on $P_3$ and $e'$ be the next edge on $f$ incident to $v'$. Consider a region $R$ bounded by (i) the sides of $t$ on super-face $f$,(ii) trail $P_2u_3P_3$,and (iii) any path from $v'$ to $v$ using only cactus edges in $S_v^t$. This close region must contain super-face $f$, so $e'$ must be drawn inside $R$ (see Figure~\ref{fig:unique-blue-p-base-p-free-faces}). This is a contradiction since $e'$ cannot connect $v'$ to a node in $S_w^t$ (same reasoning as before), and similarly it cannot connect $v'$ to $S_u^t$ (this would contradict the choice of $u_2$ or the edge $u_2u_3$.).
\end{proof}

    \begin{figure}[H]
        \centering
        \includegraphics[width=0.9\textwidth]{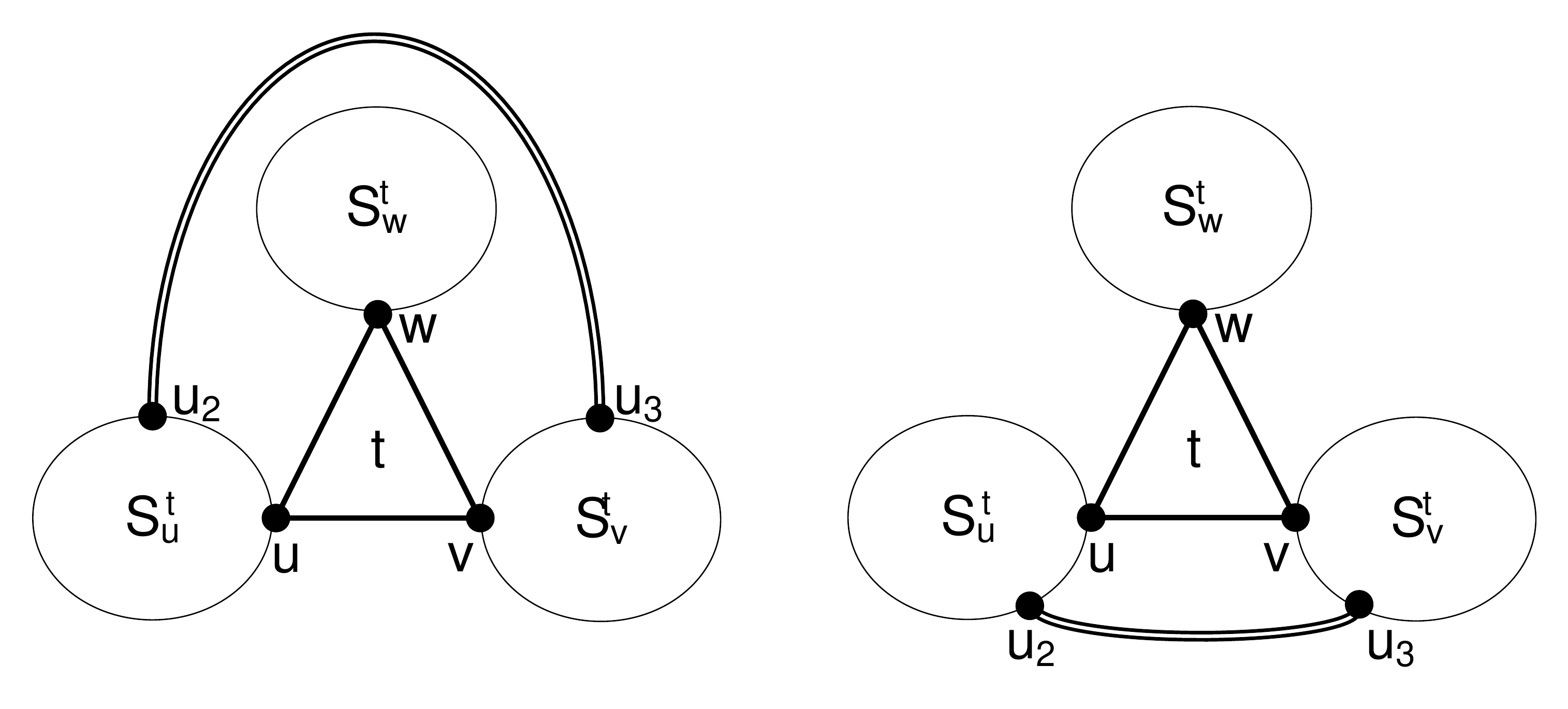}
        \caption{The regions containing free sides or base sides of a heavy triangle $t$ in the proof for Lemma~\ref{lem:one-blue-per-occupied-t}.}
        \label{fig:unique-blue-p-base-p-free-faces}
    \end{figure}

\begin{Observation}
\label{Observation:different-face-base-free-t}
For any heavy triangle $t$, the free or base sides will be adjacent to two different super-faces in $\fset$.
\end{Observation}

    \begin{figure}[H]
        \centering
        \includegraphics[width=0.7\textwidth]{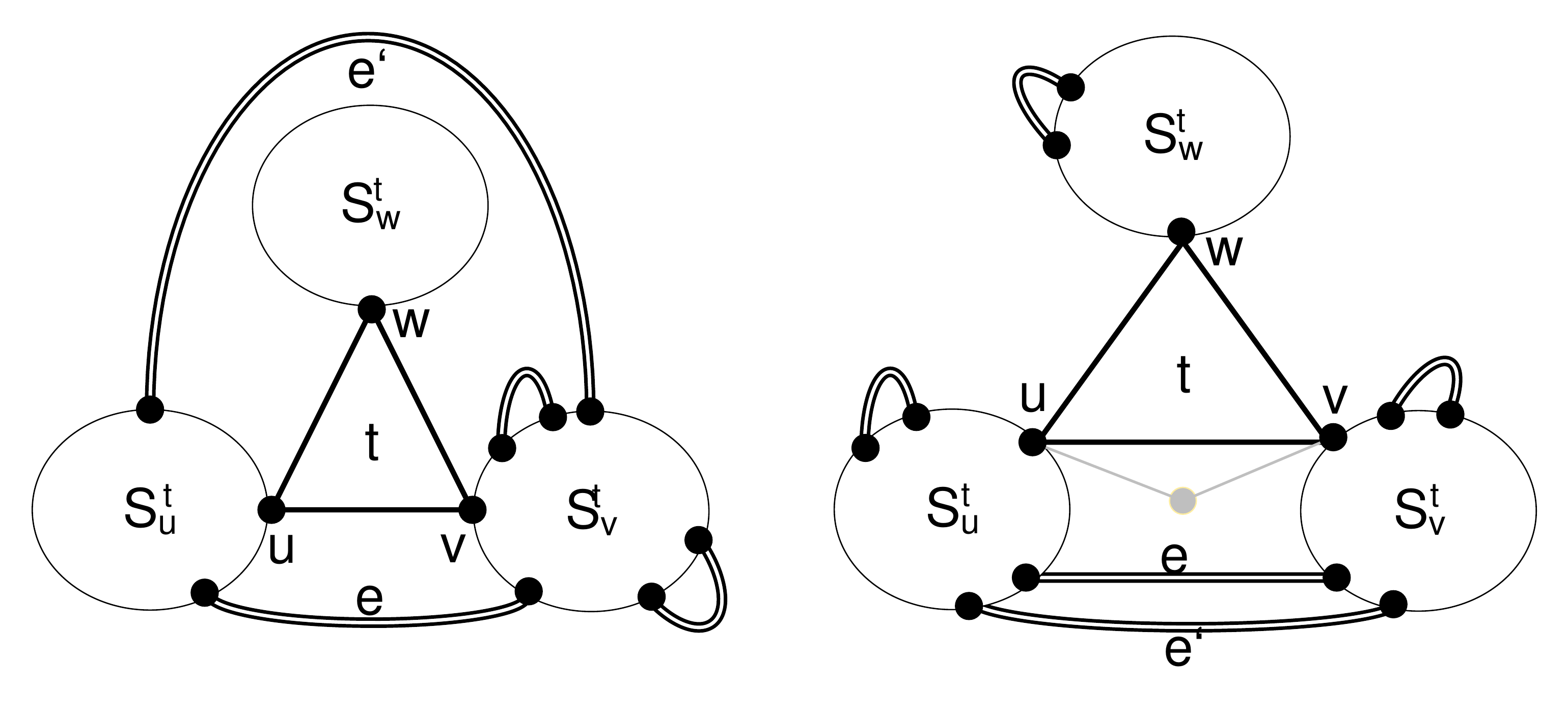}
        \caption{The structure of super-faces containing free and base sides of any heavy triangle $t$.}
        \label{fig:p-base-p-free-faces}
    \end{figure}

\begin{lemma}
\label{lem:different-one-blue-per-occupied-base-free}
Let $t$ be a heavy triangle. Let $f, f' \in \fset$ be the two different super-faces that contain the base and free sides of $t$ respectively.   
Let $e, e'$ be the unique type-$1$ or type-$2$ edges on $f$ and $f'$ across the occupied components of $t$ (given by Lemma~\ref{lem:one-blue-per-occupied-t}). Then $e \neq e'$.
\end{lemma}
\begin{proof}
Assume otherwise that $e = e'$, so the super-faces $f$ and $f'$ are adjacent at $e$. 
This means that there is only one type-$1$ or type-$2$ edge across the occupied components, contradicting to the fact that $t$ is heavy (see Fig.~\ref{fig:p-base-p-free-faces} for illustration).
\end{proof}

\begin{figure}[H]
    \centering
    \includegraphics[width=0.5\textwidth]{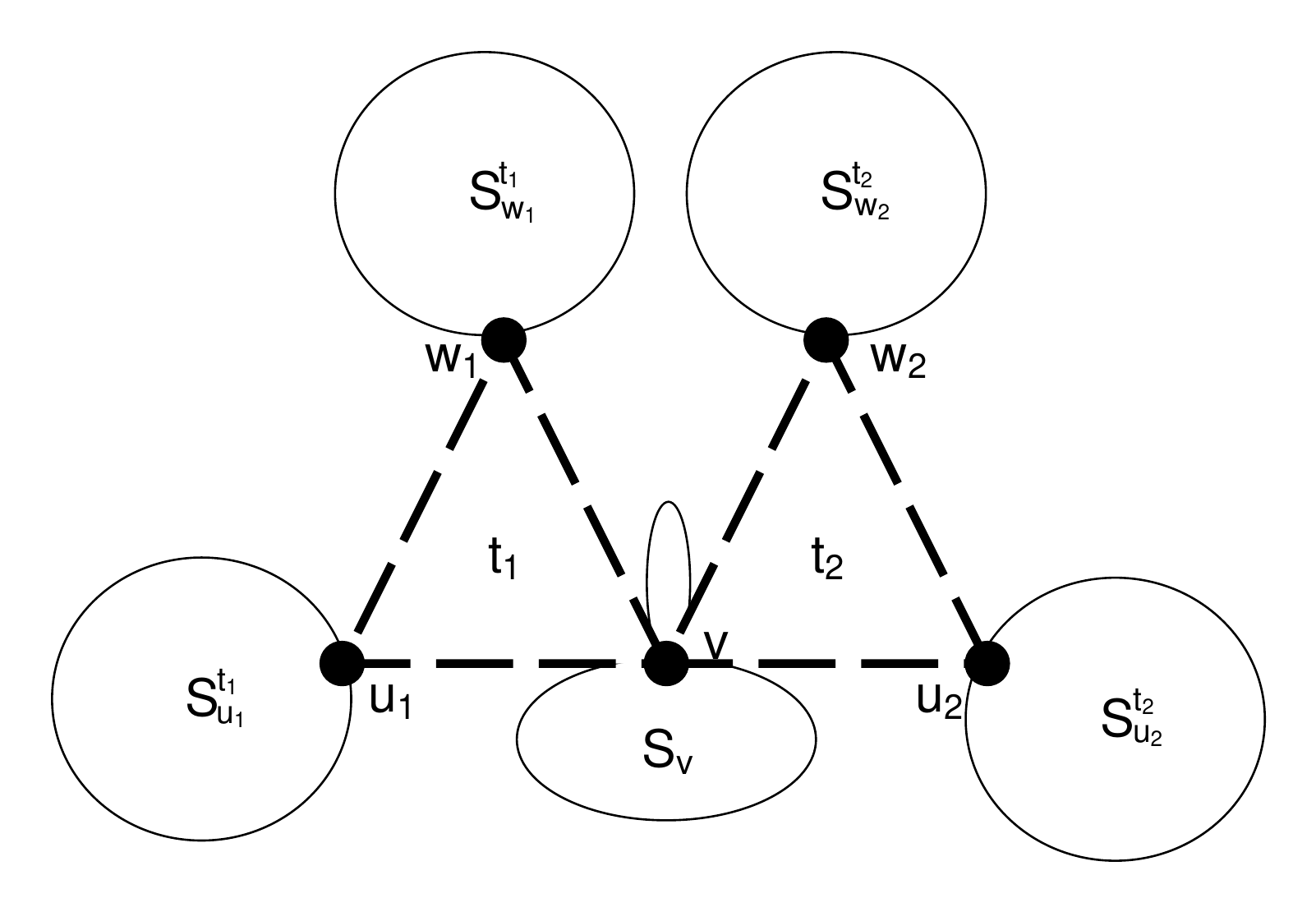}
    \caption{The structure of split components formed by removing two adjacent triangles.}
    \label{fig:adjacent-triangles-components}
\end{figure}

\paragraph{Components for two adjacent heavy triangles:} Now we fix the labelling for the new components created by the operation of removing edges for two adjacent heavy triangles from $\cset[S]$, which we will use in the rest of this section. Every time when we talk about two adjacent heavy triangles we will denote them by $t_1, t_2$ such that $V(t_1) = \{u_1, v_1 = v, w_1 \}$ and $V(t_2) = \{u_2, v_2 = v, w_2 \}$, where $w_1, w_2$ will be the corresponding free vertices and $v$ the common base vertex of $t_1$ and $t_2$. The vertices of the new components formed by removing edges $E(t_1) \cup E(t_2)$ from $\cset[S]$ will be $S_{w_1}^{t_1}, S_{w_2}^{t_2}, S_{u_1}^{t_1}, S_{u_2}^{t_2}, S_v$, such that $w_1 \in S_{w_1}^{t_1} , w_2 \in S_{w_2}^{t_2}, u_1 \in S_{u_1}^{t_1}, u_2 \in S_{u_2}^{t_2}$ and $v \in S_v$. Notice that the free components of $t_1, t_2$ are $S_{w_1}^{t_1}, S_{w_2}^{t_2}$ respectively, the occupied components of $t_1$ are $S_{u_1}^{t_1}, S_{v_1}^{t_1} = S_v \cup S_{w_2}^{t_2} \cup S_{u_2}^{t_2}$ and the occupied components of $t_2$ are $S_{u_2}^{t_2}, S_{v_2}^{t_2} = S_v \cup S_{w_1}^{t_1} \cup S_{u_1}^{t_1}$.

\begin{lemma}
\label{lem:blue-edges-same-face-p-free-sides}
Let $f \in \fset$ be a super-face. Let $t_1, t_2: V(t_i) = (u_i,v_i, w_i)$ be two adjacent heavy cactus triangles with $v_1 = v_2$ (say $v$) such that $E(t_i) \cap E(f) \neq \emptyset$ for $i  \in \{1,2\}$. 
    For each $i$, let $u_i v_i \in E(t_i)$ be the base edge and $B^{t_i}_{u_i v_i} \cap E(f)  = \{e_i\}$ (unique due to lemma~\ref{lem:one-blue-per-occupied-t}).
\begin{enumerate}
    \item \label{itm:same-blue-edges} $e_1 = e_2 := e$ if and only if the common edge $e$ goes across $S^{t_1}_{u_1}$ and $S^{t_2}_{u_2}$.
    \item  \label{itm:diff-blue-edges} $e_1 \neq e_2$ if and only if both $e_1, e_2$ are incident to $S_v$. 
\end{enumerate}

\end{lemma}

\begin{proof}
The first direction for item~\ref{itm:same-blue-edges} is easy to see by the way $S_{u_1}^{t_1}, S_{u_2}^{t_2}$ are defined and by the fact that $e$ goes across the occupied components for both $t_1, t_2$. 
In the other direction, if $e_1$ goes across $S_{u_1}^{t_1}, S_{u_2}^{t_2}$, hence it also goes across the occupied components $S_{u_2}^{t_2}, S_{v_2}^{t_2}$ for $t_2$. This along with the fact that $e_1$ belongs to $f$ and Lemma~\ref{lem:one-blue-per-occupied-t}, it implies that $e_2 = e_1$.

One direction for item~\ref{itm:diff-blue-edges} follows from the negation of item~\ref{itm:same-blue-edges} because if any one of $e_1$ or $e_2$ goes across $S_{u_1}^{t_1}, S_{u_2}^{t_2}$, then it implies $e_1 = e_2$. On the other hand, if one of $e_1$ or $e_2$ is incident on $S_v$, then they cannot be same by item~\ref{itm:same-blue-edges}.
\end{proof}

\begin{lemma}
\label{lem:different-one-blue-per-p-base}
Let $f \in \fset$ be a super-face. Let $t_1, t_2: V(t_i) = (u_i,v_i, w_i)$ be two adjacent heavy cactus triangles with $v_1 = v_2$ (say $v$) such that both of whose free sides belong to $f$. If the base sides for $t_1$ and $t_2$ belong to two different super-faces $f_1, f_2 \in \fset$ and for each $i$, let $B^{t_i}_{u_i v_i} \cap E(f_i)  = \{e_i\}$ (unique due to lemma~\ref{lem:different-one-blue-per-occupied-base-free}). Then at least one of $e_1$ or $e_2$ is incident on some vertex in $S_v$, which in turn implies $e_1 \neq e_2$.
\end{lemma}
\begin{proof}
As $t_1, t_2$ are adjacent, we use the notations defined above for the various components corresponding to two adjacent heavy triangles. First, assuming that at least one of $e_1$ or $e_2$ is incident on some vertex in $S_v$, we prove that $e_1 \neq e_2$.

By contradiction, let $e_1 = e_2 :=u'v'$. Now we show that there will be a cycle in $\cset[S]$ sharing an edge with $t_1$, contradicting the fact that $\cset$ is a triangular cactus (see Figure~\ref{fig:same-blue-edge-across-adjacent-p-contradiction}). By the above claim, this edge is incident to $S_v$ (say $v' \in S_v$). Also, by the way $e_1$ and $e_2$ are defined, the other end point $u'$ belongs to both $S_{u_1}^{t_1}$ and $S_{u_2}^{t_2}$. Hence in $\cset[S] \setminus (E[t_1] \cup E[t_2])$, using only cactus edge, there is a path $P_1$ from $u'$ to $u_1$ and another path $P_2$ from $u'$ to $u_2$. Hence, $u' P_1 u_1 \cup u_1 v \cup v u_2 \cup u_2 P_2 u'$ is a cycle in $\cset[S]$ sharing edge $u_1 v$ with $t_1$, contradicting the fact that $\cset$ is a triangular cactus.

\begin{figure}[H]
    \centering
        \includegraphics[width=0.5\textwidth]{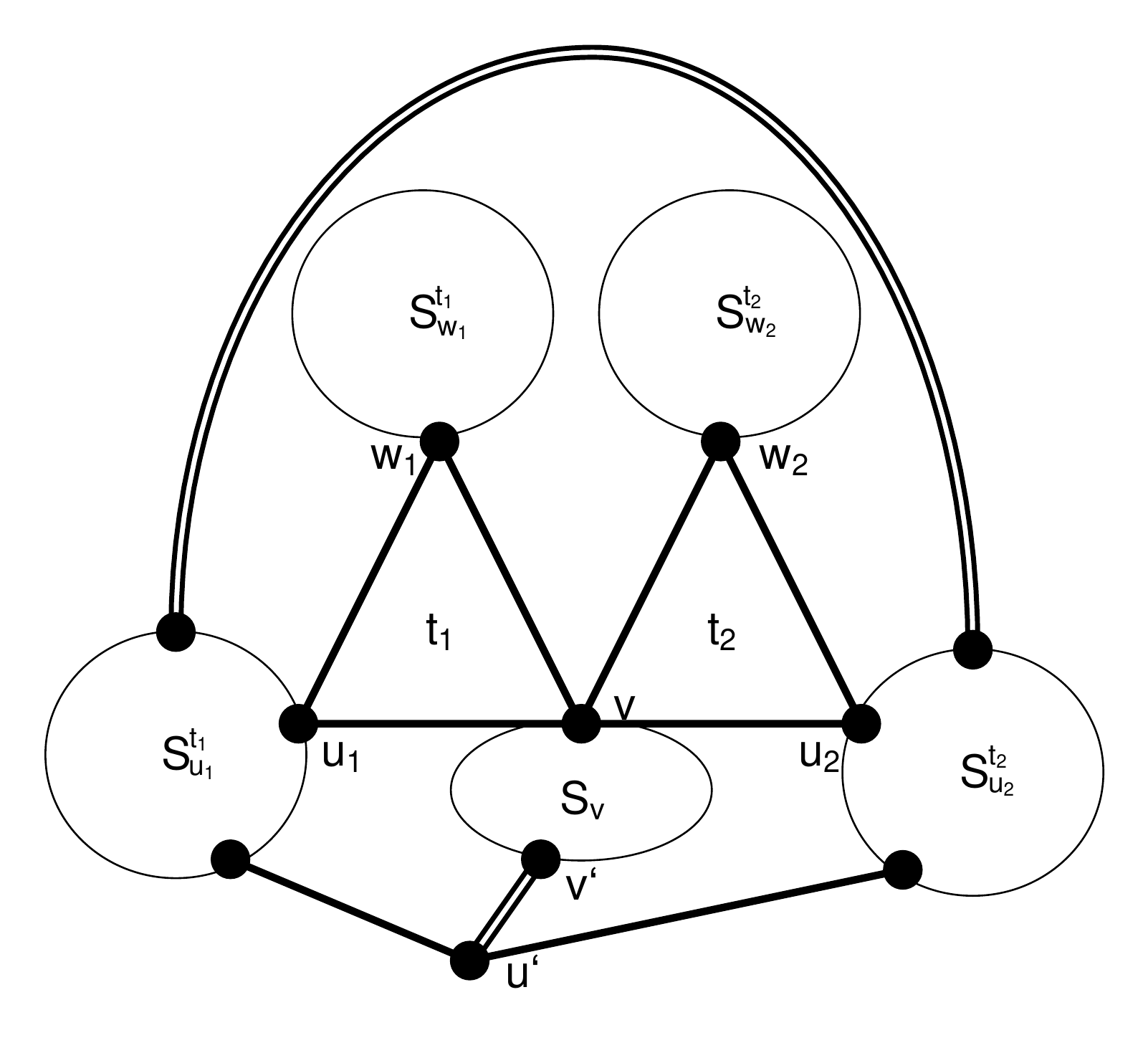}
        \caption{An illustration of the example used to reach a contradiction. The cycle in $\cset[S]$ sharing edge $v_2^1 w$ with $t_1$, when $e_1 =  e_2 :=uv$ and $v \in S_w$.}
        \label{fig:same-blue-edge-across-adjacent-p-contradiction}
    \end{figure}

Finally, to finish the proof for our lemma, we prove that at least one of $e_1$ or $e_2$ is incident on some vertex in $S_v$. 

For contradiction assume none of $e_1, e_2$ are incident on $S_v$. Notice that since free sides of $t_1, t_2$ belongs to the same super-face $f$, we can partition the component $S_v$ into two parts $S_v', S_v''$ (with an exception that $v$ is a common vertex), such that $S_v'$ is drawn inside $f$, $S_v''$ outside of $f$ and $v$ lies on $f$.

Starting with vertex $u_1$ and the base side $u_1v$ create a maximal trail $P_1$ in $H[S]$ along the boundary of $f_1$ by visiting type-$1$ or type-$2$ or cactus edges and vertices only from $S_v''$. This trail should end at some vertex $v' \in S_v$ (possibly $v$), such that there is a type-$1$ or type-$2$ edge leaving $S_v''$ incident on $v'$ (say $v'u'$). If not, then the trail would end at $v$ and the base side $v u_2$ will be the next edge belonging to super-face $f_1$ in the graph $H[S]$, which contradicts our assumption. Notice that since $S_{w_1}^{t_1}, S_{w_2}^{t_2}$ are the free components, hence either $u' \in S_{u_1}^{t_1}$ or $u' \in S_{u_2}^{t_2}$. In case when $u' \in S_{u_1}^{t_1}$, it implies that $v'u'$ is an edge going across the components of $t_1$ and also belongs to $f_1$. But by Lemma~\ref{lem:one-blue-per-occupied-t}, it implies that $e_1 = u'v'$, contradicting our assumption (see Figure~\ref{fig:diff-blue-edges-across-adjacent-p}).

In the other case, when $u' \in S_{u_2}^{t_2}$, we look at the original graph $H[S]$. Since both $v', v \in S_v''$, there exists a path $P_v$ from $v'$ to $v$ using only cactus edges and vertices from $S_v''$. Similarly since $u', u_2 \in S_{u_2}^{t_2}$, there exists a path $P_{u'}$ from $u'$ to $u_2$ using only cactus edges and vertices from $S_{u_2}^{t_2}$. Hence, the region $R$ bounded by $v' P_v v \cup v u_1 \cup u_2P_{u'} u' \cup u'v'$ contains only the vertices from $S_v'' \cup S_{u_2}^{t_2}$ at its boundary and also contains the base edge $v u_2$ from the side outside of $f$. This implies that the super-face $f_2$ can only be drawn inside $R$ and can only contain vertices from $S_v'' \cup S_{u_2}^{t_2}$ and hence for $e_2$ to go across the occupied components of $t_2$, the only possibility is to go across $S_v''$ and  $S_{u_2}^{t_2}$, contradicting our assumption (see Figure~\ref{fig:diff-blue-edges-across-adjacent-p}). 


    \begin{figure}[H]
        \centering
        \includegraphics[width=0.9\textwidth]{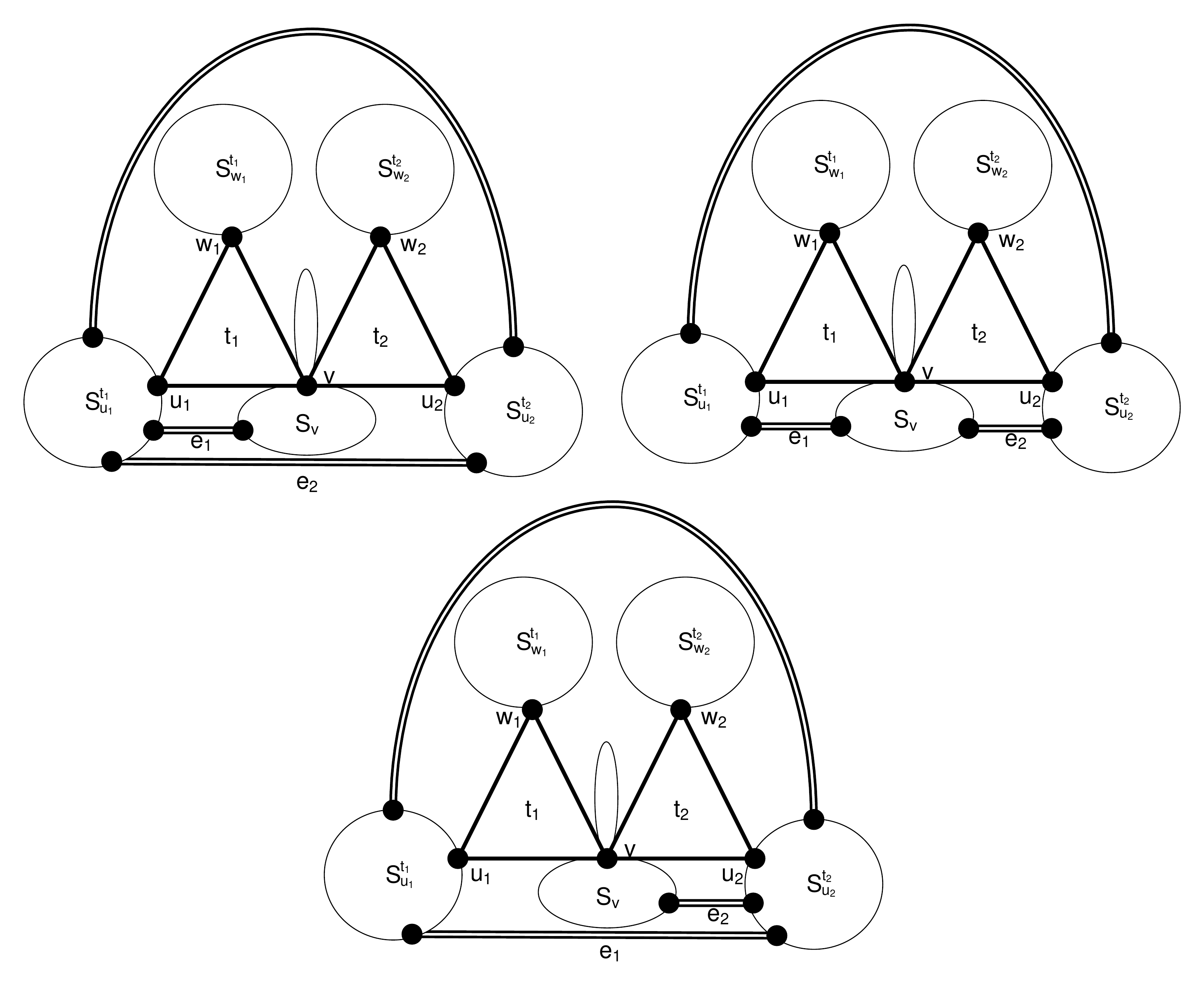}
        \caption{The possible drawings when two adjacent heavy cactus triangles $t_1, t_2$, both of whose base sides belong to two different super-faces $f_1, f_2 \in \fset$, where $e_1$ and $e_2$ are the corresponding edges for $t_1$ and $t_2$ given by Lemma~\ref{lem:one-blue-per-occupied-t}.}
        \label{fig:diff-blue-edges-across-adjacent-p}
    \end{figure}

\end{proof}

\subsubsection*{Proof of the first part}
We assume that for contradiction that $t_1, t_2$ are friends where $t_2$ is type-$1$ triangle.    
We will argue that there exists an improving $2$-swap, contradicting the fact that $\cset$ is the optimal cactus.
As $t_1, t_2$ are adjacent, we use the notations defined above for the various components corresponding to two adjacent heavy triangles.

Let $t'$ be the supported cross triangle of $t_2$ and let $t_3$ be the empty triangle formed by vertices $\{w_1, w_2, v\}$. 
Also let $e_1, e_2$ (possibly same) be the type-$1$ or type-$2$ edges belonging to the super-face $f$ going across the occupied components of $t_1, t_2$ respectively (exists by Lemma~\ref{lem:one-blue-per-occupied-t}). Also, let $e_2'$ be the edge going across the occupied components of $t_2$ which belongs to the super-face $f_2$ containing the base side for $t_2$ (exists by Lemma~\ref{lem:one-blue-per-occupied-t}). By Lemma~\ref{lem:different-one-blue-per-occupied-base-free}, $e_2 \neq e_2'$.

Now there could be two cases based on the landing components for supported cross triangles $t', t_1'$. The second case will be further divided into sub-cases based on the way $e_1$ is drawn in $\phi_H$.
\begin{itemize}
\item ($e_1$ is a type-$1$ edge and different landing components for supported cross triangles $t_1', t'$): 
We modify our cactus by $\cset' = (\cset \setminus (E(t_1) \cup E(t_2))) \cup E(t') \cup E(t_1') \cup E(t_3)$ (see Figure~\ref{fig:case-e1-type-1-non-common-landing-components}). Note that $t'$ will attach $S_2^{t_2}$ to  $S_v$, $t_3$ will attach $S_v$ with $S_{w_1}^{t_1}$ and $S_{w_2}^{t_2}$ and finally $t_1'$ will attach $S_{u_1}^{t_1}$ to this structure, hence $\cset'$ will be a triangular cactus with one more cactus triangle, which contradicts the optimality of $\cset$.

    \begin{figure}[H]
        \centering
        \includegraphics[width=0.9\textwidth]{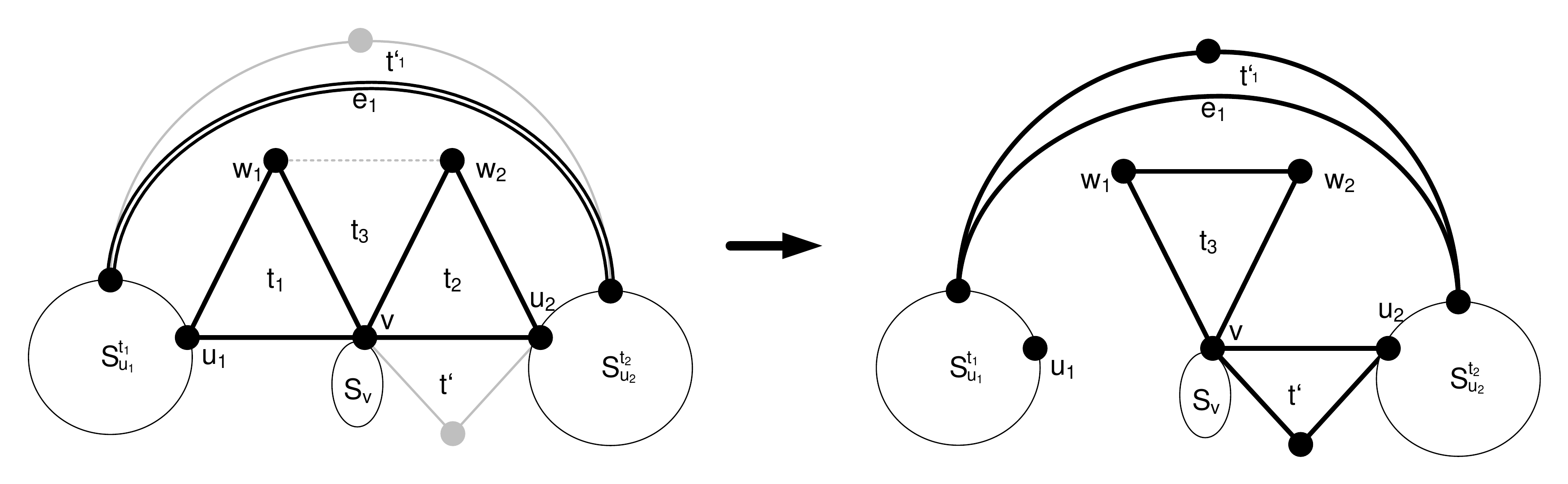}
        \caption{Improving $2$-swap when the edge $e_1$ given by Lemma~\ref{lem:one-blue-per-occupied-t} is type-$1$ and the cross triangles supported by $e_1$ and $t_2$ have different landing components.}
        \label{fig:case-e1-type-1-non-common-landing-components}
    \end{figure}

\item ($e_1$ is a type-$1$ edge and $t_1', t'$ share a common landing component): 
Since $e_1$ is the unique edge belonging to $f$ going across the occupied components of $t_1$ (see Lemma~\ref{lem:one-blue-per-occupied-t}), there could be two sub-cases.

\begin{itemize}

    \item ($e_1$ goes across $S_{u_1}^{t_1}, S_{u_2}^{t_2}$): From Lemma~\ref{lem:blue-edges-same-face-p-free-sides} it implies that $e_1 = e_2 =: e$. Now if we focus on $t_2$, the base side of it must belong to a super-face $f_2$ such that $e_2'$ goes across its occupied components such that $f \neq f_2$ and $e_2' \neq e$ (by Observation~\ref{Observation:different-face-base-free-t}, Lemma~\ref{lem:one-blue-per-occupied-t} and~\ref{lem:different-one-blue-per-occupied-base-free}). Also, by the uniqueness of the edge $e_2'$, the edge $e$ cannot belong to $f_2$. But this implies that $t'$ is drawn inside $f_2$ and $t_1'$ is drawn outside it, hence by Observation~\ref{Observation:different-components-green-blue-cycle} $t_1', t'$ cannot share their landing components, contradiction.
    
    \begin{figure}[H]
        \centering
        \includegraphics[width=0.5\textwidth]{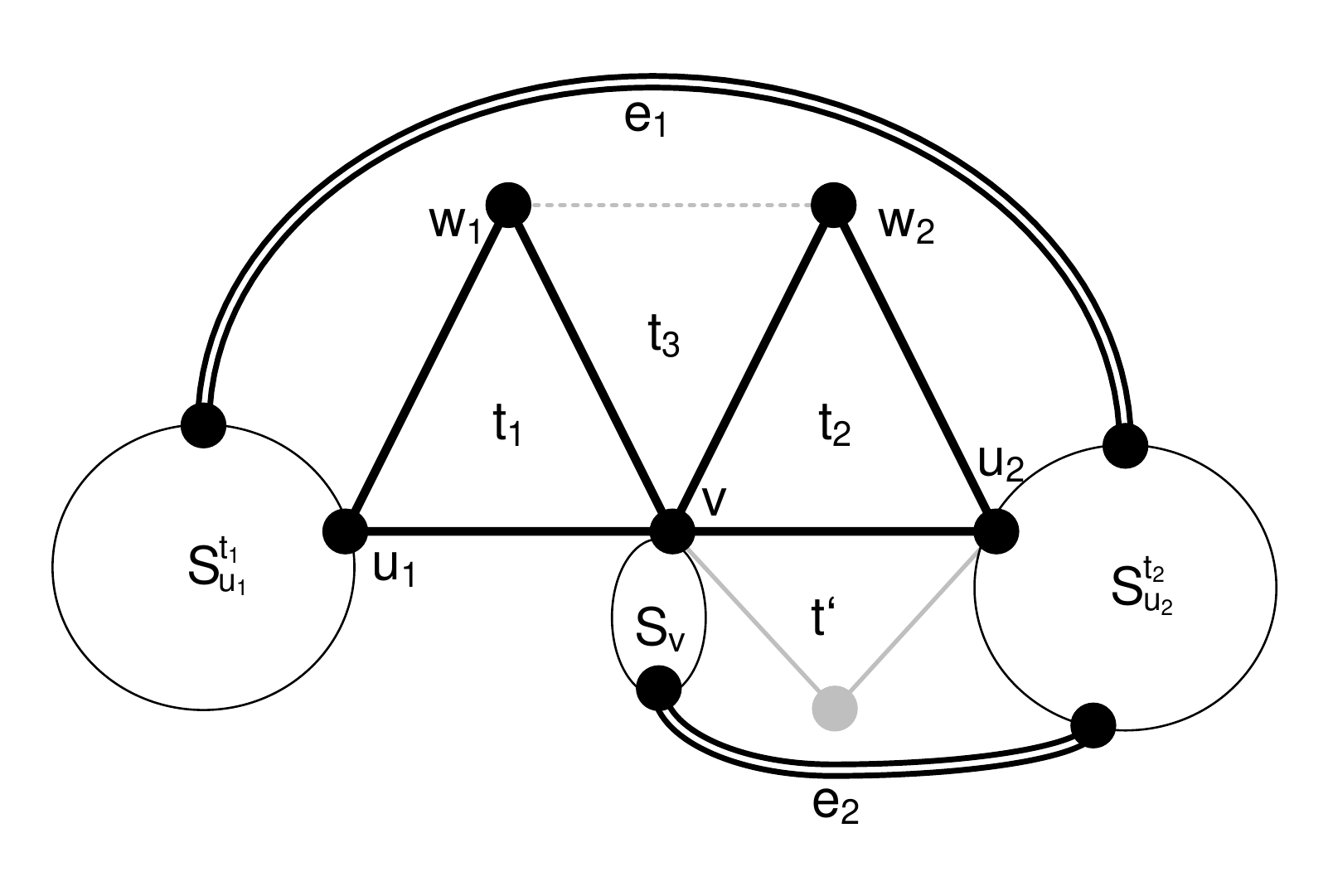}
        \caption{The case when the edge $e_1$ given by Lemma~\ref{lem:one-blue-per-occupied-t} is type-$1$ and the cross triangles supported by $e_1$ and $t_2$ have the same landing component. If $e_1$ goes across $S_{u_1}^{t_1}, S_{u_2}^{t_2}$, then due to the presence of $e_2'$, the landing component for $t_1', t'$ cannot be the same.}
        \label{fig:case-e1-type-1-common-landing-component-S_2_t2}
    \end{figure}
  
    \item ($e_1$ goes across $S_{u_1}^{t_1}, S_v$): By Lemma~\ref{lem:blue-edges-same-face-p-free-sides}, it implies $e_2 \neq e_1$ and both $e_1, e_2$ are incident on $S_v$. Now let $u'v' := e_2$ such that $u' \in S_{u_2}^{t_2}$ and $v' \in S_v$. Since both $v', v \in S_v$, there exists a path $P_{v'}$ from $v'$ to $v$ using only cactus edges and vertices from $S_v$. Similarly since $u', u_2 \in S_{u_2}^{t_2}$, there exists a path $P_{u'}$ from $u'$ to $u_2$ using only cactus edges and vertices from $S_{u_2}^{t_2}$. Hence, the region $R$ bounded by $v' P_{v'} v \cup v u_2 \cup u_2 P_{u'} u' \cup u'v'$ contains only the vertices from $S_v \cup S_{u_2}^{t_2}$ at its boundary and also contains the base edge $v u_2$ (see Fig.~\ref{fig:case-e1-type-1-common-landing-component-S_w}). This implies that the super-face $f_2$ can only be drawn inside $R$ and consecutively the triangle $t'$ is drawn inside $R$. This implies that $e_1$ should be drawn outside $R$ and consecutively $t_1'$ is drawn outside $R$, hence by Observation~\ref{Observation:different-components-green-blue-cycle} $t_1', t'$ cannot share their landing components, contradiction.
    
    \begin{figure}[H]
    \centering
        \includegraphics[width=0.5\textwidth]{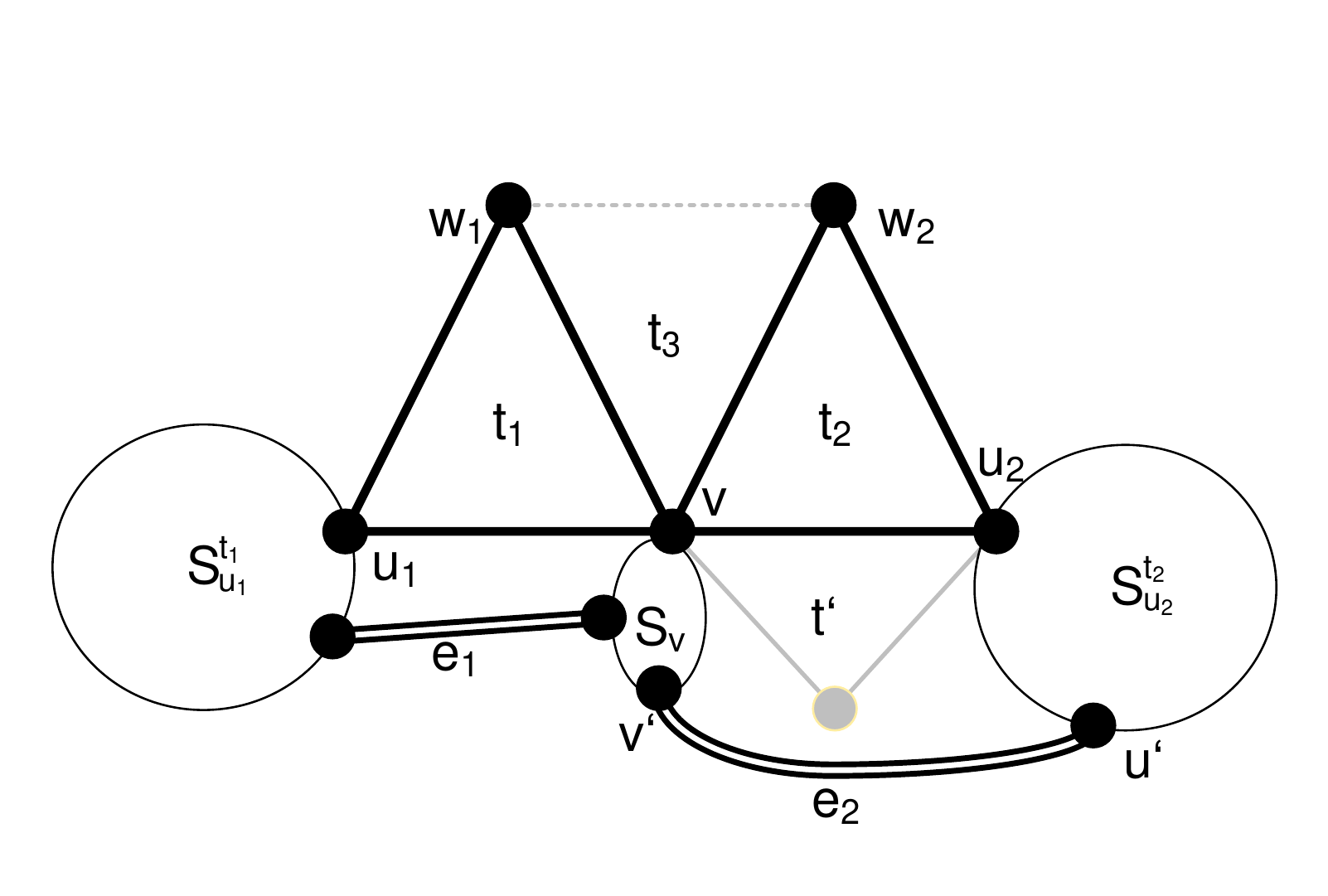}
        \caption{The setting before we reach a contradiction for the case when the edge $e_1$ given by Lemma~\ref{lem:one-blue-per-occupied-t} is type-$1$ and the cross triangles supported by $e_1$ and $t_2$ have the same landing component. If $e_1$ goes across $S_{u_1}^{t_1}, S_v$, then $e_2$ will go across $S_v, S_{u_2}^{t_2}$, hence the landing component for $t_1', t'$ cannot be the same.}
        \label{fig:case-e1-type-1-common-landing-component-S_w}
    \end{figure}

\end{itemize}
\end{itemize}

\subsubsection*{Second part}
For contradiction we assume that $t_1, t_2$ are friends. Again $t_1, t_2$ are adjacent, hence we use the notations defined above for the various components corresponding to two adjacent heavy triangles.
Let $t_3$ be the triangle formed by vertices $\{w_1, w_2, v\}$. Also let $e_1$ be the unique type-$1$ or type-$2$ edge belonging to the super-face $f_1$ containing base side of $t_1$ going across occupied components of $t_1$ (exists by Lemma~\ref{lem:one-blue-per-occupied-t}) and $e_2$ be the unique type-$1$ or type-$2$ edge belonging to the super-face $f_2$ containing base side of $t_2$ going across occupied components of $t_2$ (exists by Lemma~\ref{lem:one-blue-per-occupied-t}). Let $e_1', e_2'$ (possibly same) be the unique type-$1$ or type-$2$ edges belonging to the super-face $f$ going across occupied components of $t_1, t_2$ respectively (exists by Lemma~\ref{lem:one-blue-per-occupied-t} and the fact that $f$ contains free sides for both $t_1, t_2$.). By Lemma~\ref{lem:different-one-blue-per-occupied-base-free} and~\ref{lem:different-one-blue-per-p-base}, $e_1 \neq e_2$, $e_1 \neq e_1'$ and $e_2 \neq e_2'$. 

Now we fix the cross triangles $t_1', t_2', t_1''$ each supported by $e_1, e_2, e_1'$ respectively, as follows.
The idea here is to fix these supported cross triangles in such a way that their landing components are as different as possible. If $e_1'$ supports a cross triangle drawn inside $f$, then we fix $t_1''$ to be that triangle, otherwise $t_1''$ is any supported cross triangle of $e_1'$. 
If there exists a cross triangle supported by $e_1$ which does not share its landing component with $t_1''$ then we fix $t_1'$ to be that triangle, otherwise $t_1'$ is any supported cross triangle of $e_1$. Similarly, we choose the supported cross triangle $t_2'$ of $e_2$ such that it does not share its landing component with any of $t_1''$ or $t_1'$ (or both), otherwise $t_2'$ is any supported cross triangle of $e_2$. 

By the way $t_1', t_2', t_1''$ are chosen, it ensures that all three of them can share a landing component if and only if all three $e_1, e_2, e_1'$ are type-$1$ edges (by Lemma~\ref{lem:type-2-edge-support}). Now there could be three cases.

\begin{itemize}
\item ($t_1', t_2'$ have different landing components): Since the base sides for $t_1, t_2$ are in different super-faces, Lemma~\ref{lem:different-one-blue-per-p-base}) implies that at least one of $e_1, e_2$ is incident on $S_v$ (by renaming assume $e_1$). Hence, if the triangles $t_1', t_2'$ do not share their landing components then we modify our cactus by $\cset' = (\cset \setminus (E(t_1) \cup E(t_2))) \cup E(t_1') \cup E(t_2') \cup E(t_3)$ (See Figure~\ref{fig:p0-2-swap1}). Note that $t_1'$ will attach $S_{u_1}^{t_1}$ to  $S_v$, $t_3$ will attach $S_v$ with $S_{w_1}^{t_1}$ and $S_{w_2}^{t_2}$ and finally $t_2'$ will attach $S_{u_2}^{t_2}$ to this structure, hence $\cset'$ will be a triangular cactus with one more cactus triangle, which contradicts the optimality of $\cset$.
    
    \begin{figure}[H]
    \centering
        \includegraphics[width=0.9\textwidth]{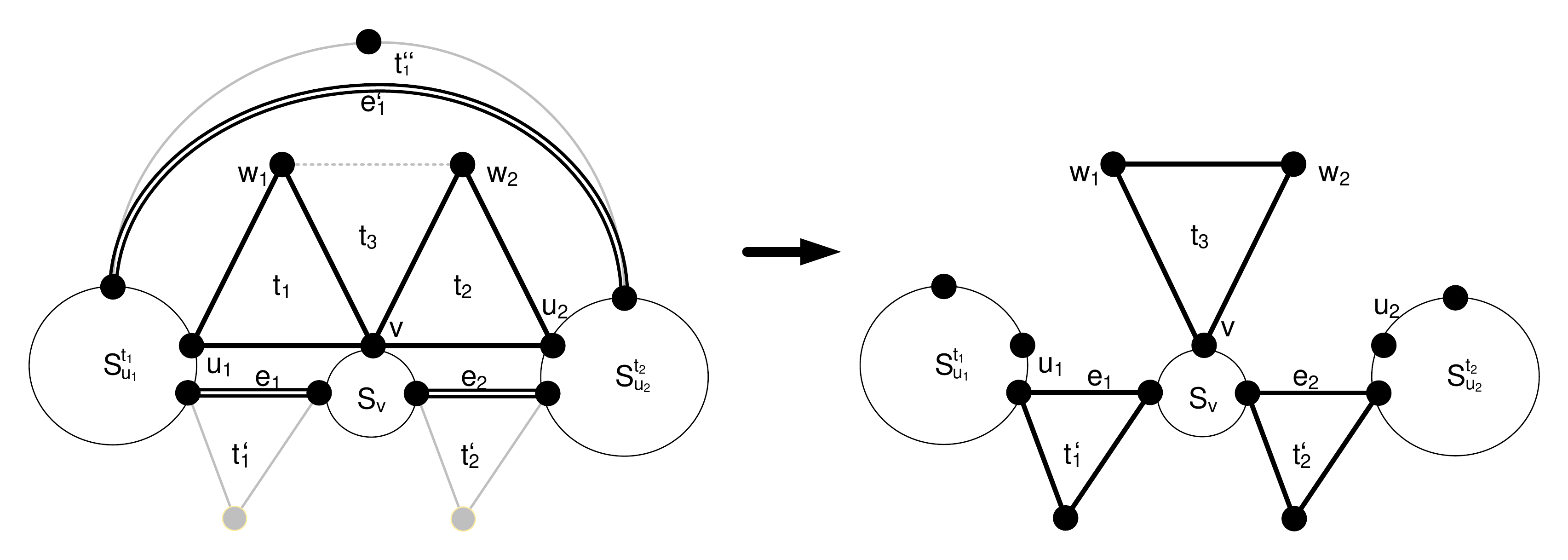}
        \caption{Improving $2$-swap when there exist two cross triangles $t_1', t_2'$ supported by $e_1, e_2$ respectively such that their landing components are different.}
        \label{fig:p0-2-swap1}
    \end{figure}
    
\item ($t_1''$ has a different landing component than the common landing component for $t_1', t_2'$): In this case we know that $t_1', t_2'$ share their landing components but the landing component for $t_1''$ is different. Again, since the base sides for $t_1, t_2$ are in different super-faces, Lemma~\ref{lem:different-one-blue-per-p-base}) implies that at least one of $e_1, e_2$ is incident on $S_v$. Now there are two sub-cases:
\begin{itemize}
    \item ($e_2$ incident on $S_v$): In this case, we modify our cactus by $\cset' = (\cset \setminus (E(t_1) \cup E(t_2))) \cup E(t_2') \cup E(t_1'') \cup E(t_3)$ (See Figure~\ref{fig:p0-2-swap2}). Again $t_2'$ will attach $S_v$ with $S_{u_2}^{t_2}$, $t_3$ will attach $S_v$ with $S_{w_1}^{t_1}$ and $S_{w_2}^{t_2}$ and finally $t_1''$ will attach $S_{u_1}^{t_1}$ to this structure, hence $\cset'$ will be a triangular cactus with one more cactus triangle, which contradicts the optimality of $\cset$.
    
\begin{figure}[H]
        \centering
        \includegraphics[width=0.9\textwidth]{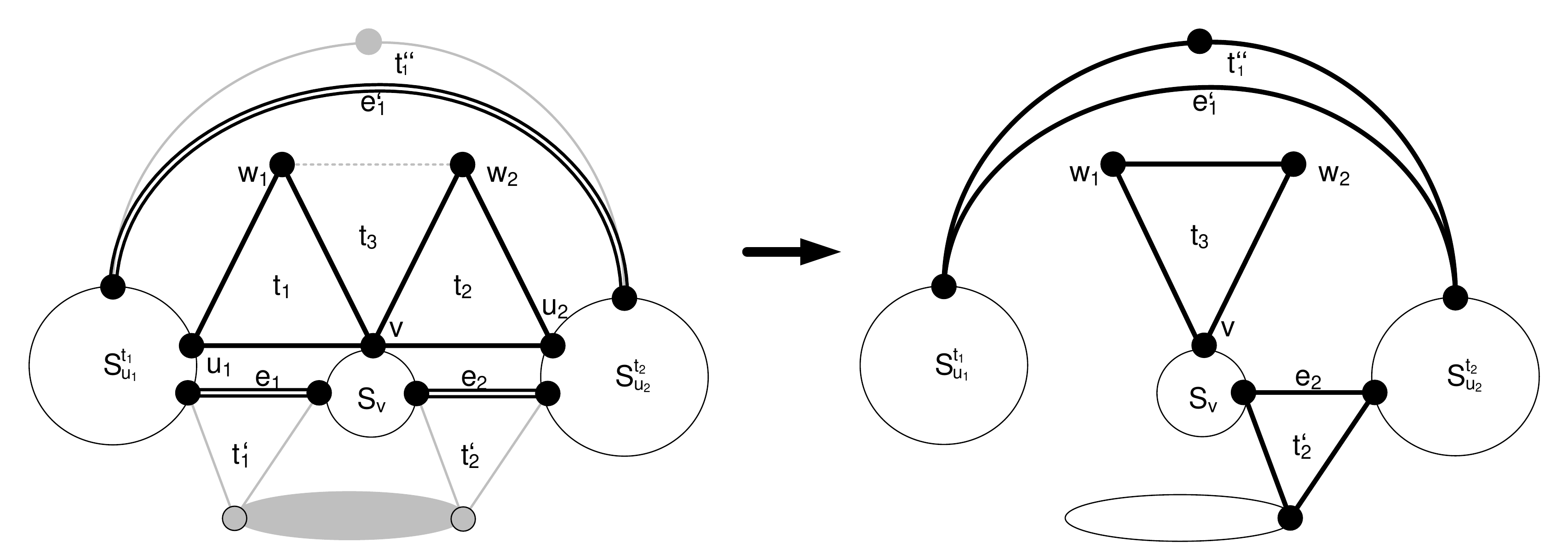}
        \caption{Improving $2$-swap when there exists two cross triangles $t_2', t_1''$ supported by $e_2, e_1'$ respectively such that their landing components are different.}
        \label{fig:p0-2-swap2}
    \end{figure}
    
    \item (Only $e_1$ incident on $S_v$): In this case $e_2$ goes across $S_{u_1}^{t_1}, S_{u_2}^{t_2}$.
    \begin{itemize}
        \item ($e_1'$ incident on $S_v$): The modification $\cset' = (\cset \setminus (E(t_1) \cup E(t_2))) \cup E(t_2') \cup E(t_1'') \cup E(t_3)$ gives us the contradiction since $t_1'$ will attach $S_{u_1}^{t_1}$ to  $S_v$, $t_3$ will attach $S_v$ with $S_{w_1}^{t_1}$ and $S_{w_2}^{t_2}$ and finally $t_2'$ will attach $S_{u_2}^{t_2}$ to this structure, hence $\cset'$ will be a triangular cactus with one more cactus triangle, which contradicts the optimality of $\cset$.
        \item ($e_1'$ goes across $S_{u_1}^{t_1}, S_{u_2}^{t_2}$): The modification $\cset' = (\cset \setminus (E(t_1) \cup E(t_2))) \cup E(t_1') \cup E(t_1'') \cup E(t_3)$ (See Figure~\ref{fig:p0-2-swap3}) gives us the contradiction since $t_1'$ will attach $S_{u_1}^{t_1}$ to $S_v$, $t_3$ will attach $S_v$ with $S_{w_1}^{t_1}$ and $S_{w_2}^{t_2}$ and finally $t_1''$ will attach $S_{u_2}^{t_2}$ to this structure, hence $C'$ will be a triangular cactus with one more cactus triangle, which contradicts the optimality of $\cset$.
        
    \begin{figure}[H]
        \centering
        \includegraphics[width=0.9\textwidth]{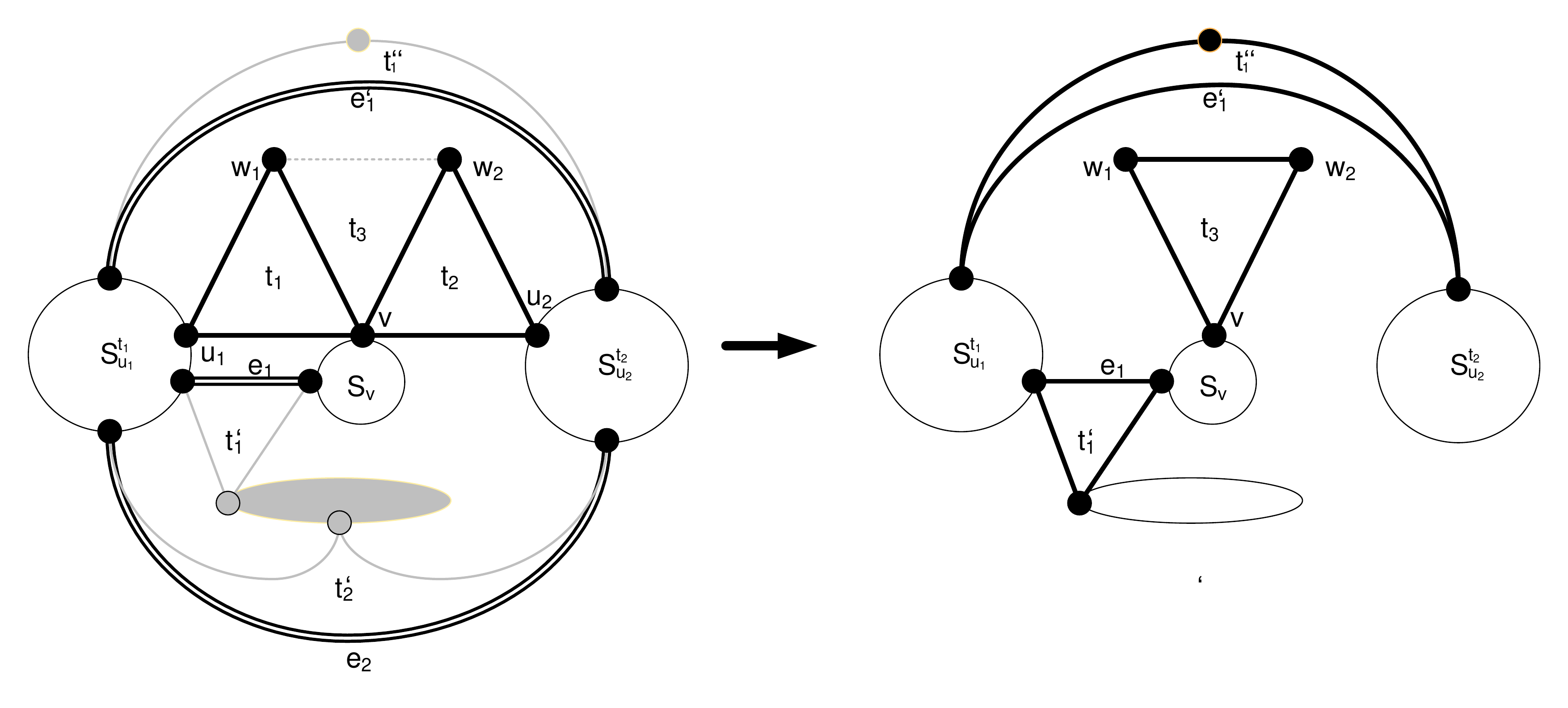}
        \caption{Improving $2$-swap when there exists two cross triangles $t_2', t_1''$ supported by $e_2, e_1'$ respectively such that their landing components are different.}
        \label{fig:p0-2-swap3}
    \end{figure}
        
    \end{itemize}
\end{itemize}

\item (All three triangles $t_1', t_2', t_1''$ share their landing components): In this case, all three $e_1, e_2, e_1'$ are type-$1$ edges. Also by Lemma~\ref{lem:different-one-blue-per-p-base}, at least one of $e_1, e_2$ will be incident on $S_v$. And since $S_{w_1}^{t_1}, S_{w_1}^{t_2}$ are free components, none of the three edges $e_1, e_2, e_1'$ can be incident on $S_{w_1}^{t_1}, S_{w_1}^{t_2}$. Based on these facts, there could be two sub-cases:
\begin{itemize}
    \item (Exactly one of $e_1$ or $e_2$ is incident on $S_v$): We will argue that this case cannot occur, by showing that there is no way for $t_1''$ to share the same landing component with $t_1', t_2'$. Since $t_1, t_2$ are friends, all the vertices of $S_v$ (except $v$) are drawn outside $t_3$. This also implies that there is a trail $P$ starting from vertex $u_1$, using all the cactus/type-$1$/type-$2$ edges on the outer-face for $H[S_v]$ and finally reaching $u_2$, such that the only repeated vertex in the trail is $v$. Since exactly one of $e_1$ or $e_2$ is incident on $S_v$, this implies that the other one goes across $S_{u_1}^{t_1}, S_{u_2}^{t_2}$ (say $u'v'$) such that $u' \in S_{u_1}^{t_1}$ and $v' \in S_{u_2}^{t_2}$. This means that there exists a circuit $C$ comprising of only cactus/type-$1$/type-$2$ edges formed by concatenating the trail $P$, the path $P_{u'}$ between $u'$ and $u_1$ using cactus edges/vertices only from $S_{u_1}^{t_1}$, the path $P_{v'}$ between $v'$ and $u_2$ using cactus edges/vertices only from $S_{u_2}^{t_2}$ and the type-$1$ or type-$2$ edge $u'v'$. It is easy to see that this circuit partitions the plane into two regions, say $R_1, R_2$, such that all the vertices of $S_v$ are drawn inside $R_1$ as a hole and the free sides for $t_1, t_2$ are drawn in $R_2$ such that the only vertex from $S_v$ on the boundary for these regions is $v$. Also, the presence of the edge $w_1 w_2$ does not allow the vertex $v$ to be a part of any type-$1$ or type-$2$ edge drawn inside $R_2$. This implies that the edge out of $e_1, e_2$ which is incident on $S_v$ will be drawn inside $R_1$ and $e_1'$ will be drawn outside $R_2$, which contradicts the fact that all three supported cross triangles $t_1', t_2', t_1''$ share their landing component.
    
    \begin{figure}[H]
        \centering
        \includegraphics[width=0.9\textwidth]{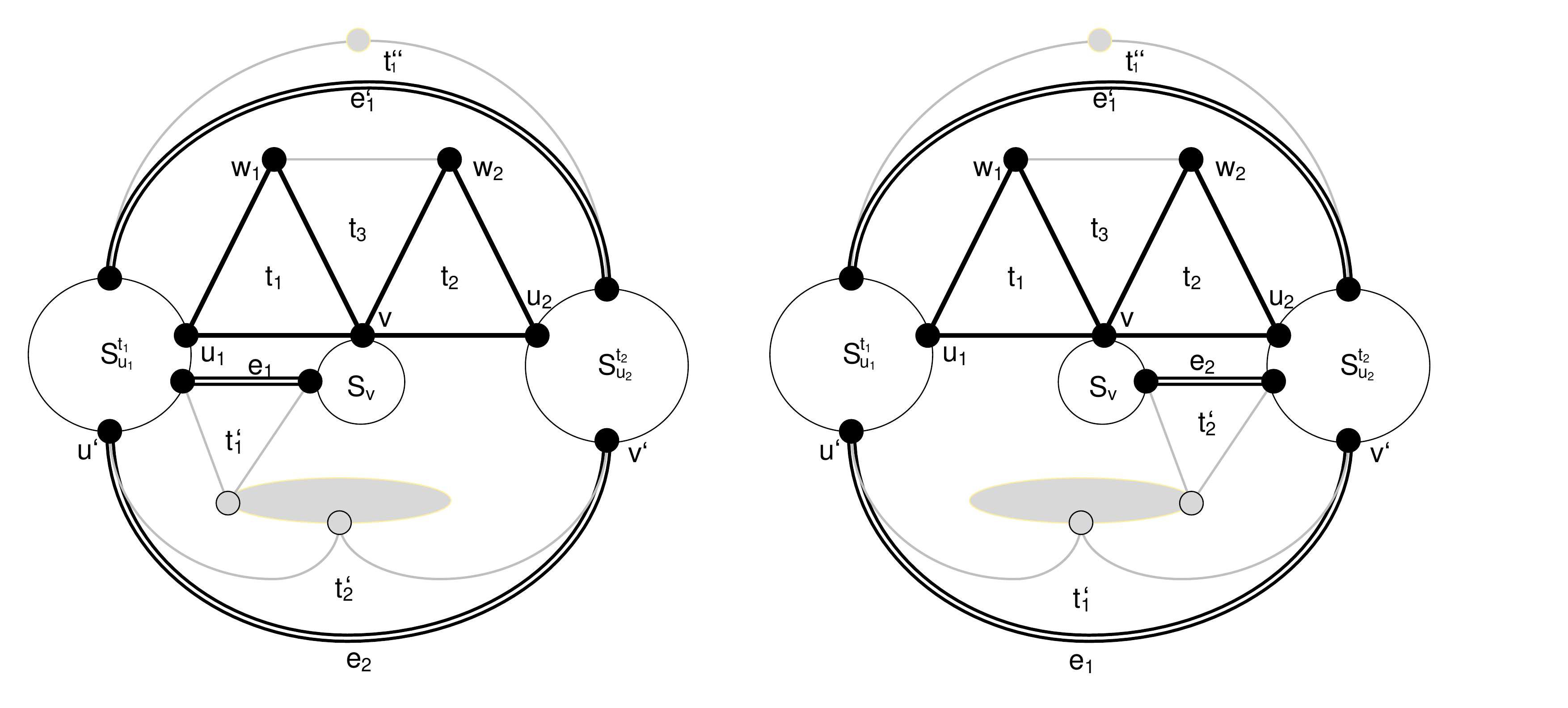}
        \caption{The setting before we reach a contradiction in the case where all three edges $e_1, e_2, e_1'$ are type-$1$ and the landing component for the respective supported cross triangles $t_1', t_2', t_1''$ is the same. Also, exactly one of $e_1$ or $e_2$ is incident on $S_v$. This case cannot occur since then $t_1''$ cannot reach the landing component of $t_1', t_2'$.}
        \label{fig:p0-$2$-swap-contradict-scenario1}
    \end{figure}
    
    \item (Both $e_1$ and $e_2$ are incident on $S_v$):
    Now we focus on $t_1$, which is a heavy type-$0$ triangle and look at the type-$1$ or type-$2$ edges going across $t_1$'s occupied components. The two type-$1$ edges $e_1$ and $e_1'$ are surely going across the occupied components of $t_1$. By Prop.~\ref{prop:structure-heavy}, $t_1$ should have at least one more such type-$1$ or type-$2$ edge (say $e_1'':=u'v'$). Now let $u' \in S_{u_1}^{t_1}$ and $v' \in S_{v_1}^{t_1}$. This means that there is a path $P_{u'}$ from $u'$ to $u_1$ in $\cset[S]$ and another path $P_{v'}$ from $v'$ to $v_1 = v$  in $\cset[S]$ such that the cycle $C_1 := u' P_{u'} u_1 \cup u_1 v \cup v P_{v'} v' \cup u'v'$ is made of only cactus/type-$1$/type-$2$ edges and cactus vertices such that it divided the plane into two regions such that one region contains the base side of $t_1$ and another contains the free side for $t_1$. Since, $e_1, e_1'$ see the base and free sides for $t_1$ respectively, hence they have to be drawn in the different region bounded by $C_1$. Hence, the cross triangles $t_1', t_1''$ supported by $e_1, e_1'$ cannot share their landing components, contradiction.
    
    \begin{figure}[H]
        \centering
        \includegraphics[width=0.9\textwidth]{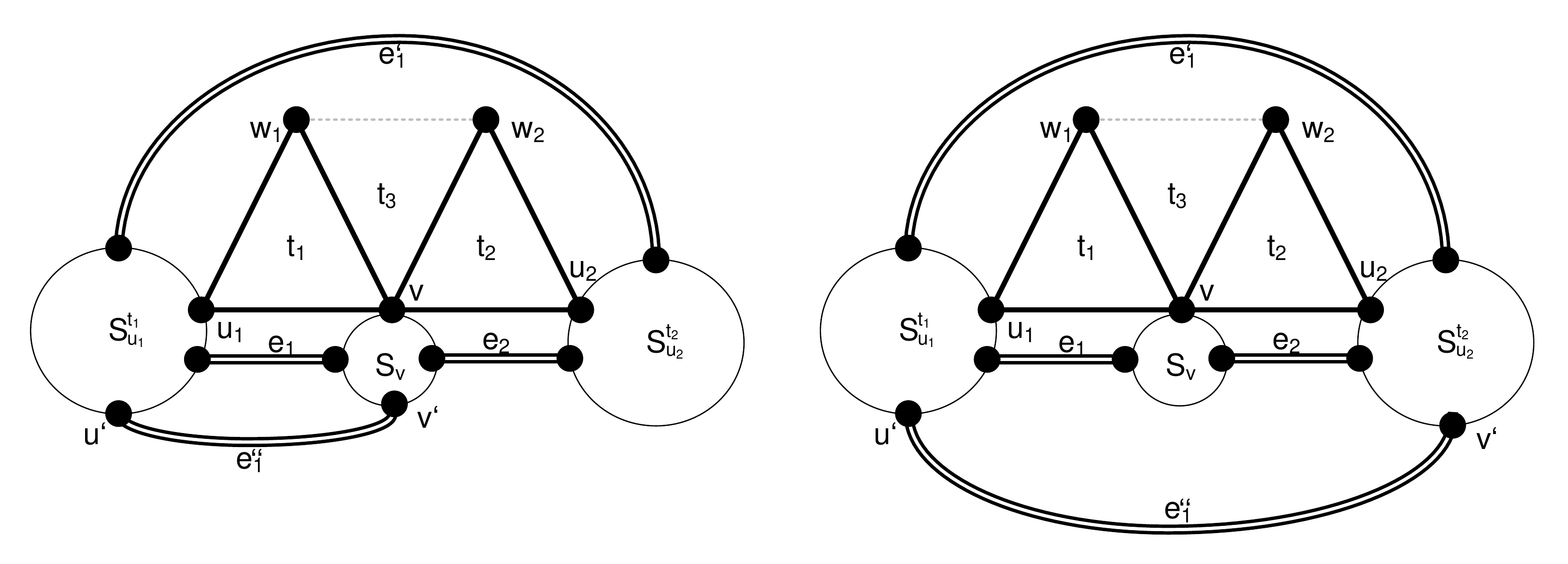}
        \caption{The setting before we reach a contradiction in the case where the edges $e_1, e_2, e_1'$ are type-$1$ and the landing component for the respective supported cross triangles $t_1', t_2', t_1''$ is the same. Also, both $e_1$ and $e_2$ are incident on $S_v$. This case cannot occur since $t_1$ is a type-$0$ heavy triangle and the additional type-$1$ or type-$2$ edge $e_1''$ going across $t_1$'s occupied components will separate $e_1, e_1'$ into different regions.}
        \label{fig:p0-$2$-swap-contradict-scenario2}
    \end{figure}
    
    \end{itemize}

\end{itemize}

\section{On the Strength of Our Result}\label{sec:extras}

\subsection{Our Bound is Almost Tight}
\label{sec:lower_bound}

In this section, we show that there exists a graph $G$ for which $\beta(G) \leq (\frac{1}{6}+o(1))f_3(G)$. We show this indirectly using a family of graphs presented in~\cite{CS17}, as stated in the following lemma. 

\begin{lemma}\label{lem:bad-1over6}\cite{CS17}
There is a family of $n$-vertex planar graphs $\{H_n\}_{n \in {\mathbb Z}}$ for which there exist a maximal cactus subgraph $C_n$ of $H_n$ such that $\frac{f_3(C_n)}{f_3(H_n)} \leq \frac{1}{12} + o_n(1)$. 
\end{lemma}

In \cite{CS17}, this family of graphs is used to show that a maximal cactus (not maximum) is not sufficient to improve over the best known greedy strategies when approximating MPT. In the context of this paper we use $C_n$ to compare it to a maximum cactus for $H_n$ to prove the following.

\begin{theorem}
Let $H_n$ be the graph family as in Lemma~\ref{lem:bad-1over6}. Then, $\frac{\beta(H_n)}{f_3(H_n)} \leq \frac{1}{6} + o_n(1)$. 
\end{theorem}

\begin{proof}
By Lemma \ref{lem:bad-1over6}, it suffices to argue that $f_3(C_n) \geq \frac{\beta(H_n)}{2}$.
Let $C^*_n$ be an optimal cactus with $\beta(H_n)$ triangles. 
Notice that for any triangle $t$ in $C_n$, $E(t)$ intersects at most two other triangles in $C^*_n$. If all three edges of $t$ were to be used by three different triangles in $C^*_n$, this would contradict the cactus property. Moreover, if $t$ does not intersect any triangle in $C^*_n$ this would imply that one of its edges would complete a cycle if added to $C^*_n$. By these two observations we can use a simple counting scheme to upper-bound the number of triangles in $C_n^*$ depending on the number of triangles in $C_n$.
We iteratively add triangles of $C_n$ to $C^*_n$ and count in every step how many triangles in $C^*_n$ need to be removed to maintain the cactus property. For every triangle in $C_n$ that intersects $C^*_n$ in one or two edges, we have to remove at most two triangles from $C^*_n$. For every triangle in $C_n$, that does not intersect $C^*_n$ in any edge, we have to break a cycle in the resulting $C^*_n$ by deleting one other triangle from it. In each iteration we therefore destroy at most two triangles from the original $C^*_n$ and therefore get $f_3(C^*_n) \leq 2f_3(C_n)$. This concludes the proof as
$f_3(C_n) \geq f_3(C^*_n)/2 = \beta(H_n)/2$.  
\end{proof}

\subsection{Comparison to the Previous Bound}

One integral part to derive the improved approximation ration for MPS in~\cite{cualinescu1998better} was to show that for any given planar graph $G=(V,E)$ with $n=|V|$ vertices and $|E|=3n-6-t(G)$ edges, we have:

\begin{theorem}[\cite{cualinescu1998better}]
\label{thm:cualinescu_cactus_extremal}
Let $G$ be as above, then $\beta(G) \geq \frac{1}{3} (n - t(G) - 2)$.
\end{theorem}

As removing one edge from a triangulated planar graph merges exactly two faces, we can easily derive a lower bound that depends on $t(G)$, for the number of triangular faces in $G$: 

\[ f_3(G) \geq 2n - 2t(G) - 4 \]

By Theorem~\ref{thm:main}, we have that $\beta(G) \geq \frac{1}{6} f_3(G)$. Combining these two facts implies Theorem~\ref{thm:cualinescu_cactus_extremal}. 

\begin{figure}[!ht]
  \centering
   \includegraphics[width=0.5\textwidth]{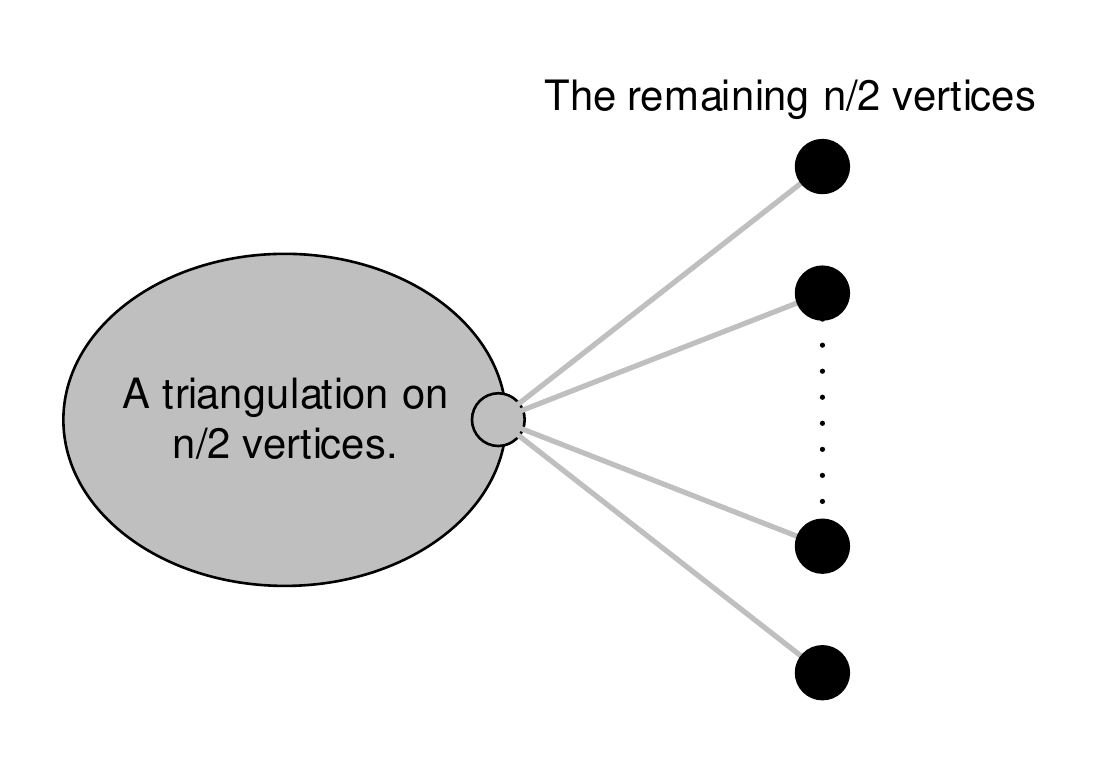}
    \caption{Bad example which shows that a extremal bound like the one in in~\cite{cualinescu1998better} for MPS does not necessarily imply a similarly strong result to MPT.}
    \label{fig:BadExampleMPT}
\end{figure}

We end this section by showing that the bound in~\cite{cualinescu1998better} alone is not sufficient for approximating MPT. 
To this end we construct a graph in which $\frac{1}{3} (n - t(G) - 2) \leq 0$, even though $f_3(G) = \Theta(n)$, 
Let $G$ be a planar graph with $n$ vertices, where $\frac{n}{2}$ vertices form a triangulated planar subgraph. Let $v$ be a vertex on the outer-face of this triangulated structure. The remaining $\frac{n}{2}$ vertices are embedded in the outer-face and are incident to exactly one edge each, with the other endpoint being $v$ (see Figure~\ref{fig:BadExampleMPT} for an illustration of this construction). 
Therefore by Euler's formula, the number of edges in this graph is equal to $3(\frac{n}{2}) - 6 + \frac{n}{2} = 2n - 6$ and thus $t(G) = n$, while the number of triangular faces is $f_3(G) = 2(\frac{n}{2}) - 4 -1 = n - 5$. 

\section{Conclusions and Open Problems} 
\label{sec:conclusion}

Our work implies that a natural local search algorithm gives a $(\frac{4}{9}+ \epsilon)$-approximation for MPS and a $\frac{1}{6}+\epsilon$ approximation for MPT. To be more precise, when given any graph $G$, we follow the $t$-swap local search strategy for $t= O(1/\epsilon)$: Start from any cactus subgraph $H$. Try to improve it by removing $t$ triangles and adding $(t+1)$ triangles in a way that ensures that the graph remains a cactus subgraph. 
A local optimal solution will always be a $(\frac{4}{9}+\epsilon)$ approximation for MPS and a $(\frac 1 6 + \epsilon)$ approximation for MPT. 

Knowing this fact, there is an obvious candidate algorithm for improving over the long-standing best approximation factor for MPS. 
We call a graph $H$ a diamond-cactus if every block in $H$ is either a diamond\footnote{A diamond subgraph is a graph that is isomorphic to the graph resulting from deleting any single edge from a $K_4$.} or a triangle.
Start from any diamond-cactus subgraph $H$ of $G$ and then try to improve it by removing $t$ triangles from $H$ and adding $(t+1)$ triangles, maintaining the fact that $H$ is a diamond-cactus subgraph. 
We conjectured that this algorithm gives a better than $\frac{4}{9}$-approximation for MPS, but we suspect that the analysis will require substantially new ideas. 

Another interesting direction is to see whether there is a general principle that captures a denser planar structure than cactus subgraphs by going above matroid parity in the hierarchy of efficiently computable problems. For instance, are diamond-cactus subgraphs captured by matroid parity? Or can it be formulated as an even more abstract structure than matroids (e.g. commutative rank~\cite{BlaserJP17}) that can still be computed efficiently?
We believe that studying this direction will lead to a better understanding of algebraic techniques for finding dense planar structures. 

Finally, the absence of LP-based techniques in this problem domain seems rather unfortunate. There have been some experimental studies recently, but the theoretical understanding of what can be proven formally in the context of power of relaxation is certainly lacking \cite{Juenger1996,ChimaniLP_SEA18,chimani_et_al_LP_ESA18}.
Is there a convex relaxation that allows us to find a relatively dense planar subgraph (e.g. $(3-\epsilon)$-approximation for MPS using LP-based techniques)?

\paragraph{Acknowledgement:} Parinya Chalermsook is supported by the European Research Council (ERC) under the European Union’s Horizon 2020 research and innovation programme (grant agreement No 759557) and by Academy of Finland Research Fellows, under grant number 310415 and 314284. Sumedha Uniyal is partially supported by Academy of Finland under the grant agreement number 314284.

\bibliographystyle{plain}
\bibliography{references}

\begin{thebibliography}{10}

\bibitem{berge1958theorie}
Claude Berge.
\newblock La theorie des graphes.
\newblock {\em Paris, France}, 1958.

\bibitem{BlaserJP17}
Markus Bl{\"{a}}ser, Gorav Jindal, and Anurag Pandey.
\newblock Greedy strikes again: {A} deterministic {PTAS} for commutative rank
  of matrix spaces.
\newblock In {\em 32nd Computational Complexity Conference, {CCC} 2017, July
  6-9, 2017, Riga, Latvia}, pages 33:1--33:16, 2017.

\bibitem{cualinescu1998better}
Gruia C{\u{a}}linescu, Cristina~G Fernandes, Ulrich Finkler, and Howard
  Karloff.
\newblock A better approximation algorithm for finding planar subgraphs.
\newblock {\em Journal of Algorithms}, 27(2):269--302, 1998.

\bibitem{calinescu2003new}
Gruia Calinescu, Cristina~G Fernandes, Howard Karloff, and Alexander
  Zelikovsky.
\newblock A new approximation algorithm for finding heavy planar subgraphs.
\newblock {\em Algorithmica}, 36(2):179--205, 2003.

\bibitem{cualinescu2012maximum}
Gruia C{\u{a}}linescu, Cristina~G Fernandes, Hemanshu Kaul, and Alexander
  Zelikovsky.
\newblock Maximum series-parallel subgraph.
\newblock {\em Algorithmica}, 63(1-2):137--157, 2012.

\bibitem{chalermsook2017gap}
Parinya Chalermsook, Marek Cygan, Guy Kortsarz, Bundit Laekhanukit, Pasin
  Manurangsi, Danupon Nanongkai, and Luca Trevisan.
\newblock From gap-eth to fpt-inapproximability: Clique, dominating set, and
  more.
\newblock In {\em Foundations of Computer Science (FOCS), 2017 IEEE 58th Annual
  Symposium on}, pages 743--754. IEEE, 2017.

\bibitem{CS17}
Parinya Chalermsook and Andreas Schmid.
\newblock Finding triangles for maximum planar subgraphs.
\newblock In {\em {WALCOM:} Algorithms and Computation, 11th International
  Conference and Workshops, {(WALCOM'17)}, Proceedings.}, pages 373--384, 2017.

\bibitem{CSU19-stacs}
Parinya Chalermsook, Andreas Schmid, and Sumedha Uniyal.
\newblock {A Tight Extremal Bound on the Lov{\'a}sz Cactus Number in Planar
  Graphs}.
\newblock In {\em 36th International Symposium on Theoretical Aspects of
  Computer Science (STACS 2019)}, volume 126 of {\em Leibniz International
  Proceedings in Informatics (LIPIcs)}, pages 19:1--19:14. Schloss
  Dagstuhl--Leibniz-Zentrum fuer Informatik, 2019.

\bibitem{cheung2014algebraic}
Ho~Yee Cheung, Lap~Chi Lau, and Kai~Man Leung.
\newblock Algebraic algorithms for linear matroid parity problems.
\newblock {\em ACM Transactions on Algorithms (TALG)}, 10(3):10, 2014.

\bibitem{ChimaniLP_SEA18}
Markus Chimani, Ivo Hedtke, and Tilo Wiedera.
\newblock Exact algorithms for the maximum planar subgraph problem: New models
  and experiments.
\newblock In {\em 17th International Symposium on Experimental Algorithms,
  {SEA} 2018, June 27-29, 2018, L'Aquila, Italy}, pages 22:1--22:15, 2018.

\bibitem{chimani_et_al_LP_ESA18}
Markus Chimani and Tilo Wiedera.
\newblock {Cycles to the Rescue! Novel Constraints to Compute Maximum Planar
  Subgraphs Fast}.
\newblock In Yossi Azar, Hannah Bast, and Grzegorz Herman, editors, {\em 26th
  Annual European Symposium on Algorithms (ESA 2018)}, volume 112 of {\em
  Leibniz International Proceedings in Informatics (LIPIcs)}, pages
  19:1--19:14, Dagstuhl, Germany, 2018. Schloss Dagstuhl--Leibniz-Zentrum fuer
  Informatik.

\bibitem{gabow1986augmenting}
Harold~N Gabow and Matthias Stallmann.
\newblock An augmenting path algorithm for linear matroid parity.
\newblock {\em Combinatorica}, 6(2):123--150, 1986.

\bibitem{halldorsson2004approximations}
Magn{\'u}s~M Halld{\'o}rsson.
\newblock Approximations of weighted independent set and hereditary subset
  problems.
\newblock In {\em Graph Algorithms And Applications 2}, pages 3--18. World
  Scientific, 2004.

\bibitem{haastad1999clique}
Johan H{\aa}stad.
\newblock Clique is hard to approximate withinn 1- $\varepsilon$.
\newblock {\em Acta Mathematica}, 182(1):105--142, 1999.

\bibitem{jenkyns1974matchoids}
T.A. Jenkyns.
\newblock {\em Matchoids : a Generalization of Matchings and Matroids}.
\newblock Thesis (Ph.D.)--University of Waterloo, 1974.

\bibitem{Juenger1996}
M.~J{\"u}nger and P.~Mutzel.
\newblock Maximum planar subgraphs and nice embeddings: Practical layout tools.
\newblock {\em Algorithmica}, 16(1):33--59, Jul 1996.

\bibitem{khot2006better}
Subhash Khot and Ashok~Kumar Ponnuswami.
\newblock Better inapproximability results for maxclique, chromatic number and
  min-3lin-deletion.
\newblock In {\em International Colloquium on Automata, Languages, and
  Programming}, pages 226--237. Springer, 2006.

\bibitem{lawler1976combinatorial}
Eugene~L Lawler.
\newblock {\em Combinatorial optimization: networks and matroids}.
\newblock Courier Corporation, 1976.

\bibitem{lee2013matroid}
Jon Lee, Maxim Sviridenko, and Jan Vondr{\'a}k.
\newblock Matroid matching: the power of local search.
\newblock {\em SIAM Journal on Computing}, 42(1):357--379, 2013.

\bibitem{lewis1980node}
John~M Lewis and Mihalis Yannakakis.
\newblock The node-deletion problem for hereditary properties is np-complete.
\newblock {\em Journal of Computer and System Sciences}, 20(2):219--230, 1980.

\bibitem{lovasz1980matroid}
L{\'a}szl{\'o} Lov{\'a}sz.
\newblock Matroid matching and some applications.
\newblock {\em Journal of Combinatorial Theory, Series B}, 28(2):208--236,
  1980.

\bibitem{lovasz2009matching}
L{\'a}szl{\'o} Lov{\'a}sz and Michael~D Plummer.
\newblock {\em Matching theory}, volume 367.
\newblock American Mathematical Soc., 2009.

\bibitem{lund1993approximation}
Carsten Lund and Mihalis Yannakakis.
\newblock The approximation of maximum subgraph problems.
\newblock In {\em International Colloquium on Automata, Languages, and
  Programming}, pages 40--51. Springer, 1993.

\bibitem{orlin2008fast}
James~B Orlin.
\newblock A fast, simpler algorithm for the matroid parity problem.
\newblock In {\em International Conference on Integer Programming and
  Combinatorial Optimization}, pages 240--258. Springer, 2008.

\bibitem{szigeti1998min}
Zolt{\'a}n Szigeti.
\newblock On a min-max theorem of cacti.
\newblock In {\em International Conference on Integer Programming and
  Combinatorial Optimization}, pages 84--95. Springer, 1998.

\bibitem{tutte1947factorization}
William~T Tutte.
\newblock The factorization of linear graphs.
\newblock {\em Journal of the London Mathematical Society}, 1(2):107--111,
  1947.

\end{thebibliography}

\newpage
\appendix

\section{Missing Proofs} 
\label{sec:missing-proofs}

\subsection{Proof of Proposition \ref{prop:structure-heavy}}
\label{subsec:proof-prop-heavy}

Observation \ref{Observation:different-components-green-blue-cycle} immediately leads to a simple lemma which will prove helpful in the proof of Proposition \ref{prop:structure-heavy}.

\begin{lemma}
\label{lem:type-2-edge-support}
Let $e := uv$ be a type-$2$ edge in $G[S]$, then the cross triangles $t_1$ and $t_2$ supported by $e$ can not have the same landing component.
\end{lemma}

\begin{proof}
Since both $u$ and $v$ are in $S$, there exists a path $P$ from $u$ to $v$ in $G[S]$ containing only cactus edges and vertices (see for example Figure~\ref{fig:classification-edges-vertices-triangles}). Hence, the cycle $D := uPv \cup vu$ consists of only type-$1$ or type-$2$ or cactus edges and cactus vertices, such that the two cross triangles supported by $e$ will be drawn in different regions corresponding to $D$. Thus by Observation~\ref{Observation:different-components-green-blue-cycle} the two cross triangles supported by $e$ cannot have the same landing component.
\end{proof}

The $1$-swap introduced in Figure~\ref{fig:1-swap} and the $2$-swap optimality of $\cset$, implies the following lemma.

    \begin{figure}[H]
    \centering
    \includegraphics[width=0.35\textwidth]{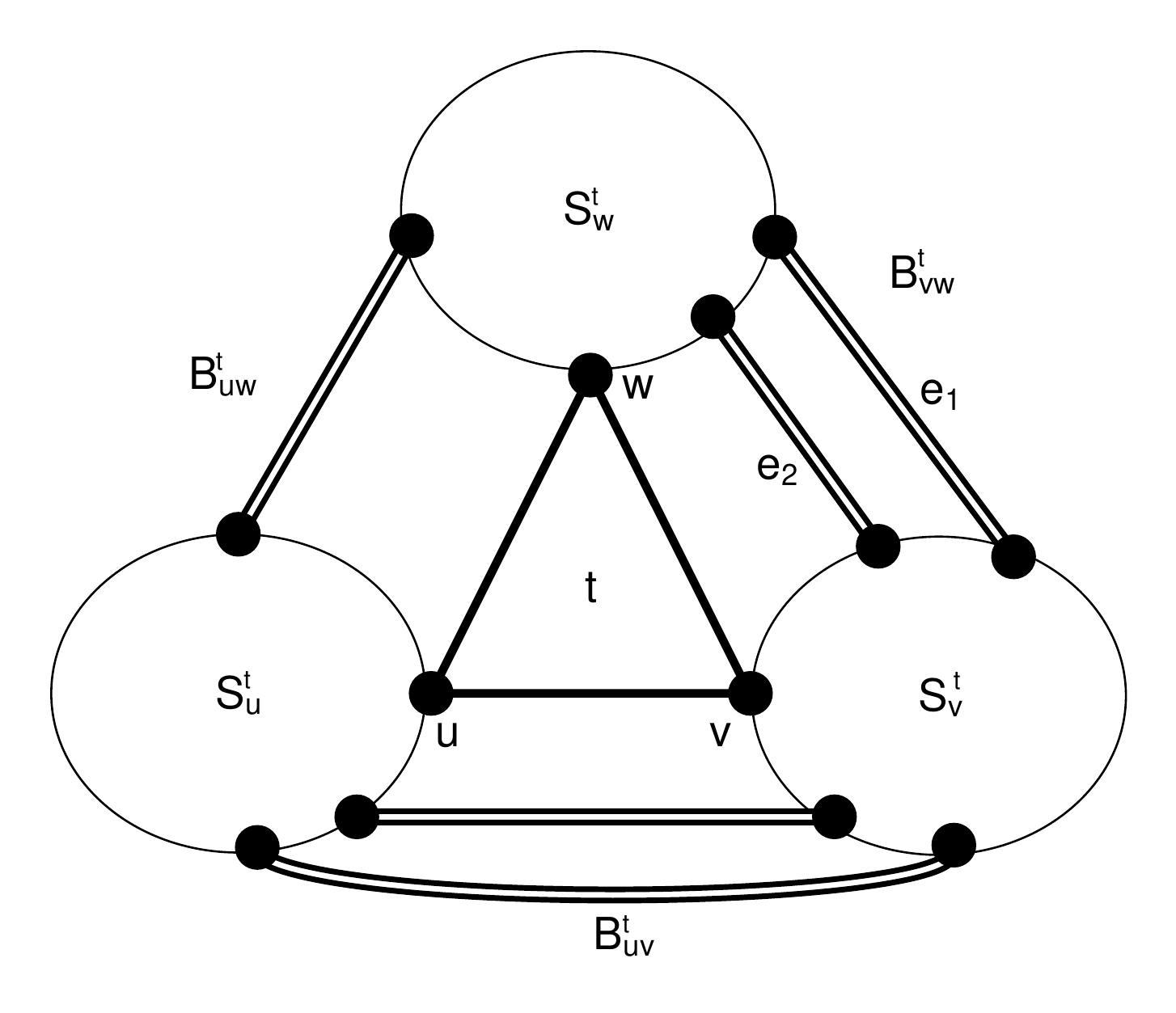}
        \caption{The split components $S_u^t, S_v^t, S_w^t$ and the sets $B_{uv}^t, B_{uw}^t, B_{vw}^t$ for a cactus triangle $t$. In fact there cannot be two edges $e_1, e_2$ as shown by Lemma \ref{lem:green-triangles-supporting} property~\ref{itm:at-most-one-blue-edge-each}.}
        \label{fig:split-components-t}
    \end{figure}

\begin{lemma}
\label{lem:green-triangles-supporting}
Let $t$ be a cactus triangle with vertices $u,v$ and $w$ and let there exist at least two cross triangles $t_1,t_2$ in $G$ such that $(V(t_1)\cup V(t_2)) \cap S_x^t \neq \emptyset$, for $x\in \{u,v,w\}$, then 
\begin{enumerate}[i]
    \item \label{itm:same-landing}  $t_1$ and $t_2$ must have the same landing component,
    \item \label{itm:all-type-1-blue-edges} any edge $e$ in $B_{uv}^t \cup B_{uw}^t \cup B_{vw}^t$ is of type-$1$, 
    \item \label{itm:at-most-one-blue-edge-each} $|B_{uv}^t|,|B_{uw}^t|,|B_{vw}^t| \leq 1$ and
    \item \label{itm:at-most-one-gray-each} any set of edges $\{xy\} \cup B_{xy}^t$ for $xy \in E(t)$ support at most one cross triangle.
\end{enumerate}
\end{lemma}

\begin{proof}
To prove property (\ref{itm:same-landing}), assume for contradiction that $t_1$ and $t_2$ do not share the same landing component. In this case we can increase the number of triangles in $\cset$ by removing $t$ from $\cset$ and adding $t_1$ and $t_2$ to $\cset$ in its place. As the landing components are disjoint this operation does not introduce any new cycle to $\cset$ other than the supported cross triangles and therefore the resulting structure is a cactus subgraph of $G$. This contradicts that $\cset$ is $2$-swap optimal.

Property (\ref{itm:all-type-1-blue-edges}) follows from property (\ref{itm:same-landing}). Assume for contradiction that there exists a type-$2$ edge $e \in B_{uw}^t$ (the same argument will hold for $B_{uv}^t$ and $B_{vw}^t$). Only one of $t_1$ and $t_2$ can have its two cactus vertices in the same split components of $t$ as the endpoints of $e$. We may assume that this is not the case for $t_1$. Let $t'$ and $t''$ denote the cross triangles supported by $e$. By property (\ref{itm:same-landing}) $t'$ and $t_1$ must have the same landing component, the same holds for $t''$ and $t_1$. But by Lemma~\ref{lem:type-2-edge-support} $t'$ and $t''$ can not have the same landing component, thus we reach a contradiction.

We will prove property (\ref{itm:at-most-one-blue-edge-each}) also by contradiction. Assume that $|B_{vw}^t| \geq 2$ and let $e_1, e_2 \in B_{vw}^t$ be any two type-$1$ edges (the same argument will hold for $B_{uv}^t$ and $B_{uw}^t$). 
As both endpoints of $e_1$ are cactus vertices, there exists a path in $\cset[S]$ connecting both the endpoints, thus there is a cycle $C_1$ in $G[S]$ containing $e_1$ and only cactus edges otherwise. Similarly, there exists a cycle $C_2$ in $G[S]$ that contains $e_2$ and only cactus edges otherwise. In $G$ either $e_1$ is embedded in the inside of the closed region bounded by $C_2$ or $e_2$ is embedded in the inside of the closed region bounded by $C_1$ (see Figure~\ref{fig:split-components-t}). We may assume the former case as the proof for the later case is symmetric.

Only one of $t_1$ and $t_2$ can have its two cactus vertices in the same split components of $t$ as the endpoints of $e_1$. We may assume that this is not the case for $t_1$. By property (\ref{itm:same-landing}) the cross triangle supported by $e_1$ and $t_1$ must have the same landing component. Note $t_1$ can not lay in the inside of the region bounded by $C_1$ in $G$. Therefore the landing component shared by the two cross triangles must lie on the outside of $C_1$. However, by property (\ref{itm:same-landing}) the cross triangle supported by $e_2$ and $t_1$ must have the same landing component. We reach a contradiction using Obs.~\ref{Observation:different-components-green-blue-cycle}. 

We prove property (\ref{itm:at-most-one-gray-each}) also by contradiction. Assume that the set of edges $\{uv\} \cup B_{uv}^t$ supports two cross triangles (the same argument will hold for $\{uw\} \cup B_{uw}^t$ and $\{vw\} \cup B_{vw}^t$). Property (\ref{itm:at-most-one-blue-edge-each}) implies that there is only one type-$1$ edge in $B_{uv}^t$ hence $uv$ will support the other cross triangle. Let $t'$ be the triangles supported by $uv$, $t''$ be the cross triangle supported by an edge $e' \in B_{uv}^t$. Only one of $t_1$ and $t_2$ can have its two cactus vertices in the same split components of $t$ as the endpoints of $e$. We may assume that this is not the case for $t_1$. By property (\ref{itm:same-landing}), $t'$ and $t_1$ must have the same landing component. But this is also true for $t''$ and $t_1$. In addition there is a cycle $C$ in $H[S]$ that contains $e$ and a path $P$ from $u'$ to $v'$ in $\cset[S]$ (where $u'v' = e$) containing only cactus vertices and edges such that $t'$ is embedded in its inside in $G$ and $S_w^t$ outside of it. As $t_1$ intersects $S_w^t$ it must be drawn outside of $C$ in $G$. But by Observation \ref{Observation:different-components-green-blue-cycle}, $t'$ and $t_2$ cannot have the same landing component and we reach a contradiction.
\end{proof}

Further we can show that

\begin{lemma}\label{lem:split_region}
If $t$ supports cross triangles $t_1$ and $t_2$, where $u$ denotes their common cactus vertex, then $B_{uv}^t$ and $B_{uw}^t$ are both empty. 
\end{lemma}

\begin{proof}
By Lemma~\ref{lem:green-triangles-supporting}, item (\ref{itm:same-landing}), $t_1$ and $t_2$ must have the same landing component. Note that if $t_1$ and $t_2$ have a common landing vertex, then the claim is trivially true, as then $u$ is incident to exactly three faces, namely $t$,$t_1$ and $t_2$ which by definition are all empty and thus $B_{uv}^t$ and $B_{uw}^t$ are empty in this case.
Thus we assume that $t_1 \cap t_2 = u$.

    \begin{figure}[H]

        \centering
        \includegraphics[width=0.35\textwidth]{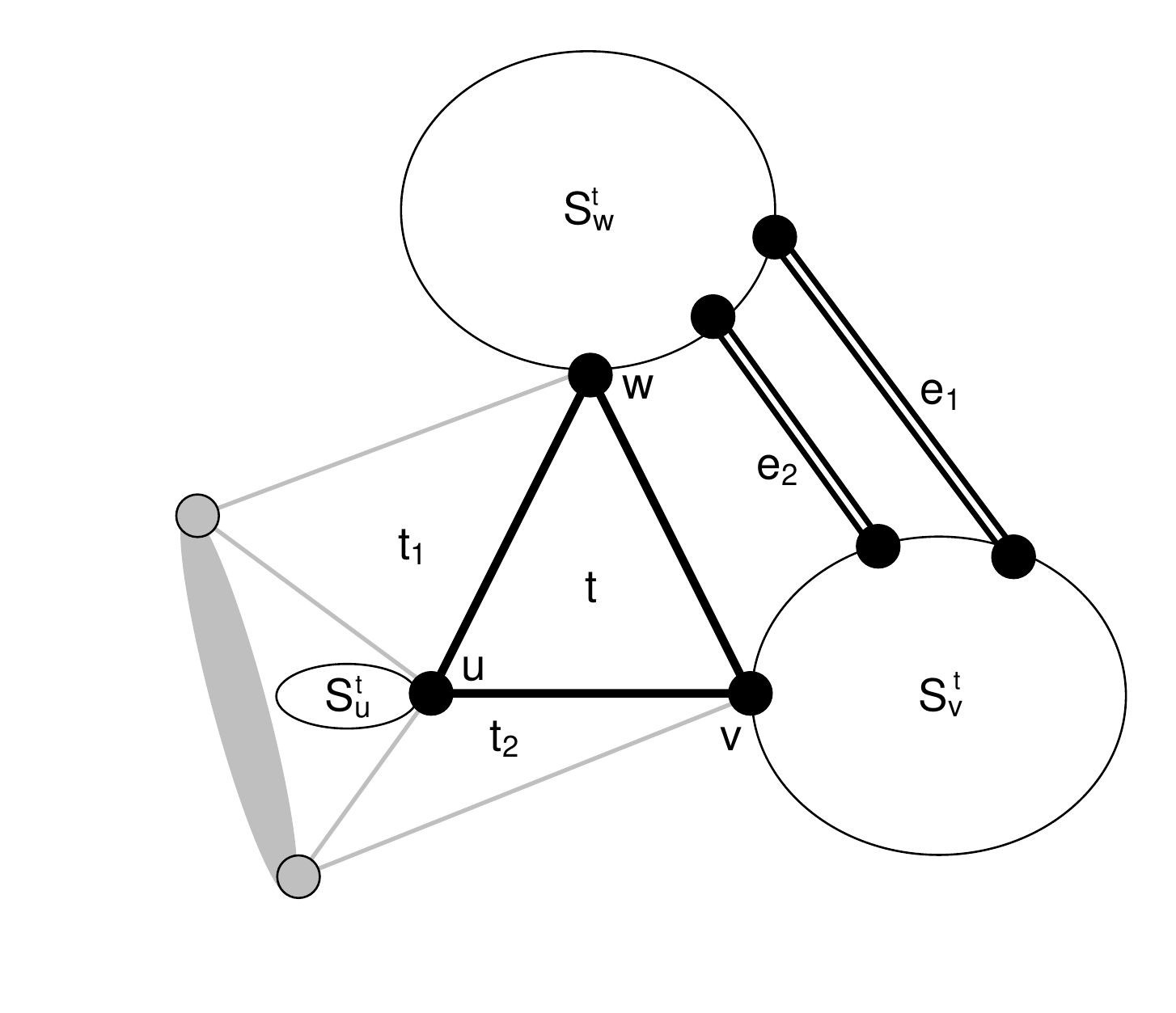}
        \caption{When $t$ supports two cross triangles with a common vertex $u$, then $B_{uv}^t = B_{uw}^t= \emptyset$ as shown by Lemma~\ref{lem:split_region}.}
        \label{fig:empty-E_ij-two-supported-trinangles-t}
    \end{figure}

Let $u_1$ and $u_2$ denote the landing vertices of $t_1$ and $t_2$ respectively. As $t_1$ and $t_2$ have the same landing component (say $S'$), $G$ must contain a path $P$ from $u_1$ to $u_2$ consisting of edges only in $\cset[S']$. Furthermore $uu_1 \cup P \cup u_2u$ forms a cycle $C$ with only one cactus vertex $u$ and cross edges and edges in $\cset[S']$.
Note that the fact that $t,t_1$ and $t_2$ are empty in $G$, implies that the two cactus edges of $t$ incident to $u$, as well as the edges $uu_1$ and $uu_2$ are consecutive in the circular edge incident list of $u$ in $G$. This observation gives us two important facts. First, as $C$ contains $uu_1$ and $uu_2$, any other edge incident to $u$ in $G$ must be drawn in the region bounded inside of $C$ in $G$. Second, any split component $S_x^t$, for $x\in \{v,w\}$, must be drawn outside of $C$ in $G$.
Assume for contradiction that there exists an edge $e$ with endpoints $u$ and $z \in S_{x}^t$, with $x \in \{v,w\}$, by the previous observation $e$ has to cross $C$ in $G$ and therefore the existence of $e$ contradicts that $G$ is a plane graph.
Similarly, there cannot exist any edge $e$ with one endpoint in $S_u^t \setminus \{u\}$ and another endpoint in $S_x^t$, with $x \in \{v,w\}$, since all these vertices are drawn {\em strictly} inside of $C$ and  $S_x^t$'s, with $x \in \{v,w\}$, are drawn strictly outside of $C$.
\end{proof}

\begin{figure}[H]
  \centering
    \includegraphics[width=0.35\textwidth]{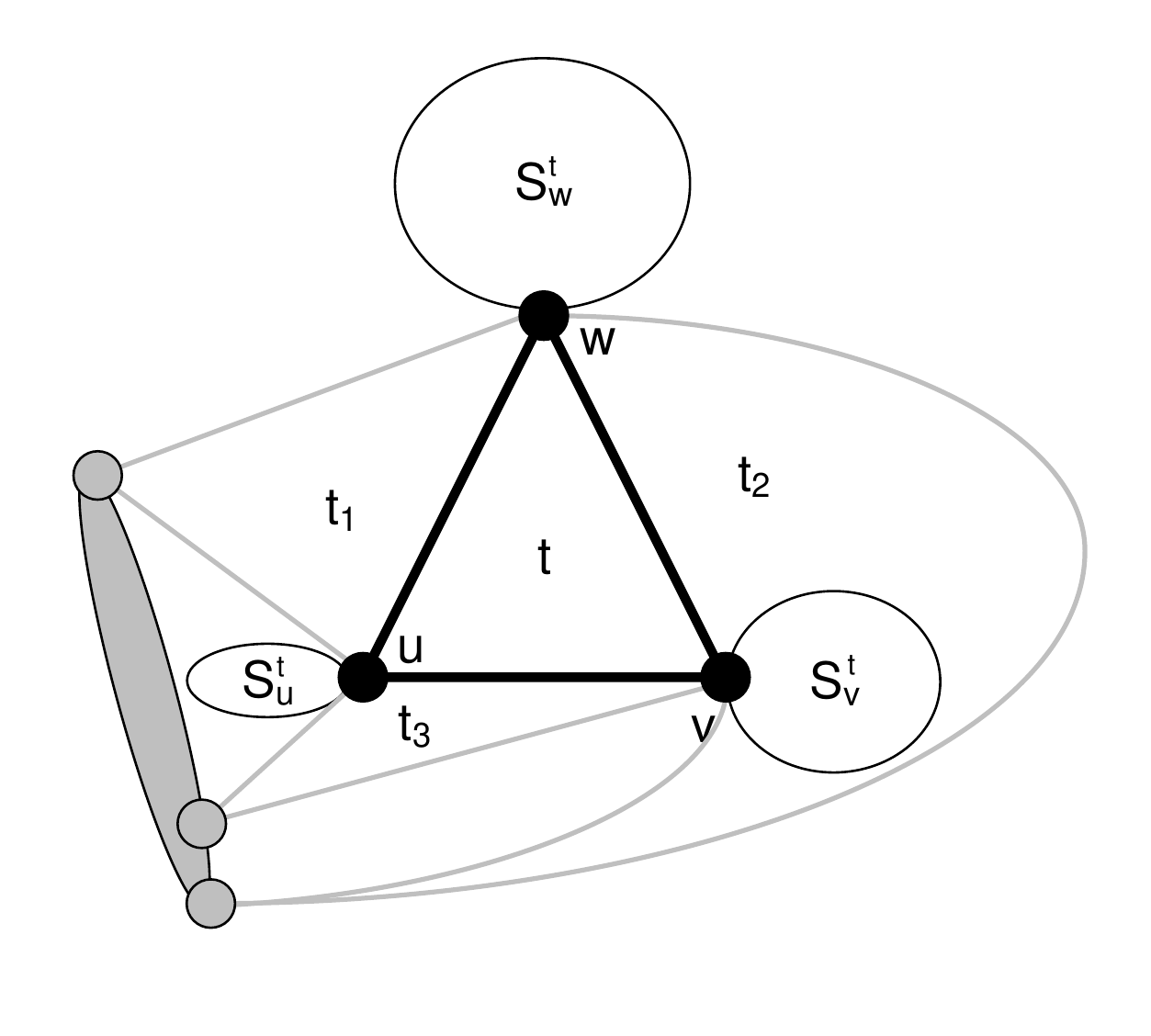}
    \caption{A type-$3$ light triangle $t$ for which the third property of Prop.~\ref{prop:structure-light} holds.}
    \label{fig:type-3-light}
\end{figure}

We are now ready to prove the different properties of heavy triangles claimed in Proposition \ref{prop:structure-heavy}. In the following we will prove one lemma for every such claim.

\begin{lemma}\label{lem:type3_light}
Any cactus type-$3$ triangle $t$ in $G[S]$ is light.
\end{lemma}
\begin{proof}

For any vertex $v$ in $t$, there is a pair of cross triangles supported by $t$ such that their intersection is $v$, thus by Lemma~\ref{lem:split_region}, $B_{vv'}^t$ must be empty for any $v' \in V(t) \setminus v$. Hence the number of cross triangles supported by $E(t)$ and $\cup_{ww'\in E(t)}B_{ww'}^t$ is less than four and each edge $vv' \in E(t)$ supports one cross triangle, thus $t$ is a light triangle.
\end{proof}

\begin{figure}[H]
  \centering
    \includegraphics[width=0.35\textwidth]{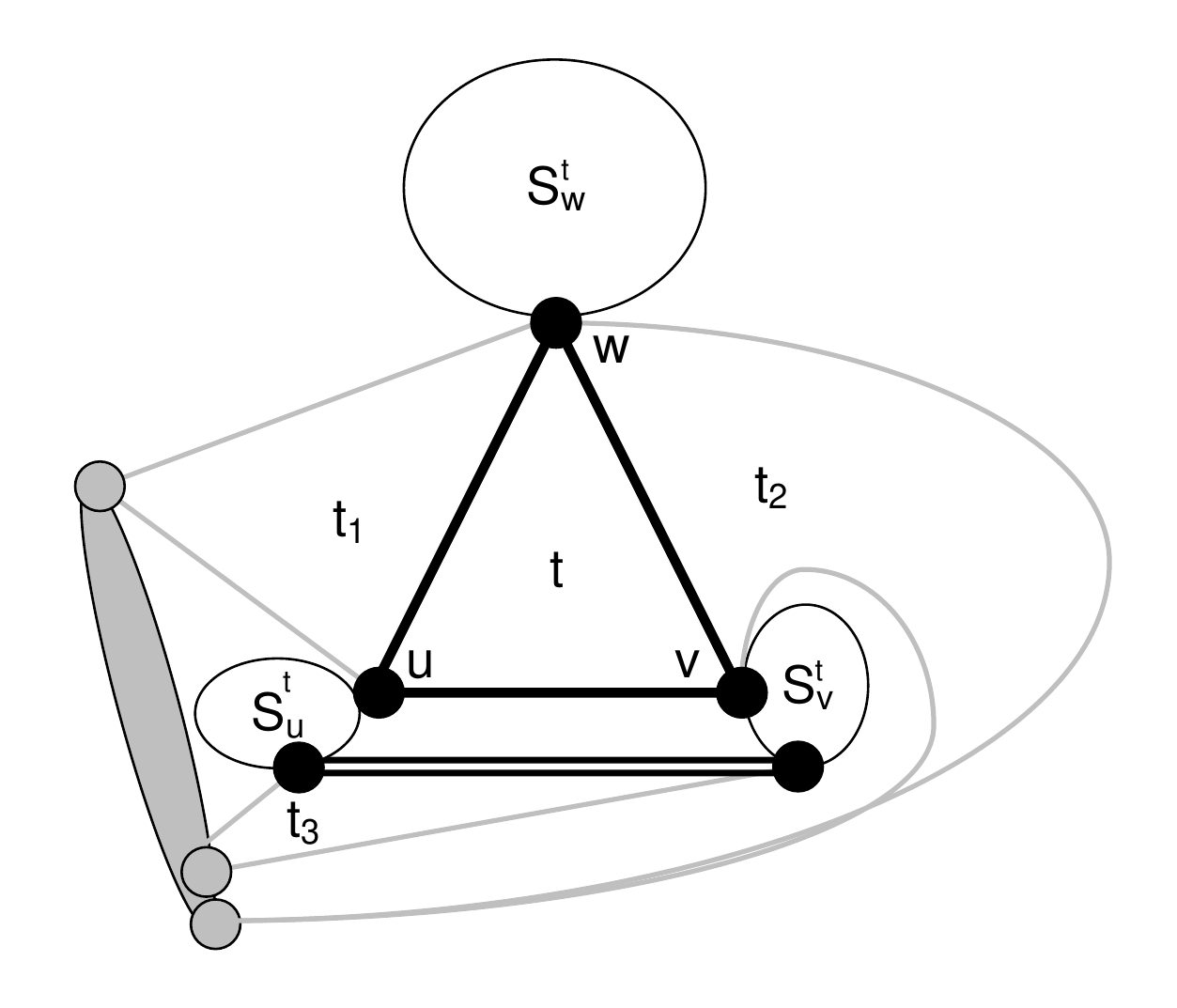}
    \caption{A type-$2$ light triangle $t$ for which the third property of Prop.~\ref{prop:structure-light} holds.}
    \label{fig:type-2-light}
\end{figure}

\begin{lemma}\label{lem:type-2_light}
Any cactus type-$2$ triangle $t$ in $G[S]$ is light.
\end{lemma}
\begin{proof}
Let $t$ be a type-$2$ triangle, such that each of the cactus edges $uw$ and $vw$ support a cross triangle $t_1, t_2$ respectively (see Figure~\ref{fig:type-2-light}).
By Lemma~\ref{lem:split_region}, $B_{uw}^t$ and $B_{vw}^t$ must be empty. By Lemma~\ref{lem:green-triangles-supporting} properties (\ref{itm:all-type-1-blue-edges}) and (\ref{itm:at-most-one-blue-edge-each}) there is at most one edge in $B_{uv}^t$ and if it exists it must be of type-$1$. Thus there are at most three cross triangles supported by $t$ and the edge in $B_{uv}^t$ and in addition at least two edges in $E(t)$ support a cross triangle, thus $t$ is a light triangle.
\end{proof}

\begin{figure}[H]
  \centering
  \begin{subfigure}[b]{0.3\textwidth}
        \includegraphics[width=\textwidth]{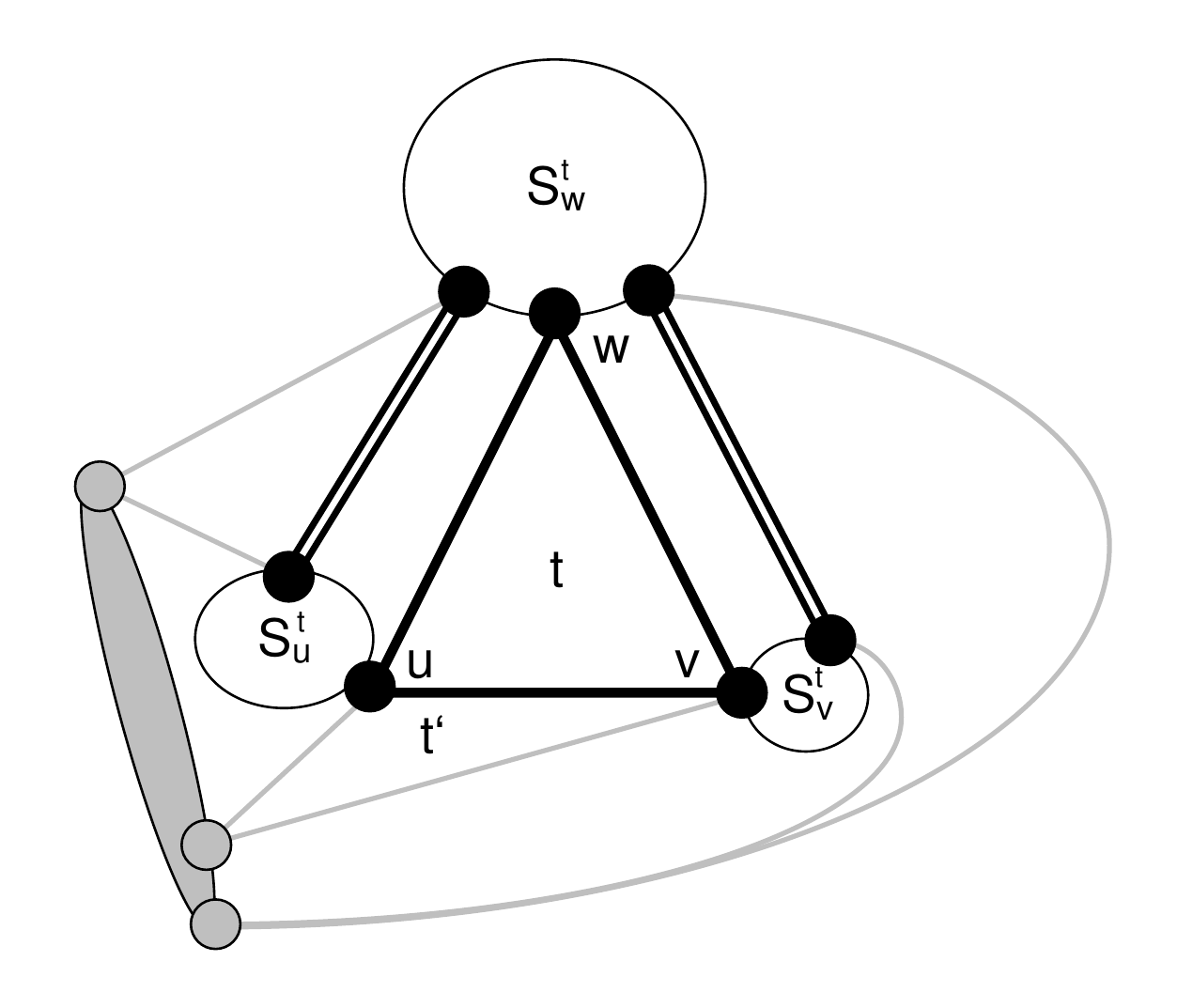}
        \caption{A type-$1$ triangle which is light for which the third property of prop.~\ref{prop:structure-light} holds.}
        \label{fig:type-1-light}
    \end{subfigure}\hspace{0.1\textwidth}%
      \begin{subfigure}[b]{0.3\textwidth}
        \includegraphics[width=\textwidth]{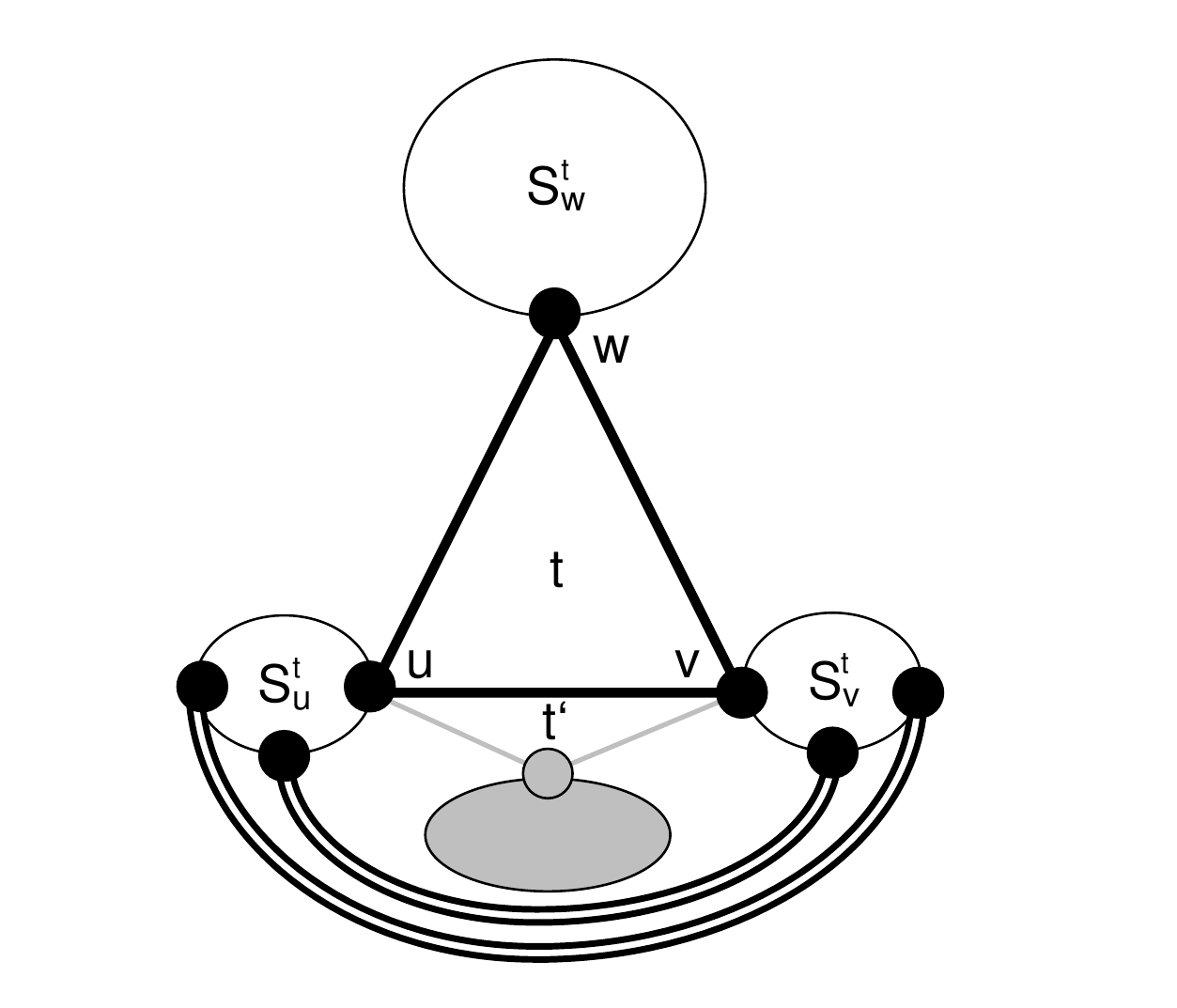}
        \caption{A type-$1$ triangle which could be either light or heavy, depending on the number of supported cross triangles supported by $B_{uv}^t$.}
        \label{fig:type-1-light-heavy}
    \end{subfigure}
    \caption{The classification of type-$1$ triangles into light and heavy.}
    \label{fig:type-1}
\end{figure}

\begin{lemma} \label{lem:type-1-heavy-light}
If $t$ is a heavy type-$1$ triangle, with $V(t)=\{u,v,w\}$, let $uv$ denote the edge in $E(t)$ that supports the cross triangle supported by $t$, then $B^t_{ww'} = \emptyset$ for all $ww' \in E(t) \setminus \{uv\}$ and the total number of cross triangles supported by edges in $B^t_{uv}$ is greater than or equal to two.
\end{lemma}
\begin{proof}
We first show that $B_{ww'}^t$ is empty for every $ww'\in E(t)\setminus uv$. Let $t'$ denote the cross triangle supported by $t$. Assume for contradiction, that there exists an edge $e$ in some $B_{ww'}^t$ for some edge $ww'\in E(t) \setminus uv$. As $t'$ and the cross triangle supported by $e$ fulfill the requirements of Lemma \ref{lem:green-triangles-supporting}, property (\ref{itm:at-most-one-gray-each}) implies that there are at most three cross triangles supported by edges in $E(t) \cup_{vv'\in E(t)} B_{vv'}^t$, which contradicts the definition of a heavy triangle.

As $B_{uw}^t$ and $B_{vw}^t$ are empty there must be at least two cross triangles in $G$ supported by edges in $B_{uv}^t$, as otherwise $t$ would be light.
\end{proof}

    \begin{figure}[H]
        \centering
        \begin{subfigure}[b]{0.3\textwidth}
            \includegraphics[width=\textwidth]{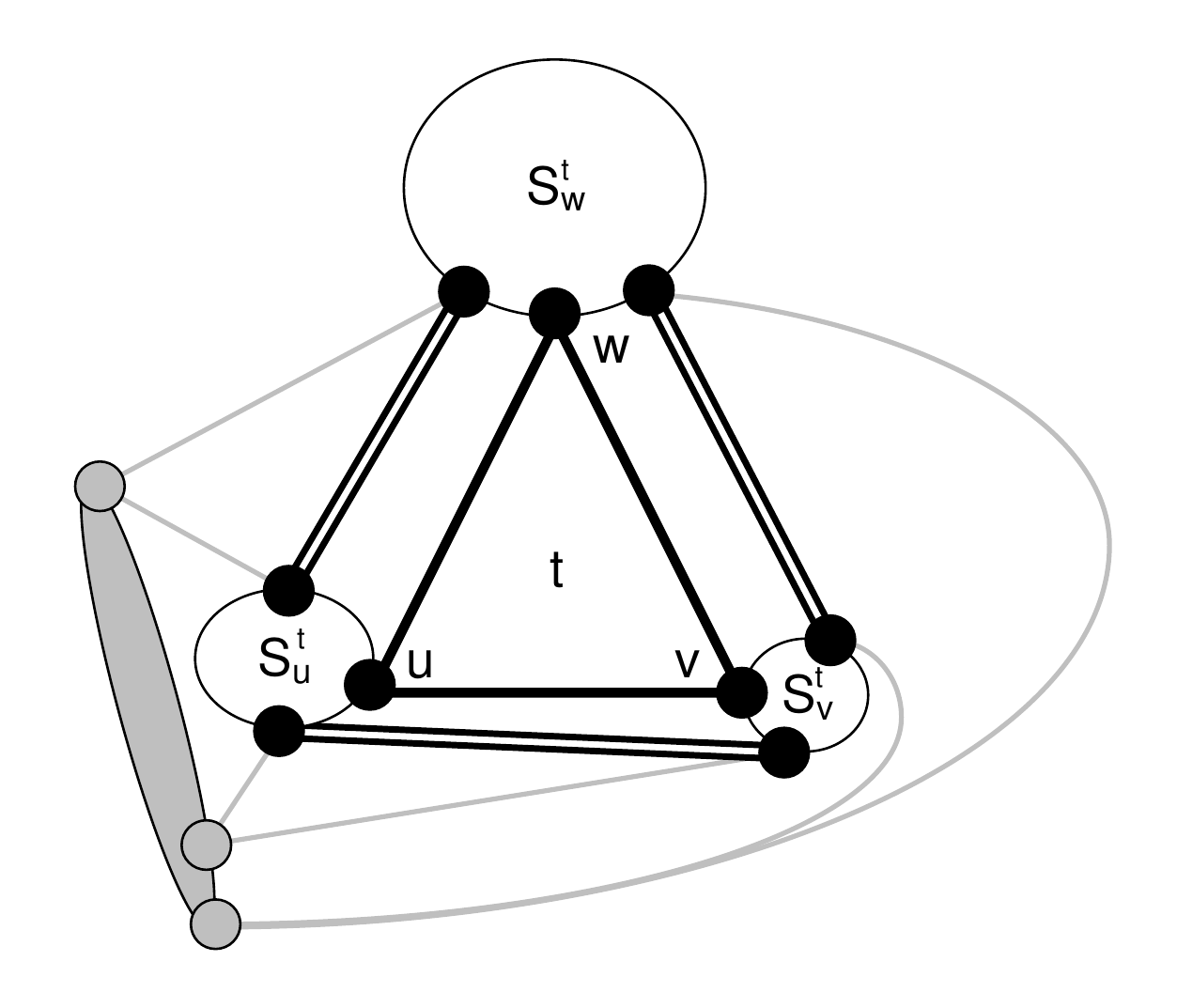}
            \caption{A type-$0$ light triangle for which the third property of Prop.~\ref{prop:structure-light} holds.}
            \label{fig:type-0-light}
        \end{subfigure}\hspace{0.1\textwidth}%
        \begin{subfigure}[b]{0.3\textwidth}
            \includegraphics[width=\textwidth]{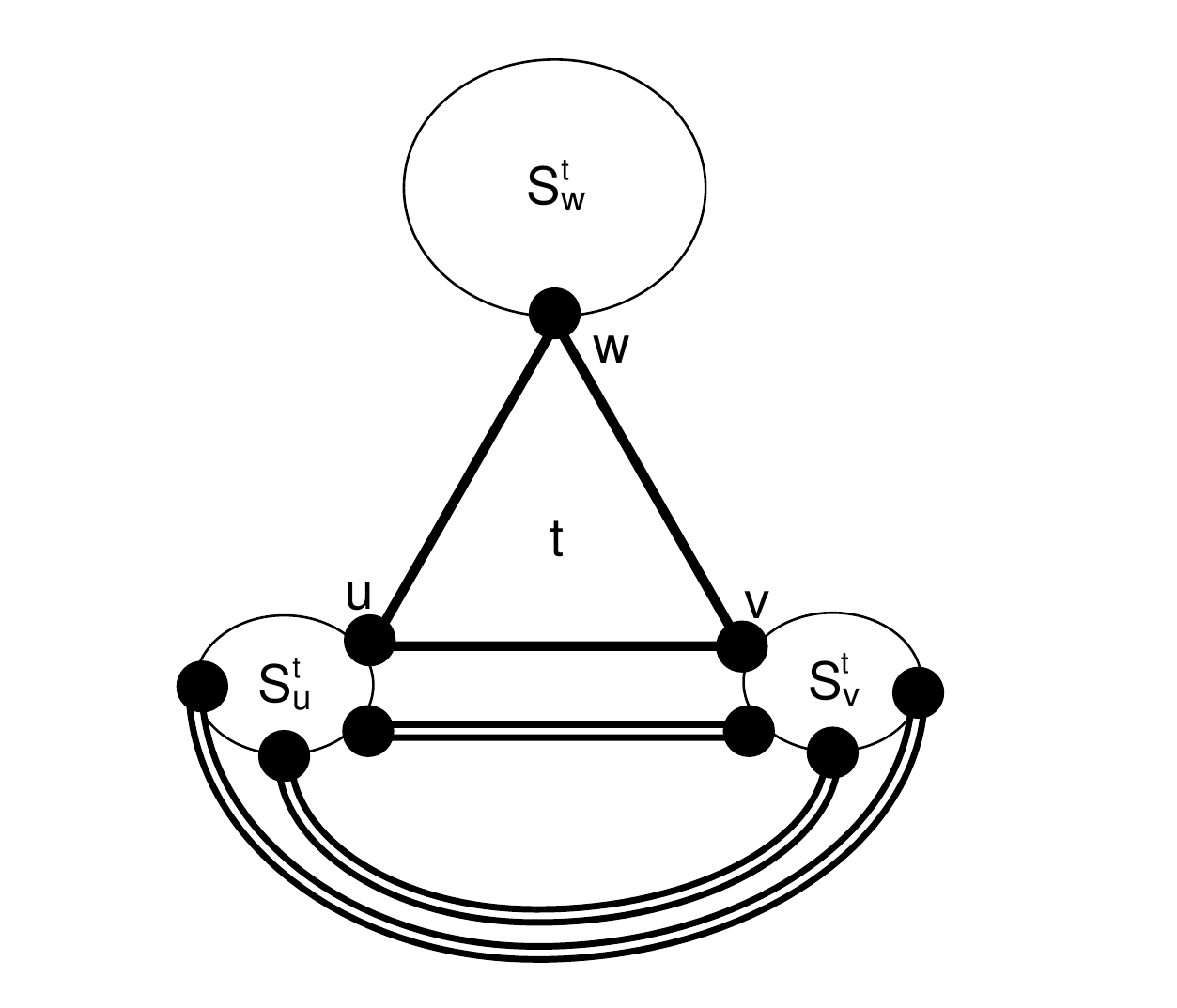}
            \caption{A type-$0$ triangle which could be either light or heavy, depending on the number of supported cross triangles supported by $B_{uv}^t$ .}
            \label{fig:type-0-light-heavy}
        \end{subfigure}
        \caption{The classification of type-$0$ triangles into light and heavy.}
        \label{fig:type-0}
    \end{figure}

\begin{lemma}\label{lem:type-0-heavy-light}
If $t$ is a heavy type-$0$ triangle, then there is an edge $uv \in E(t)$ such that $B^t_{ww'} = \emptyset$ for all $ww' \in E(t) \setminus \{uv\}$ and the total number of cross triangles supported by edges in $B^t_{uv}$ is greater than or equal to three.
\end{lemma}

\begin{proof}
We will first show that at most one of $B_{uu'}^t$ for $uu'\in E(t)$ can be non-empty. Assume for contradictions that there are two sets $B_{uv}^t$ and $B_{uw}^t$ which are non-empty. Then the cross triangles supported by the edges in these two sets fulfill the requirements of Lemma~\ref{lem:green-triangles-supporting}. Hence $|B_{uv}^t|,|B_{uw}^t|,|B_{vw}^t| \leq 1$ and the number of cross triangles supported by $E(t) \cup B_{uv}^t \cup B_{uw}^t \cup B_{vw}^t$ is at most three, contradicting the fact that $t$ is heavy.

Therefore we know that there is only one edge $uv\in E(t)$ such that $B_{uv}^t$ is non-empty. As $t$ is heavy $B_{uv}^t$ must contain edges which support at least three cross triangles as otherwise $t$ would be light.
\end{proof}

\subsection{Proof of Proposition \ref{prop:structure-light}}
\label{subsec:proof-prop-light}
In this section we will prove the properties stated in Proposition~\ref{prop:structure-light} about light triangles. Recall that for a light triangle the edges in $E(t) \cup_{uv\in E(t)} B_{uv}^t$ support at most three cross triangles.

\begin{lemma}
If $t$ is a light type-$0$ triangle with one edge $uv\in E(t)$ such that $B^t_{ww'} = \emptyset$ for all $ww' \in E(t) \setminus \{uv\}$, then the total number of cross triangles supported by edges in $B^t_{uv}$ is at most two. 
\end{lemma} 

\begin{proof}
This simply follows from the definition of heavy triangles. If there where more than two cross triangle supported by the edges in $B_{uv}^t$, then $t$ would be a heavy triangle.
\end{proof}

\begin{lemma}
If $t$ is a light type-$1$ triangle where $uv$ supports the cross triangle supported by $t$ and $B^t_{ww'} = \emptyset$ for all $ww' \in E(t) \setminus \{uv\}$, then the total number of cross triangles supported by edges in $B^t_{uv}$ is at most one. 
\end{lemma}

\begin{proof}
This simply follows from the definition of heavy triangles. If there was more than one cross triangle supported by the edges in $B_{uv}^t$, then $t$ would be a heavy triangle.
\end{proof}

\begin{lemma}
If $t$ is a light triangle where the edges in $\bigcup_{uv \in E(t)} B_{uv}^t${$\cup E(t)$} support either two or three cross triangles such that at least two different sets of edges $\{uv\} \cup B_{uv}^t$ for $uv \in E[t]$ support a cross triangle each, then each set of edges $\{uv\} \cup B_{uv}^t$ supports at most one cross triangle and all the supported cross triangles have the same landing component.
\end{lemma}

\begin{proof}
For any pair of cross triangles supported by edges in two different sets in $\{uv\} \cup B_{uv}^t$ for $uv \in E[t]$, Lemma \ref{lem:green-triangles-supporting} implies that both cross triangles must have the same landing component. Since there exists at least one pair of such triangles, by Lemma~\ref{lem:green-triangles-supporting} property~\ref{itm:at-most-one-gray-each}, the claim of this lemma follows.
\end{proof}

\subsection{Proof of Lemma~\ref{lem: derive factor 6}}

\label{subsec:proof-factor-6}
Below, we analyze the contribution from non outer-faces. 

\begin{table}[H]
\centering
\begin{tabular}{||c|c|} 
\hline
Coordinates & Value \\
\hline \hline 
$\vec{\chi}[1]$ & $|\fset[1, 0, 0] \setminus \fset_{fri}[1, 0, 0]|$  \\ \hline 
$\vec{\chi}[2]$ & $|\fset_{fri}[1,0,0]|$  \\ \hline
$\vec{\chi}[3]$ & $|\fset[1,0,\geq 1]|$  \\ \hline

$\vec{\chi}[4]$ & $|\fset[1,1,0] \setminus \fset_{fri}[1,1,0]|$  \\ \hline
$\vec{\chi}[5]$ & $|\fset_{fri}[1,1,0]|$  \\ \hline
$\vec{\chi}[6]$ & $|\fset[1,1,\geq 1]|$  \\ \hline

$\vec{\chi}[7]$ & $|\fset[2,0,0] \setminus \fset_{fri}[2,0,0]|$  \\ \hline
$\vec{\chi}[8]$ & $|\fset_{fri}[2,0,0]|$  \\ \hline
$\vec{\chi}[9]$ & $|\fset[2,0,\geq 1]|$  \\ \hline

$\vec{\chi}[10]$ & $|\fset[2,1, \bdot]|$  \\ \hline
$\vec{\chi}[11]$ & $|\fset[2,2,\bdot]|$  \\ \hline
$\vec{\chi}[12]$ & $|\fset[\geq 3, \bdot, \bdot]|$  \\ \hline

\hline
\end{tabular}
\label{tab:var}
\caption{Definition of characteristic vector of $\fset$ \label{tab:char}} 
\end{table}

This is simply an algebraic manipulation.
First, we write 
\[\overrightarrow{gain} \cdot \vec{\chi} \geq  4.5 ({\mathbb 1}^T \vec{\chi}) -(0, \textcolor{blue}{2},2, 0.5, \textcolor{blue}{2.5}, 2.5, 1.5, \textcolor{blue}{2.5}, 2.5, 2, 2.5, 3)^T \vec{\chi} \] 

We will gradually decompose the vector $$ (0, \textcolor{blue}{2},2, 0.5, \textcolor{blue}{2.5}, 2.5, 1.5, \textcolor{blue}{2.5}, 2.5, 2, 2.5, 3)^T \vec{\chi}$$ 
into several meaningful terms that we could upper bound. 
First, we focus on the coordinates that correspond to the $\eta_{fri}$ (highlighted in blue):  
$$(0, \textcolor{blue}{2},2, 0.5, \textcolor{blue}{2.5}, 2.5, 1.5, \textcolor{blue}{2.5}, 2.5, 2, 2.5, 3)^T \vec{\chi}= \fbox{$2\eta_{fri}$} +(0, \textcolor{blue}{0},2, 0.5, \textcolor{blue}{0.5}, 2.5, 1.5, \textcolor{blue}{0.5}, 2.5, 2, 2.5, 3)^T \vec{\chi} $$
where we simply applied the fact that $\eta_{fri}[1,0,0] + \eta_{fri}[1,1,0] + \eta_{fri}[2,0,0] = \eta_{fri}$. 
Next, we focus on the components of $\eta[2,\bdot, \bdot]$ and $\eta[3,\bdot, \bdot]$ (shown in red). 
$$(0, \textcolor{black}{0},2, 0.5, \textcolor{black}{0.5}, 2.5, \textcolor{red}{1.5}, \textcolor{red}{0.5}, \textcolor{red}{2.5}, \textcolor{red}{2}, \textcolor{red}{2.5}, \textcolor{red}{3})^T \vec{\chi} \leq  
\fbox{$1.5(p_1 + |\fset| -2)$} +(0, \textcolor{black}{0},2, 0.5, \textcolor{black}{0.5}, \textcolor{black}{2.5}, \textcolor{red}{0}, \textcolor{red}{-1}, \textcolor{red}{1}, \textcolor{red}{0.5}, \textcolor{red}{1}, \textcolor{red}{0})^T \vec{\chi}  $$
where we applied the upper bound from Lemma~\ref{lem:various-bounds-factor-6} (first bound). 
We further extract the ``components'' of $\eta[1,1,0]$, $\eta[2,1,\bdot]$ and $\eta[2,2,\bdot]$: 
\begin{align*}
(0, 0 ,2, \textcolor{green}{0.5}, \textcolor{green}{0.5}, \textcolor{green}{2.5}, 0,-1,1, \textcolor{green}{0.5}, \textcolor{green}{1}, 0)^T \vec{\chi} & = 0.5 (\eta[1,1,0] + \eta[2,1,\bdot] + 2 \eta[2,2,\bdot])  +  (0, 0 ,2, \textcolor{green}{0}, \textcolor{green}{0}, \textcolor{green}{2}, 0,-1,1, \textcolor{green}{0}, \textcolor{green}{0}, 0)^T \vec{\chi} \\ 
& \leq \fbox{$0.5 a_1$} +  (0, 0 ,2, \textcolor{green}{0}, \textcolor{green}{0}, \textcolor{green}{2}, 0,-1,1, \textcolor{green}{0}, \textcolor{green}{0}, 0)^T \vec{\chi} 
\end{align*}
the inequality was obtained by applying Lemma~\ref{lem:various-bounds-factor-6} (second bound).
Now, we extract the components of $\eta[1,1,\geq 1]$, $\eta[2,0,\geq 1]$ and $\eta[1,0, \geq 1]$ (the 3rd, 6th, and 9th coordinates respectively.) 
\begin{align*} 
(0, 0 ,\textcolor{blue}{2}, 0,0, \textcolor{blue}{2}, 0, -1, \textcolor{blue}{1}, 0, 0, 0)^T \vec{\chi}  &= 2(\eta[1,0,\geq 1] + \eta[1,1,\geq 1] + \eta[2,0,\geq 1]) + 
(0, 0 ,\textcolor{blue}{0}, 0,0, \textcolor{blue}{0}, 0, -1, \textcolor{blue}{-1}, 0, 0, 0)^T \vec{\chi} \\
 &\leq  \textcolor{blue}{2(p_0 - \eta_{fri})} + 
(0, 0 ,\textcolor{blue}{0}, 0,0, \textcolor{blue}{0}, 0, -1, \textcolor{blue}{-1}, 0, 0, 0)^T \vec{\chi}   \\ 
& \leq \fbox{$2(p_0 - \eta_{fri})$}
\end{align*}
here we applied the third bound of Lemma~\ref{lem:various-bounds-factor-6}, and the fact that all coordinates of vector $\vec{\chi}$ are non-negative.   
Finally, by summing over all terms in the boxes, we get the upper bound of 
\[2 \eta_{fri} + 1.5(p_1 + |\fset| -2) + 0.5 a_1 + 2(p_0 - \eta_{fri}) = 2p - 0.5 p_1 +2 a_1 +1.5 a_2 - 1.5\]
Now, since ${\mathbb 1}^T \vec{\chi} = a_1+ a_2$ and $gain(f_0) \geq g(S) - 1$, we have that 
\[\sum_{f \in \fset} gain(f) \geq 4.5(a_1+a_2) - (2p - 0.5 p_1 +2 a_1 +1.5a_2- 1.5) + g(S) - 1\]  
Hence, $-(\sum_{f \in \fset} gain(f))  \leq  - g(S) + (2p - 0.5p_1 - 3 a_2 - 2.5 a_1 - 0.5)$.  

We substitute the bound from the lemma into Eq~\ref{eq:gain-highlighted}, we would get: 
\[q \leq (4p+0.5p_1 +2.5a_1 + 3a_2)  - g(S) + (2p - 0.5p_1 - 3 a_2 - 2.5 a_1 - 0.5) \] 
This would give $q \leq 6p - g(S) - 0.5$ as desired. 
In Section~\ref{sec: factor 6, rules}, we first give the classification rules, and in Section~\ref{sec: factor 6, eq}, we describe the inequalities that we use.  

\newpage
\section{Tight Examples}
\label{sec:tight-examples}

\subsection{Tight Example of Factor $7$ for a $1$-swap Optimal Solution}
\label{subsec:tight-7}
Following is the family of tight examples for the factor $7$ analysis. In these examples, $p = 2^k + 1$ for some $k \in \mathbf{Z}^+$, $q = 7p - 4$,  $p_1 = p$, $p_0 = 0$, $a_1 = 0$, $a_2 = 2p - 3$, $|E(f_0)| = 4$, $g(S) = 2$, $\eta[1, 0] = 0.5 (p - 1)$, $\eta[\geq3, \bdot] = 1.5 (p - 1) - 1$ and all other $\eta[i, j] = 0$. Notice that each inequity which we use in the above proof is asymptotically tight for these example.

\begin{figure}[H]
  \centering
    \includegraphics[width=0.9\textwidth]{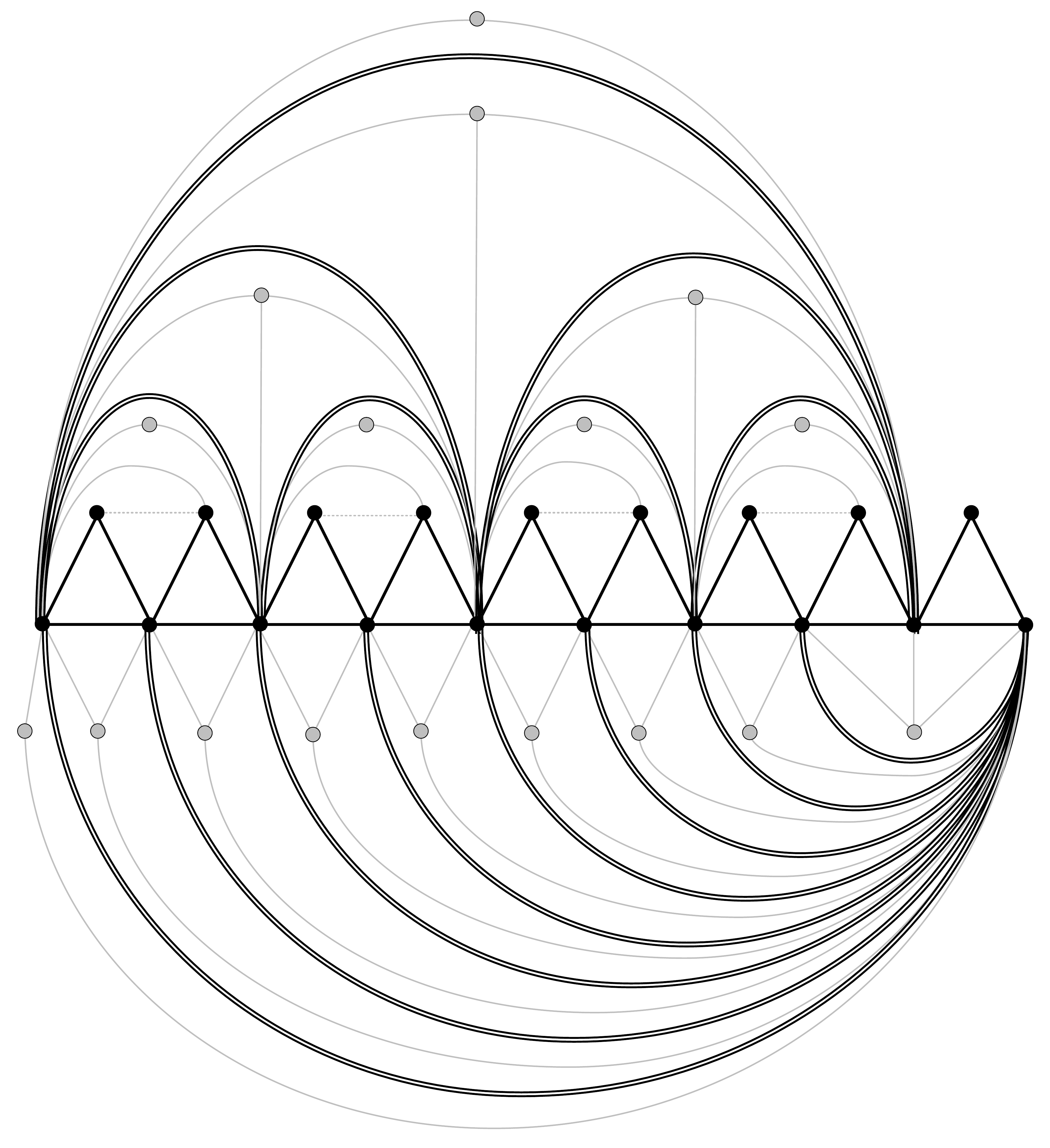}
    \caption{Tight example for factor $7$ analysis.}
    \label{fig:tight-factor-7}
\end{figure}

\subsection{Tight Example of Factor $6$ for a $2$-swap Optimal Solution}
\label{subsec:tight-6}
Following is the family of tight examples for the factor $6$ analysis. In these examples, $p = 2^k + 1$ for some $k \in \mathbf{Z}^+$, $q = 6p - 3$,  $p_1 = p$, $p_0 = 0$, $a_1 = 0$, $a_2 = 2p - 3$, $|E(f_0)| = 4$, $g(S) = 2$, $\eta[1, 0] = 0.5 (p - 1)$, $\eta[\geq3, \bdot] = 1.5 (p - 1) - 1$ and all other $\eta[i, j] = 0$. Notice that each inequity which we use in the above proof is asymptotically tight for these example.

\begin{figure}[H]
  \centering
    \includegraphics[width=0.9\textwidth]{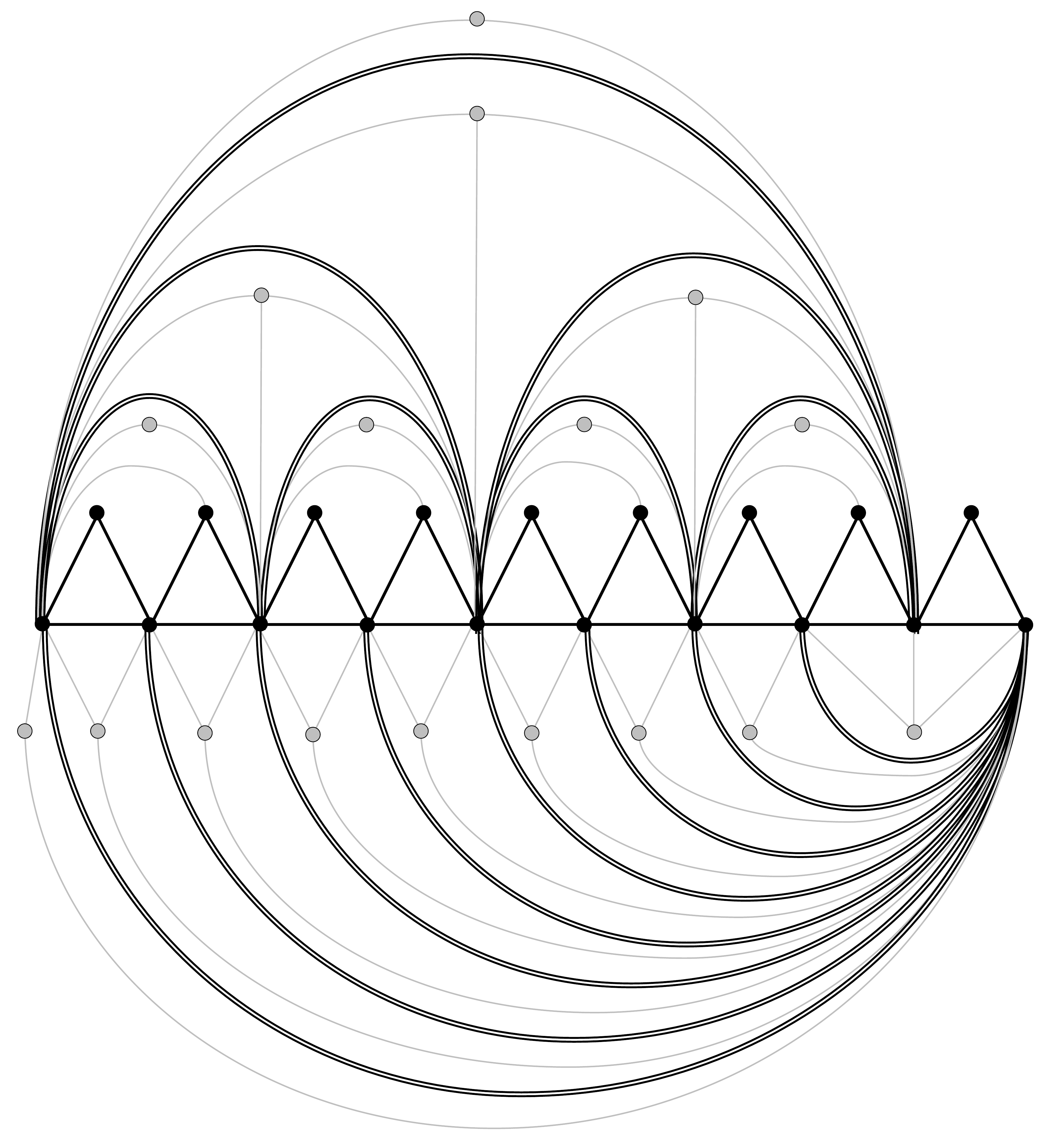}
    \caption{Tight example for factor $6$ analysis.}
    \label{fig:tight-factor-6}
\end{figure}

\section{Definition of Variables}
\begin{table}[H]
\centering
\begin{tabular}{||c|c|c|} 
\hline
Variable & Description \\
\hline \hline 
$G$ & initial plane graph with a fixed drawing \\ \hline 
$\cset$ & optimal triangular cactus which is a subgraph of $\cset$ \\ \hline
$S$ & a connected component of cactus $\cset$ \\ \hline
$p$ & \# triangular faces in the cactus $\cset[S]$ \\  \hline
$q$ & \# triangular faces of $G$ with at least two vertices in $G[S]$ \\  \hline
$a_0$ & type-$0$ edges supporting no triangles \\ \hline  
$a_1$ & type-$1$ edges supporting $1$ triangle \\  \hline 
$a_2$ & type-$2$ edges supporting $2$ triangles \\  \hline
$p_0$ & type-$0$ triangles supporting no triangles \\   \hline
$p_1$ & type-$1$ triangles supporting $1$ triangle \\   \hline
$p_2$ & type-$2$ triangles supporting $2$ triangles \\  \hline
$p_3$ & type-$3$ triangles supporting $3$ triangles \\  \hline
$H$ & subgraph of $G[S]$ after removing type-$0$ edges \\  \hline
$\cset_u^t$ &  The new components in $\cset[S] \setminus E(t)$, such that $v \in V(\cset_v^t)$ for every $v \in V(t)$\\ \hline
$S_v^t$ &  = $V(\cset_v^t)$ for every $v \in V(t)$\\ \hline
$B_{uv}^t$ & type-$1$ or type-$2$ edge with one end-point each in $S_u^t, S_v^t$ \\ \hline
$\fset$ & set of super-faces in $H$ excluding the $p$ cactus triangular faces of $\cset[S]$ \\  \hline
$a_1^{occ}(f)$ & \# of supporting side of type-$1$ edges in $f \in \fset$ \\ \hline
$a_1^{free}(f)$ & \# of non-supporting side of type-$1$ edges in $f \in \fset$ \\ \hline
$a_1(f)$ &$= a_1^{free}(f) + a_1^{occ}(f)$ \\ \hline
$a_2(f)$ & \# of type-$2$ edges in $f \in \fset$ \\ \hline
$p^{base}_0(f)$ & \# of base edges of type-$0$ triangles in $f \in \fset$ \\ \hline
$p^{free}_0(f)$ & \# of free pair of edges of type-$0$ triangles in $f \in \fset$ \\ \hline
$p^{base}_1(f)$ & \# of base edges of type-$1$ triangles in $f \in \fset$ \\ \hline
$p^{free}_1(f)$ & \# of free pair of edges of type-$1$ triangles in $f \in \fset$ \\ \hline
$p^{free}(f)$ & $=p^{free}_0(f) + p^{free}_1(f)$ \\ \hline
$p^{base}(f)$ & $=p^{base}_0(f) + p^{base}_1(f)$ \\ \hline
$|E(f)|$ &$= 2 p^{free}(f) + p^{base}(f) + a_2(f) + a_1(f)$ (length of face $f \in \fset$) \\ \hline
$|Occ(f)|$ &$=  p_1^{base}(f) + a_2(f) + a_1^{occ}(f)$ ({\em occupied} length of face $f \in \fset$) \\ \hline
$|Free(f)|$ &$= 2 p^{free}(f) + p_0^{base}(f) + a^{free}_1(f)$ ({\em free} length of face $f \in \fset$) \\ \hline
$\mu(f)$ & $= |Free(f)| + |Occ(f)|/2$ \\ \hline
$survive(f)$ & \# surviving faces in some $f \in \fset$ \\ \hline
$\ell(S)$ & length of outer-super-face of graph $G[S]$ \\ \hline
$o(S)$ & \# of $a_2$, $a^{occ}_1$ and $p^{base}_1$ side edges in outer-face of graph $G[S]$ \\ \hline
$\phi(S)$ & $= \ell(S) - o(S)$ \\ \hline
$f_0$ & Outer-face of graph $H[S]$\\ \hline
$\fset[i, j, k]$ & set of super-faces of type-$[i, j, k]$ such that for any $f \in \fset[i, j, k]$,\\ & $p_1^{base}(f) + a_2(f) + a_1(f) = i$, $a^{free}_1(f) = j$ and $p^{base}_0(f) = k$.\\ \hline
$\eta[i, j, k]$ & $=|\fset[i, j, k]|$\\ \hline
$\fset_{fri}[i, j, k]$ & $=|\fset[i, j, k]|$\\ \hline
$\eta_{fri}[i, j, k]$ & $=|\fset_{fri}[i, j, k]|$\\ \hline
$\eta_{fri}$ & $=\eta_{fri}[1, 0, 0] + \eta_{fri}[1, 1, 0] + \eta_{fri}[2, 0, 0]$\\ \hline
\hline
\end{tabular}
\caption{Definition of variables used in the proof}
\label{tab:vars}
\end{table}

\end{document}